\documentclass[12pt,draftclsnofoot,onecolumn]{IEEEtran} 	% <---- 	1-column format

\usepackage{amsmath}
\usepackage{amsthm}
\usepackage{amsfonts}
\usepackage{dsfont}
\usepackage{amssymb}
\usepackage{stackrel}
\usepackage{color}
\usepackage{cite}
\usepackage{hyperref,graphicx,enumitem}
\usepackage{multirow,longtable,bigstrut,caption}

\newcommand{\ud}{\,\mathrm{d}}
\newcommand{\ue}{\mathrm{e}}
\newcommand{\uj}{\mathrm{j}}
\newcommand{\bR}{\mathbb{R}}
\newcommand{\bN}{\mathbb{N}}
\newcommand{\bZ}{\mathbb{Z}}
\renewcommand{\H}[1]{\mathrm{H}\left( {#1} \right)}
\renewcommand{\Pr}[1]{\mathrm{Pr}\left\lbrace{#1}\right\rbrace}
\newcommand{\h}[1]{\mathrm{h}\left( {#1} \right)}
\newcommand{\E}[1]{\mathbb{E}\left[{#1}\right]}

\newcommand{\quant}[2]{{\left[ {#1} \right]}_{#2}}
\newcommand{\EQnum}{\addtocounter{equation}{1}\tag{\theequation}}
\newcommand{\InProd}[2]{\left \langle {#1} , {#2} \right \rangle}
\newcommand{\sgn}[1]{\mathrm{sgn}\left({#1}\right)}

\newtheorem{definition}{Definition}
\newtheorem{theorem}{Theorem}
\newtheorem{lemma}{Lemma}
\newtheorem{corollary}{Corollary}
\newtheorem{remark}{Remark}

\begin{document}
\makeatletter

\title{How Compressible are Innovation Processes?}

% Authors
\author{Hamid~Ghourchian,
		Arash~Amini,~\IEEEmembership{Senior Member,~IEEE}, and
		Amin~Gohari,~\IEEEmembership{Senior Member,~IEEE}
	
	%\thanks{Manuscript received January 7 2017.}

	\thanks{The authors are with the department of Electrical Engineering and Advanced Communication Research Institute (ACRI), Sharif University of Technology, Tehran, Iran.
	(e-mails: h\_ghourchian@ee.sharif.edu,~\{aamini,~aminzadeh\}@sharif.ir)}
}
\maketitle

%%%%%
% Abstract
%%%%%
\begin{abstract}
%\boldmath
The sparsity and compressibility of finite-dimensional signals are of great interest in fields such as compressed sensing.
The notion of compressibility is also extended to infinite sequences of i.i.d. or ergodic random variables based on the observed error in their nonlinear $k$-term approximation.
In this work, we use the entropy measure to study the compressibility of continuous-domain innovation processes (alternatively known as white noise).
Specifically, we define such a measure as the entropy limit of the doubly quantized (time and amplitude) process.
This provides a tool to compare the compressibility of various innovation processes. It also allows us to identify an analogue of the concept of ``entropy dimension" which was originally defined by R\'enyi for random variables.
Particular attention is given to stable and impulsive Poisson innovation processes. Here, our results recognize Poisson innovations as the more compressible ones with an entropy measure far below that of stable innovations.
While this result departs from the previous knowledge regarding the compressibility of fat-tailed distributions, our entropy measure ranks stable innovations according to their tail decay.
\end{abstract}

% Keywords

%%%%%%%
% Introduction
%%%%%%%
\section{Introduction}
\if@twocolumn
	\IEEEPARstart{T}{he}
\else
	The
\fi
 compressible signal models have been extensively used to represent or approximate various types of data such as audio, image and video signals. 
The concept of compressibility has been separately studied in information theory and signal processing societies. In information theory, this concept is usually studied via the well-known entropy measure and its variants. For instance, the notion of \emph{entropy dimension} was introduced in \cite{Renyi59} for continuous random variables based on the concept of differential entropy. The entropy dimension was later studied for discrete-domain random processes in \cite{Verdu10} and \cite{Jalali2017} in the context of compressed sensing.
In signal processing, a signal is intuitively called compressible if in its representation using a known dictionary only a few atoms contribute significantly and the rest amount to negligible contribution. 
Sparse signals are among the special cases for which the mentioned representation consists of a few non-zero (instead of insignificant) contributions.
Compressible signals in general, and sparse signals in particular, are of fundamental importance in fields such as compressed sensing \cite{Donoho2006,Candes2006}, dimensionality reduction \cite{Baraniuk2010}, and nonlinear approximation theory \cite{DeVore1998}. 

In this work, we consider the compressibility of continuous-domain innovation processes which were originally studied in the context of signal processing \cite{Unser2014_P1,Unser2014_P2}. However, we try to apply information theoretic tools to measure the compressibility. To cover the existing literature in both parts, we begin by the signal processing perspective of compressibility. 

Traditionally, compressible signals are defined as infinite sequences within the Besov spaces \cite{DeVore1998}, where the decay rate of the $k$-term approximation error could be efficiently controlled.
A more recent deterministic approach towards modeling compressibility is via weak-$\ell_p$ spaces \cite{Candes2006, Cohen2009}.
The latter approach is useful in compressed sensing, where $\ell_1$ (or $\ell_p$) regularization techniques are used.

The study of stochastic models for compressibility started with identifying compressible priors.
For this purpose, independent and identically distributed (i.i.d.) sequences of random variables with a given probability law are examined. Cevher in \cite{Cevher2009} defined the compressibility criterion based on the decay rate of the mean values of order statistics.
A more precise definition in \cite{Amini11} revealed the connection between compressibility and the decay rate of the tail probability.
In particular, heavy-tailed priors with infinite $p$-order moments were identified as $\ell_p$-compressible probability laws.
It was later shown in \cite{Silva2012} that this sufficient condition is indeed, necessary as well.
A similar identification of heavy-tailed priors (with infinite variance) was obtained in \cite{Gribonval2012} with a different definition of compressibility\footnote{Based on the definition in \cite{Gribonval2012}, a prior is called compressible if the compressed measurements (by applying a Gaussian ensemble) of a high-dimensional vector of i.i.d. values following this law could be better recovered using $\ell_1$ regularization techniques than the classical $\ell_2$ minimization approaches.}.

The first non-i.i.d. result appeared in \cite{Silva15}.
By extending the techniques used in \cite{Silva2012}, and based on the notion of $\ell_p$-compressibility of \cite{Amini11}, it is shown in \cite{Silva15} that discrete-domain stationary and ergodic processes are $\ell_p$-compressible if and only if the invariant distribution of the process is an $\ell_p$-compressible prior.

The recent framework of sparse stochastic processes introduced in \cite{Unser2014_P1,Unser2014_P2,Unser14} extends the discrete-domain models to continuous-domain.
In practice, most of the compressible discrete-domain signals arise from discretized versions of continuous-domain physical phenomena.
Thus, it might be beneficial to have continuous-domain models that result in compressible/sparse discrete-domain models for a general class of sampling strategies.
Indeed, this goal is achieved in \cite{Unser2014_P1,Unser2014_P2} by considering non-Gaussian stochastic processes.
The building block of these models are the innovation processes (widely known as white noise) that mimic i.i.d. sequences in continuous-domain.
Unlike sequences, the probability laws of innovation processes are bound to a specific family known as \textit{infinitely divisible} that includes $\alpha$-stable distributions ($\alpha=2$ corresponds to Gaussians).

The discretization of innovation processes are known to form stationary and ergodic sequences of random variables with infinitely divisible distributions.
As the tail probability of all non-Gaussian infinitely divisible laws are slower than Gaussians \cite{Amini2012,Amini14}, they exhibit more compressible behavior than Gaussians according to \cite{Amini11,Silva15}.

In this paper, we investigate the compressibility of continuous-domain stochastic processes using the quantization entropy.
As a starting point, we restrict our attention in the present work to innovation processes as the building blocks of more general stochastic processes.
We postpone the evaluation of the quantization entropy for more general processes to future works.

In information theory, entropy naturally arises as the proper measure of compression for Shannon's lossless source coding problem.
It also finds a geometrical interpretation as the volume of the typical sets.
As the definition of entropy ignores the amplitude distribution of the involved random variables and only takes into account the distribution (or concentration) of the probability measure, it provides a fundamentally different perspective of compressibility compared to the previously studied $k$-term approximation. 
More precisely, the entropy reveals a universal compressibility measure that is not limited to a specific measurement technique, while the $k$-term approximation is tightly linked with the linear sampling strategy of compressed sensing. Hence, it is not surprising that our results based on the quantization entropy show that impulsive Poisson innovation processes are by far more compressible than heavy-tailed $\alpha$-stable innovation processes; the previous studies on their $k$-term approximation sort them in the opposite order when the jump distribution in the Poisson innovation is not heavy-tailed. It is interesting to mention that the same ordering of impulsive Poisson and heavy-tailed α-stable innovation processes is observed in \cite{Fageot2017}.

The two main challenges that are addressed in this paper are
1) defining a quantization entropy for continuous-domain innovation processes that translates into operational lossless source coding,
and 2) evaluating such a measure for particular instances to allow for their comparison.
We recall that the \emph{differential entropy} of a random variable $X$ with continuous range is defined by finely quantizing $X$ with resolution $1/m$, followed by canceling a diverging term $\log(m)$ from the discrete entropy of its quantized version.
Obviously, we shall expect more elaborate diverging terms when dealing with continuous-domain processes.
More specifically, after appropriate quantization in time and amplitude with resolutions $1/n$ and $1/m$ respectively, we propose the one of the following two expressions to cancel out the diverging terms:
\begin{align} \label{eqn:Funda1}
	\frac{\mathcal{H}_{m,n}(X)}{\kappa(n)} - \log(m) -\zeta(n),
\end{align}
or
\begin{align} \label{eqn:Funda1--2}
	\mathcal{H}_{m,n}(X) - \zeta(n),
\end{align}
where $\mathcal{H}_{m,n}(X)$ is the discrete entropy of the time/amplitude quantized process, and  $\kappa(\cdot)$ and $\zeta(\cdot)$ are  univariate functions. 
We prove that depending on the white noise process, \eqref{eqn:Funda1} or \eqref{eqn:Funda1--2} give the correct way to cancel out the diverging terms for a wide class of white L\'evy noises with a suitable choice of $\kappa(n)$ and $\zeta(n)$. 
We may view $\kappa(n)/n$ as the analogue of entropy dimension for a white noise process.
A general expression for $\kappa(n)$ is given. However, while we prove existence of a function $\zeta(n)$, we are able to provide its explicit expression  only for special cases of stable and Poisson white noise processes.
While the term $\log(m)$ is reminiscent of the amplitude quantization effect, functions $\kappa(\cdot)$ and $\zeta(\cdot)$ quantify the compressibility of a given L\'evy process: the higher the  growth rate of $\kappa(n)$, the less compressible the process.  If two processes have the same growth rate of $\kappa(n)$, then, the $\zeta(n)$ with the smaller growth rate is the more compressible process.

Finally, $\epsilon$-metric entropy of \cite{Lorentz66} and $\epsilon$-entropy of \cite{Ponser73} could be considered as other alternatives for quantifying compressibility of stochastic processes. The interested reader may refer to \cite{Ihara93} for a general discussion of defining relative entropy for stochastic processes.

% paper organization
The organization of the paper is as follows. 
We begin by reviewing the preliminaries, including some of the basic definitions and results regarding differential entropy and white L\'evy noises in Section \ref{sec.Prelim}.
Next, we present our main contributions in Section \ref{sec.MainRes} wherein we study the quantization entropy of a wide class of white L\'evy noise processes. Furthermore, special attention is given to the stable and impulsive Poisson innovation processes.
To facilitate reading of the paper, we have separated the results from their proofs.
Some of the key lemmas for obtaining the final claims are explained in Section \ref{sec.UsefulLem}, while the main body of proofs are postponed to Section \ref{sec.Proofs}.

%%%%%%%
% Preliminaries
%%%%%%%
\section{Preliminaries} \label{sec.Prelim}
The goal of this paper is to define an quantization entropy for certain random processes.
Hence, we first review the concept of entropy for random variables.
For this purpose, we provide the definition of entropy for three main types of probability distributions.
This is followed by the definition of innovation processes (white L\'evy noises) and in particular, the stable and Poisson white noise processes.

All the logarithms in this paper are in base $\ue$.
In Table \ref{table:summary}, we summarize the notation used in this paper.

\begin{table*}
\begin{center}
\begin{small}
\begin{tabular}{|c|c|}
	\hline
	\textbf{Symbol} & \textbf{Definition} \\
	\hline\hline
	Sets & Caligraphic letters like $\mathcal{A},\mathcal{C},\mathcal{D},\ldots$ \\
	\hline
	Real and natural numbers & $\bR$, $\bN$ \\
	\hline
	Borel sets in $\bR$ & $\mathcal{B}(\bR)$ or just $\mathcal{B}$ \\
	\hline
	Random variables & Capital letters like: $A,X,Y,Z,\ldots$ \\
	\hline
	Probability density function (pdf) of (continous) $X$&  $p_X$ or $q_X$  (lower-case letter $p$)\\
	\hline
	Probability mass function (pmf) of  (discrete) $X$ & $P_X$  (upper-case letter $P$) \\
	\hline
	Cumulative distribution function (cdf) of $X$ & $F_X$ \\
	\hline
	$X_0$ for a given white noise $X(t)$ & $\int_0^1{X(t) \ud t}$; a random variable \\
	\hline
		$\frac{1}{n}$  & step-size for time quantization (time sampling) \\
	\hline
		$\frac{1}{m}$ & step-size for amplitude quantization \\
	\hline
	$X_i^{(n)}$ for a given white noise $X(t)$ & $\int_{\frac{i-1}{n}}^{\frac{i}{n}}{X(t) \ud t}$; a random variable \\
	\hline
	$X_1, X_2, \cdots $  & A sequence of random variables\\
&Distinguished from $X_i^{(n)}$ by lack of $(n)$ superscript.\\
	\hline
	$\quant{x}{m}$ & Quantization of the value of $x$ with step size $1/m$;\vspace{0.0cm}\\
\vspace{0.0cm}
	$\quant{X_i^{(n)}}{m}$ & Quantization of $X_i^{(n)}$ with step size $1/m$ 
 \\
	\hline
	Discrete random variable associated to  & $X_D$ \\
a discrete-continuous random variable $X$& \\
	\hline
	Continuous random variable associated to  & $X_c$ \\
a discrete-continuous random variable $X$ &\\
	\hline
	Entropy & $\mathrm{H}$ \\
	\hline
	Differential entropy & $\mathrm{h}$ \\
	\hline
$\mathcal{H}_{m,n}(X)$ & Quantization entropy associated to a
\\&  white noise $X(t)$ in \eqref{eqn:Hmn-ent-def} and \eqref{def.CompCrit} \\
\hline
	Characteristic function & $\widehat{p}$\\
	\hline
\end{tabular}
\end{small}
\end{center}
\caption{Summary of the notation.}
\label{table:summary}
\end{table*}

%%%%%%%%%%%%%
% Types of Random Variables
\subsection{Types of Random variables} \label{subsec:DistributionTypes}
The main types of random variables considered in this paper are discrete, continuous, and discrete-continuous, which are defined below.

% Definition: Continuous Distribution
\begin{definition}[Continuous Random Variables] \label{def.ACRV} \cite{Nair06}
	Let $\mathcal{B}$ be the Borel $\sigma$-field of $\mathbb{R}$ and let $X$ be a real-valued random variable with distribution (cdf) $F(x)$ that is measurable with respect to $\mathcal{B}$. 
	We call $X$ a \emph{continuous} random variable if its probability measure $\mu$, induced on $(\bR,\mathcal{B})$, is absolutely continuous with respect to the Lebesgue measure for $\mathcal{B}$ (i.e., $\mu(\mathcal{A})=0$ for all $\mathcal{A}\in\mathcal{B}$ with zero Lebesgue measure).
	We denote the set of all absolutely continuous distributions by $\mathcal{AC}$.
\end{definition}

It is a well-known result that $X$ being a continuous random variable is equivalent to the fact that the cdf $F(x)$ is an absolutely continuous function.
% Radon-Nikodym theorem
The Radon-Nikodym theorem implies that for each $X\sim F\in\mathcal{AC}$ there exists a $\mathcal{B}$-measurable function $p:\bR\to [0,\infty)$, such that for all $\mathcal{A}\in\mathcal{B}$ we have that
\begin{equation*}
	\Pr{X\in \mathcal{A}} = \int_\mathcal{A}{p(x) \ud x}.
\end{equation*}
The function $p$ is called the probability density function (pdf) of $X$.
The property $F\in\mathcal{AC}$ is alternatively written as $p\in\mathcal{AC}$. \cite[p. 21]{Ihara93}

% Definition: Discrete Distribution
\begin{definition}[Discrete Random Variable] \cite{Nair06} \label{def.ADRV}
	A random variable $X$ is called discrete if it takes values in a countable alphabet set $\mathcal{X}\subset\mathbb{R}$.
\end{definition}

% Definition: Discrete-Continuous Random Variable
\begin{definition}[Discrete-Continuous Random Variable] \label{def.DCRV} \cite{Nair06}
	A random variable $X$ is called discrete-continuous with parameters $\left(p_c,P_D,\Pr{X\in\mathcal{D}}\right)$ if there exists a countable set $\mathcal{D}$, a discrete probability mass function $P_D$, whose support is $\mathcal{D}$, and a pdf $p_c\in \mathcal{AC}$ such that
	\begin{equation*}	
		0<\Pr{X \in \mathcal{D}}<1,
	\end{equation*}
	as well as for every Borel set $\mathcal{A}$ in $\mathbb{R}$ we have that
	\if@twocolumn
   		\begin{align*}
		& \Pr{X \in \mathcal{A} | X \notin \mathcal{D}}
		= \int_\mathcal{A}{p_c(x) \ud x},\\
		& \Pr{X \in \mathcal{A} | X \in\mathcal{D}}
		= \sum_{x\in\mathcal{D}\cap\mathcal{A}}{P_D[x]}.		
	\end{align*}
	\else
   		\begin{align*}
		 &\Pr{X \in \mathcal{A} | X \notin \mathcal{D}}
		= \int_\mathcal{A}{p_c(x) \ud x},
		\qquad \Pr{X \in \mathcal{A} | X \in\mathcal{D}}
		= \sum_{x\in\mathcal{D}\cap\mathcal{A}}{P_D[x]}.		
		\end{align*}
	\fi	
	It is clear that we can write the pdf of a discrete-continuous random variable $X$, $p_X$, as follows:
	\begin{align*}
		p_X(x) =& \Pr{X \in\mathcal{D}} P_d(x)
		+\left( 1- \Pr{X \in \mathcal{D}} \right) p_c(x),
	\end{align*}
	where $p_c \in\mathcal{AC}$, and $P_d$ is the probability mass function of the discrete
part, which is a convex combination of Dirac\rq s delta functions.
\end{definition}

% Notations
In this paper, the probability mass function of discrete random variables is denoted by capital letters like $P$ and $Q$, while the probability density function of continuous or discrete-continuous random variables is denoted by lowercase letters like $p$ and $q$.

%%%%%%%%%%%
% Definition of Entropy
\subsection{Definition of Entropy} \label{subsec:EntropyTypes}
We first define the entropy and differential entropy for discrete and continuous random variables, respectively.
Next, we define entropy dimension for discrete-continuous random variables via amplitude quantization.

% Definition: Entropy and Differential Entropy
\begin{definition}[Entropy and Differential Entropy] \label{def.Ent} \cite[Chapter 2]{Cover06}
	We define entropy $\H{X}$, or $\H{P_X}$ for a discrete random variable $X$ with probability mass function (pmf) $P_X[x]$ as
	\begin{equation*}
		\H{P_X}=\H{X} := \sum_x{P_X[x]\log\tfrac{1}{P_X[x]}},
	\end{equation*}
	if the summation converges.
	For a continuous random variable $X$ with pdf $p_X(x) \in \mathcal{AC}$, we define differential entropy $\h{X}$, or $\h{p_X}$ as
	\begin{equation*}
		\h{p_X}=\h{X}
		:= \int_{\bR}{p_X(x)\log \tfrac{1}{p_X(x)}\ud x},
	\end{equation*}
	if
	\begin{equation*}
		\int_{\bR}{p_X(x)\left|\log\tfrac{1}{p_X(x)}\right|\ud x}
		<\infty.
	\end{equation*}
	Similarly, for a discrete or continuous random vector $\bf{X}$, the entropy or the differential entropy is defined as
	\begin{align*}
		\H{\bf X} = \E{\log\frac{1}{P_{\bf X}[\bf X]}}, ~~~
		\text{or}
		~~~
		\h{\bf X} = \E{\log\frac{1}{p_{\bf X}(\bf X)}},
	\end{align*}
	where $P_{\bf X}[\bf x]$  ($p_{\bf X}(\bf x)$) is the pmf (pdf) of the random vector $\bf X$, respectively.
\end{definition}

In brief, we say that the (differential) entropy is well-defined if the corresponding (integral) summation is convergent to a finite value.

Next, we identify a class of absolutely continuous probability distributions, and show that differential entropy is uniformly convergent over this space under the total variation distance metric.

% Definition: (a,m,v)-AC
\begin{definition} {\cite{OurFirstPaper}} \label{def.(a,m,v)-AC}
	Given $\alpha,\ell,v\in(0,\infty)$, we define $(\alpha,\ell,v)\text{--}\mathcal{AC}$ to be the class of all $p\in\mathcal{AC}$ such that the corresponding density function $p:\bR\mapsto[0,\infty)$ satisfies
	\begin{align*}
		&\int_\bR{|x|^\alpha\, p(x)\ud x}\leq v, \\
		& {\rm ess}\sup_{x\in\bR} ~p(x) \leq \ell.
	\end{align*}
\end{definition}

% Theorem: Entropy Convergent for Finite Variances
\begin{theorem} {\cite{OurFirstPaper}} \label{thm.EntConv}
	The differential entropy of any distribution in $(\alpha,\ell,v)\text{--}\mathcal{AC}$ is well-defined, and for all $p_X,p_Y\in(\alpha,\ell,v)\text{--}\mathcal{AC}$ satisfying $\left\lVert p_X - p_Y \right\rVert_1\leq m$, we have that
	\begin{equation*}
		\left|\h{p_X}-\h{p_Y}\right|
		\leq c_1 D_{X,Y} + c_2 D_{X,Y} \log\tfrac{1}{D_{X,Y}},
	\end{equation*}
	where
	\begin{equation*}
		D_{X,Y}
		=\left\lVert p_X-p_Y \right\rVert_1
		:=\int_\bR{\left|p_X(x)-p_Y(x)\right|\ud x},
	\end{equation*}
	and
	\if@twocolumn
		\begin{align*}
		c_1 =&\tfrac{1}{\alpha}\left|\log(2\alpha v)\right| + |\log(m \ue)|
		+\log\tfrac{\ue}{2}\nonumber \\
		&+\log\left[2\Gamma\left(1+\tfrac{1}{\alpha}\right)\right]
		+ \tfrac{1}{\alpha}+1, \\
		c_2 =& \tfrac{1}{\alpha} + 2.
		\end{align*}
	\else
		\begin{align*}
		c_1 =&\tfrac{1}{\alpha}\left|\log(2\alpha v)\right| + |\log(\ell \ue)|
		+\log\tfrac{\ue}{2}\nonumber 
		+\log\left[2\Gamma\left(1+\tfrac{1}{\alpha}\right)\right]
		+ \tfrac{1}{\alpha}+1, \\
		c_2 =& \tfrac{1}{\alpha} + 2.
		\end{align*}
	\fi
\end{theorem}

% Definition: Discrete-Continuous Random Variable Quantization
Now, we define the quantization of a random variable in amplitude domain.
\begin{definition}[Quantization of Random Variables] \label{def.AmpQuant}
	The quantized version of a random variable $X$ with the step size $1/m$ (for $m >0$) is defined as
	\begin{equation*}
		\quant{X}{m} = \tfrac{1}{m} \left\lfloor \tfrac{1}{2} + mX \right\rfloor.
	\end{equation*}
	Thus, $\quant{X}{m}$ has the pmf $P_{X;m}$ given by
	\begin{equation*}
		P_{X;m}[i]
		:=\Pr{\quant{X}{m}=\tfrac{i}{m}}
		=\int_{\big[\frac{i-0.5}{m}, \frac{i+0.5}{m} \big)}{p_X(x) \ud x}.
	\end{equation*}
	Also, we define a continuous random variable $\widetilde{X}_m$ with pdf $q_{X;m} \in \mathcal{AC}$  as follows 
	\begin{equation*}
		q_{X;m}(x) =m \, P_{X;m}[i],
		\qquad\quad x\in \big[\tfrac{i-0.5}{m}, \tfrac{i+0.5}{m} \big).
	\end{equation*}
\end{definition}

We state a useful lemma about the entropy of quantized random variables here:
% Lemma:  H([X]m-ln m = h(q)
\begin{lemma} \cite{Renyi59} \label{lmm:H([X])-ln m=h(q)}
	Let $X\sim p_X(x)$ be a continuous random variable.
	Then,
	\begin{equation*}
		\H{\quant{X}{m}} - \log m = \h{\widetilde{X}_m},
	\end{equation*}
	where $\quant{X}{m}$ is the quantization of $X$ with step size $1/m$, and $\widetilde{X}_m\sim q_{X;m}$ is the random variable defined in Definition \ref{def.AmpQuant}.
\end{lemma}

% Theorem: Entropy Dimension
The following theorem, proved in \cite{Renyi59}, measures the entropy of quantized discrete-continuous random variables by defining \emph{entropy dimension}.
\begin{theorem} \label{thm.DCEntDim} \cite{Renyi59} 
	Let $X$ be a discrete-continuous random variable defined by the triplet $(p_c, P_D, \Pr{X\in \mathcal{D}})$.
	Let $X_D$ be a discrete random variable with pmf $P_D$ and $X_c$ be a continuous random variable with pdf $p_c$. 
	If $\H{\quant{X}{1}}<\infty$, $\H{X_D}<\infty$, and $\h{X_c}$ is well-defined and finite (i.e., $\E{|\log(p_{X_c})|}<\infty$),
	where $\quant{X}{1}$ is the quantized version of $X$ with step size $1$, then the entropy of quantized $X$ with step size $1/m$ can be expressed as
	\if@twocolumn
		\begin{align*}
		\H{\quant{X}{m}}
		=&d\log m + d\h{X_c}+(1-d)\H{X_D}+\mathrm{H}_2(d) \\
		&+ o_m(1),
		\end{align*}
	\else
		\begin{align*}
		\H{\quant{X}{m}}
		=&d\log m + d\h{X_c}+(1-d)\H{X_D}+\mathrm{H}_2(d) 
		+ o_m(1),
		\end{align*}
	\fi
	where $o_m(1)$ vanishes as $m$ tends to $\infty$, and
	\begin{align*}
		d =\;& \Pr{X\notin\mathcal{D}}, \\
		h =\;& d\h{X_c}+(1-d)\H{X_D}+\mathrm{H}_2(d), \\
		\mathrm{H}_2(d) \triangleq\;& d\log{1/d}+(1-d)\log[1/(1-d)].
	\end{align*}
	The variable $d$ is called the \emph{Entropy Dimension} of $X$.
	The theorem is true for the discrete and continuous case with $d=\h{X_c}=\mathrm{H}_2(d)=0$ and $d=1,\H{X_D}=\mathrm{H}_2(d)=0$, respectively.
\end{theorem}

% Remark: Entropy dimension for continuous RVs
\begin{remark} \cite{Renyi59} \label{rem.CEntDim}
	According to Theorem \ref{thm.DCEntDim}, if a random variable $X$ is continuous with pdf $p_c$, then, the entropy of quantized $X$ with step size $1/m$ is
	\begin{equation*}
		\H{\quant{X}{m}}=\log m + \h{p_c} + o_m(1),
	\end{equation*}
	provided that
	\begin{align*}
		\H{\quant{X}{1}}<\infty,
		\qquad\int_\bR{p_c(x) \left| \log \tfrac{1}{p_c(x)} \right| \ud x} < \infty,
	\end{align*}
	where $\h{p_c}$ is the differential entropy of $X$.
\end{remark}

\iffalse
% Mixed-Pair Entropy
Finally, we point out that in \cite{Nair06}, the entropy is defined for  a \emph{mixed-pair}  $Z:=(X,Y)$, where $X$ is a discrete random variable with pmf $P$ over the sample space $\mathcal{X}$ and $Y$ is a continuous random variable with pdf $p$. 
The distribution of $Y$ conditioned to $X=n$ is assumed to be absolutely continuous, and is denoted by pdf $p_n(x)$ and.
The entropy of $Z$, $H(Z)$,  is defined as 
\begin{equation*}
	H(Z) := \H{X} + \sum_{n\in\mathcal{X}}{P[n] \, \h{p_n}},
\end{equation*}
where $\H{\cdot}$ and $\h{\cdot}$ are the entropy and the differential entropy, respectively.
\fi

%%%%%%%%%
% White Levy Noise
\subsection{White L\'evy Noises} \label{subsec.WhiteLevyNoise}
To introduce the notion of white L\'evy noises, we first define the concept of a \emph{generalized function} that generalizes ordinary functions to include distributions such as Dirac's delta function and its derivatives \cite{Unser14}.

% Definition: Schwartz Space
\begin{definition}[Schwartz Space] \cite[p. 30]{Unser14} \label{def:S Space}
	The \emph{Schwartz} space, denoted as $\mathcal{S}(\bR)$, consists of infinitely differentiable functions $\phi:\mathbb{R}\mapsto\bR$, for which
	\begin{equation*}
		\sup_{t\in\bR}{\left|t^m \frac{d^n}{dt^n}{\phi(t)}\right|}<\infty,\qquad \forall m,n\in\bN.
	\end{equation*}
	In other words, $\mathcal{S}(\bR)$ is the class of smooth functions that, together with all of their derivatives, decay faster than the inverse of any polynomial at infinity.
\end{definition}

The space of \emph{tempered distributions} or alternatively, the continuous dual of  the Schwartz space denoted by $\mathcal{S'}(\bR)$, is the set of all continuous linear mappings from $\mathcal{S}(\bR)$ into $\bR$ (also known as generalized functions). In other words, for all $x\in \mathcal{S'}(\bR)$ and $\varphi\in \mathcal{S}(\bR)$, $x(\varphi)$ is a well-defined real number. Due to the linearity of  the mapping with respect to $\varphi$, the following notations are interchangeably used:
\begin{equation*}
	x(\varphi)=\InProd{x}{\varphi}=\int_\bR{x(t) \varphi(t) \ud t},
\end{equation*}
where $\InProd{x}{\varphi}$, $x(t)$ and the integral on the right-hand side are merely notations.
This formalism is useful because it allows for a precise mathematical definition of generalized functions, such as impulse function that are common in engineering textbooks. 

Just as generalized functions extend ordinary functions, generalized stochastic processes extend ordinary stochastic processes. In particular, a generalized stochastic process is  a probability measure on $\mathcal{S'}(\mathbb{R})$. Further, observing a generalized stochastic process $X(\cdot)$ is done by applying its realizations to Schwartz functions; \emph{i.e.}, for a given $\varphi\in\mathcal{S}(\mathbb{R})$, $X_{\varphi}=\InProd{X}{\varphi}=\int_\bR{X(t) \varphi(t) \ud t}$ represents a real-valued random variable with
\begin{align*}
	\mathop{\text{Pr}\{X_{\varphi}\in \mathcal{I}\} }_{\mathcal{I}\in\mathcal{B}({\mathbb{R}})}
	=\mu\big(\{x\in\mathcal{S'}(\mathbb{R}) \,\big|\,  \InProd{x}{\varphi}\in \mathcal{I} \}\big),
\end{align*}
where $\mu(\cdot)$ stands for the probability measure on $\mathcal{S'}(\mathbb{R})$ that defines the generalized stochastic process.

% Definition: Levy Exponent
The white L\'evy noises are a subclass of generalized stochastic processes with certain properties. Before we introduce them, we define L\'evy exponents:
\begin{definition}[L\'evy Exponent] \cite[p. 59]{Unser14} \label{def.LevyExp}
	A function $f:\mathbb{R}\to\mathbb{C}$ is called a L\'evy exponent if
	\begin{enumerate}
		\item $f(0)=0$,
		\item $f$ is continuous at $0$,
		\item $\forall n\in\bN, \forall \boldsymbol{\omega}\in \bR^n$, and  $\forall \mathbf{a}\in\mathbb{C}^n$  satisfying $\sum_{i=1}^{n}{a_i}=0$,  we have that
		\begin{equation*}
			\sum_{i,j=1}^n{a_i a_j^* f(\omega_i - \omega_j)}\geq 0.
		\end{equation*}
	\end{enumerate}
\end{definition}

% Theorem: Levy-Khintchin
The following theorem provides the algebraic characterization of L\'evy exponents.
\begin{theorem}[L\'evy-Khintchin] \cite[p. 61]{Unser14} \label{thm.LevyKhintchin}
	A function $f(\omega)$ is a L\'evy exponent if and only if it can be written as
	\begin{equation*}
		-\frac{\sigma^2}{2}\omega^2 + \uj\mu \omega
		+ \int_{\bR\setminus\{0\}}
		{\left(\ue^{\uj\omega a}-1 -\uj\omega a\mathds{1}_{(-1,1)}(a) \right) \ud V(a)},
	\end{equation*}
	where $\sigma,\mu\in\bR$ are arbitrary constants.
	The function $\mathds{1}_{(-1,1)}(a)$ is an indicator function which is $1$ when $|a|<1$ and is $0$ otherwise. 
The function $V(x)$ is a non-negative non-decreasing function  that is continuous at $a = 0$ and satisfies	
	\begin{align*}
		& \lim_{a\to-\infty}{V(a)}=0,\\
		&\int_{\bR\setminus\{0\}}{\min{\lbrace 1,a^2 \rbrace} \ud V(a)} < \infty.
	\end{align*}
\end{theorem}

% Definition: White Levy Noise
\begin{definition}[White L\'evy Noises] \cite[Sec. 4.4]{Unser14} \label{def.WhiteNoise}
	A generalized stochastic process $X$ is called a white L\'evy noise, if
	\begin{equation*}
		\E{\ue^{\uj\InProd{X}{\varphi}}}
		=\exp\bigg(\int_{\bR}{f\big(\varphi(t)\big) \ud t}\bigg),
		~~~ \forall \varphi\in \mathcal{S}(\bR),
	\end{equation*}
	where $\InProd{X}{\varphi}$ is the output of the linear operator $\varphi$ under $X$, and $f$ is a valid L\'evy exponent 
	for which 
\begin{align*}
\int_{\bR\setminus [-1,1]} |a|^{\theta}\ud V(a)<\infty,
\end{align*}	
	for some $\theta>0$.
\end{definition}

% Lemma: Seperated intervals independancy
The desired properties of a white noise could be inferred from Definition \ref{def.WhiteNoise}:
\begin{lemma} \cite{Amini14} \label{lmm:Seperatedindependent}
	A white L\'evy noise $X$ is a stationary process in the sense that $\InProd{X}{\varphi_1}$ and $\InProd{X}{\varphi_2}$ have the same probability law when $\varphi_2(t) = \varphi_1(t-t_0)$.
	In addition, the independent atom property of white noise could be expressed as the statistical independence of $\InProd{X}{\varphi_1}$ and $\InProd{X}{\varphi_2}$ when $\varphi_1$ and $\varphi_2$ have disjoint supports ($\varphi_1(t)\varphi_2(t)\equiv 0$).
\end{lemma}

Next, we explain two important types of white L\'evy noises, namely stable and impulsive Poisson, that are studied in this paper.

% Definition: stable random variables
\begin{definition}[Stable random variables] \cite[p. 5]{Samor94} \label{def.StableRV}
	A random variable $X$ is stable with parameters $(\alpha,\beta,\sigma,\mu)$ if and only if its characteristic function $\widehat{p}(\omega)$ is given by
	\begin{equation*}
		\widehat{p}(\omega):=\E{\ue^{\uj\omega X}}=\ue^{f(\omega)},
	\end{equation*}
	where $f:\bR\mapsto\mathbb{C}$ is 
	\begin{align} \label{eq.StableLevyExp}
		&f(\omega)
		= \uj\omega\mu -\sigma^\alpha |\omega|^\alpha \left(1-\uj\beta\, \sgn{\omega}\Phi(\omega)\right), \\
		&\Phi(\omega)=\left\lbrace
		\begin{array}{ll}
			 \tan\frac{\pi\alpha}{2},  & \alpha\neq 1, \\
			-\frac{2}{\pi}\log{|\omega|}, & \alpha =1,
		\end{array}\right. \nonumber
	\end{align}
	 $\alpha\in (0,2]$ is the stability coefficient, $\beta\in [-1,1]$ is the skewness coefficient, $\sigma\in(0,\infty)$ is the scale coefficient, and $\mu\in\bR$ is the shift coefficient.
	In addition, function $\sgn{x}:\bR\mapsto \{-1,0,1\}$ is the \emph{sign} function defined as
	\begin{equation*}
		\sgn{x}=\left\lbrace
		\begin{array}{ll}
			1&x>0\\
			0&x=0\\
			-1&x<0
		\end{array}\right..
	\end{equation*}
\end{definition}

% Definition: alpha-stable processes
\begin{definition}[Stable white noise] {\cite[p. 87]{Sato99}} \label{def.StableProcess}
	The random process $X$ is a stable innovation process with parameters $(\alpha,\beta,\sigma,\mu)$ if $X$ is a white L\'evy noise with the L\'evy exponent $f(\omega)$ defined in \eqref{eq.StableLevyExp}.
\end{definition}

% Definition: Poisson White Noise
\begin{definition}[Impulsive Poisson white noise] \cite[p. 64]{Unser14} \label{def.Poisson}
	When the L\'evy exponent $f(\omega)$ of a white L\'evy noise satisfies 
	\begin{equation*}
		f(\omega) =
		\lambda \int_{\bR\setminus{0}}{\left(\ue^{\uj a \omega}-1 \right) \ud F_A(a)}
	\end{equation*}
	for some scalar $\lambda > 0$ (known as the rate of impulses) and cumulative distribution function $F_A$ over $\bR$ (called the amplitude cdf), then, it is called an impulsive Poisson white noise. 
\end{definition}

% Lemma: Finite Levy distribution is Poisson
\begin{lemma} \label{lmm:finite v->Poisson}
	Let $X(t)$ be a white L\'evy noise with parameters $\sigma=0$, $\mu$ and $V(a)$ such that
	\begin{equation*}
		\int_{\bR\setminus\{0\}}{\ud V(a)} < \infty.
	\end{equation*}
	Then, $X(t)$ can be decomposed as $X(t)=Y(t)+\mu'$, where $Y(t)$ is an impulsive Poisson white noise with impulse rate $\lambda = \int_{\bR\setminus\{0\}}{\ud V(a)} < \infty$ and impulse amplitude cdf $F_A(a) = \frac{1}{\lambda}V(a)$. The constant $\mu'$ is also given by
	\begin{equation*}
		\mu' = \mu-\lambda \int_{[-1,1]\setminus\{0\}}{a \ud V(a)}.
	\end{equation*}
\end{lemma}

The proof of the lemma can be found in Section \ref{subsec:prf:lmm:finite v->Poisson}.

As can be seen from Lemma \ref{lmm:finite v->Poisson}, the function $V$ has a great influence on the type of the random variables derived linearly from an innovation process. To clarify this fact, we consider the following decomposition on a generic $V$ to three non-negative increasing functions as \cite{Tucker62, Renyi59}:
\begin{equation} \label{eqn:v=vd+vac+vsc}
	V(a) = V_{ac}(a) + V_d(a) + V_{cs}(a),
\end{equation}
where
\begin{itemize}
	\item the function $V_{d}$ known as the discrete part of $V$, consists of countable finite jumps at $a\in\bR$ such that $0<V_d(a^+)-V_d(a^-)<\infty$, while $V_d(a)$ is constant elsewhere,

	\item the function $V_{ac}$ known as the absolutely continuous part of $V$,  \emph{i.e.} for any Borel set $\mathcal{A}$ with Lebesgue measure of $0$, $\int_\mathcal{A}{\ud V_{ac}(a)} = 0$, and
	
	\item the function $V_{cs}$ known as the continuous singular part of $V$, is a continuous function, but not absolutely continuous.
\end{itemize}

Then, we have the following theorem:

% Theorem: AC of Levy processes
\begin{theorem} \label{thm:InfDivRVType}
	Let $X(t)$ be a white noise with the triplet $(\mu,\sigma,V=V_d+V_{ac}+V_{sc})$, and define the random variable $X_0$ as
	\begin{equation*}
		X_0=\int_0^1{X(t) \ud t}.
	\end{equation*}
	Then,
	\begin{enumerate}
		\item \cite{Blum59} $X_0$ is discrete if and only if
		\begin{itemize}
			\item $\sigma = 0$,
			
			\item $\int_{\bR\setminus\{0\}}{\ud V(a)} < \infty$,
			
			\item $V_{ac}(a)\equiv V_{cs}(a) \equiv 0$.
		\end{itemize}
		
		\item \cite{Tucker62} $X_0$ is continuous if at least one of the following conditions is satisfied:
		\begin{itemize}
			\item $\sigma>0$,
			
			\item $\int_{\bR\setminus\{0\}}{\ud V_{ac}(a)} = \infty$.
		\end{itemize}
		
		\item \cite{Blum59} $X_0$ is discrete-continuous if and only if
		\begin{itemize}
			\item $\sigma = 0$,

			\item $\int_{\bR\setminus\{0\}}{\ud V(a)} < \infty$,

			\item $V_{cs}(a)\equiv 0$ and $V_{ac}(a)\not\equiv 0$.
		\end{itemize}
	\end{enumerate}
\end{theorem}
Note that the conditions for the discrete and the discrete-continuous cases are necessary and sufficient, while the conditions for the continuous case are only sufficient.

%%%%%%%
% Main Results
%%%%%%%
\section{Main Results} \label{sec.MainRes}
We first propose a criterion for the compressibility of stochastic processes, and study  its operational meaning from the viewpoint of source coding.
This criterion and our general results regarding its evaluation are described in Section \ref{subsec.CompCrit}.
In Section \ref{subsec.StablePoisson}, we evaluate the compressibility criterion for some special cases namely stable white noise, impulsive Poisson white noise and their sum.
Finally, we present a qualitative comparison between the compressibility of the considered innovation processes in Section \ref{subsec:Comparison}. 

%%%%%%%%%%%%%
% Compressibility Criterion
\subsection{Compressibility via quantization} \label{subsec.CompCrit}
A generic continuous-domain stochastic process is spread over the continuum of time and amplitude. Hence, we doubly discretize the process by applying time and amplitude quantization.
This enables us to utilize the conventional definition of the entropy measure. Then, we monitor the entropy trends as the quantizations become finer.

The amplitude quantization was previously defined in Definition \ref{def.AmpQuant}.
The time quantization, or equivalently, the sampling in time shall be defined in a similar fashion:

% Definition: Time Quantization
\begin{definition}[Time Quantization] \label{def.TimeQuant}
	The time quantization with step size $1/n$ of a white L\'evy noise (an innovation process) $X(t)$ is defined as
	the sequence $\big\{X_{i}^{(n)}\big\}_{i\in\mathbb{Z}}$ of random variables
	\begin{align*}
		X_i^{(n)}=\InProd{X}{\phi_{i,n}},
	\end{align*}
	where $\phi_{i,n}(t)=\phi \left(n t -i+1 \right)$ and 
	\begin{align} \label{eq.PulseFunc}
		\phi(t) = \left\lbrace
		\begin{array}{ll}
			1 & t \in [0,1)\\
			0 & t \notin [0,1)\\
		\end{array}\right..
	\end{align}
\end{definition}

% Remark: \phi is not in Schwartz
\begin{remark}
	Observe that $\phi(t)$ is not a member of $\mathcal{S}(\bR)$ as defined in Definition \ref{def:S Space}.
	Hence, strictly speaking, for a white L\'evy noise $X(t)$ with sample space $\mathcal{S}'(\bR)$, we cannot automatically define $\InProd{X}{\phi}$ based on Definition \ref{def.WhiteNoise}.
	However, the random variables $X_i^{(n)}$ could be easily interpreted as the increments of the L\'evy process corresponding to this white noise.
	Alternatively, one can define $\InProd{X}{\phi}$ as the limit of $\InProd{X}{\psi_k}$ as $k\to\infty$, when $\{\psi_k(t)\}_{k=1}^\infty\subset\mathcal{S}(\bR)$ satisfy $\lim_{k\to\infty}\int_\bR{|\psi_k(t)-\phi(t)|\ud t}=0$.
	For definitions and arguments in this paper, convergence in probability is sufficient for $\InProd{X}{\psi_k} \stackbin[k\to\infty]{}{\to} \InProd{X}{\phi}$ to hold. 
	
	In some cases, instead of just one function $\phi(t)$, we have multiple step functions $\phi_i$ for $i=0,1,\dots, n$ which are not members of $\mathcal{S}(\bR)$ and we want to simultaneously define the random variables 
	$\InProd{X}{\phi_{0}} ,\dots, \InProd{X}{\phi_{n}}$.
	To account for the simultaneous definition of $\{\InProd{X}{\phi_i}\}_{i=0}^{n}$ that captures the joint distribution, we need the convergence in probability of the multivariate random variable $[\InProd{X}{\psi_{k,0}} ,\dots, \InProd{X}{\psi_{k,n}}]$ to $[\InProd{X}{\phi_{0}} ,\dots, \InProd{X}{\phi_{n}}]$, when
	\begin{equation*}
		\lim_{k\to\infty} \sum_{i=0}^{n} \int_\bR { |\psi_{k,i}(t)-\phi_i(t)|\ud t} = 0.
	\end{equation*}
	Such convergence results could be achieved via the approach of \cite{Rajput1989}.
\end{remark}

% Compressibility Criterion
Our next step, is to find the entropy rate of a (doubly) quantized random process. Let $X(t)$ be a 
white L\'evy noise and define the random vectors of size $n$
\begin{equation*}
	\widetilde{\bf X}^{m,n}_i
	:= \left(\quant{X_{n(i-1)+1}^{(n)}}{m}, \cdots, \quant{X_{ni}^{(n)}}{m}\right),
\end{equation*}
where $X_{i}^{(n)}$ refers to the time quantization of the process (Definition \ref{def.TimeQuant}) followed by amplitude quantization in the form $\quant{X_{i}^{(n)}}{m}$ as shown in Definition \ref{def.AmpQuant}.
The time quantization in $\widetilde{\bf X}^{m,n}_i$ spans the interval $t\in[i-1,i)$ of the process $X(t)$ via $n$ random variables $\quant{X_{n(i-1)+j}^{(n)}}{m},~ j=1,\dots,n$.
Thus, the sequence $\big\{\widetilde{\bf X}^{m,n}_i \big\}_{i\in\mathbb{Z}}$ represents the innovation process over the whole real axis $t\in\mathbb{R}$ in a quantized way.  We evaluate the quantization entropy rate (entropy per unit interval of time) for $\big\{{\widetilde{\bf X}^{m,n}_i} \big\}_{i\in\mathbb{Z}}$
\begin{equation}
	\mathcal{H}_{m,n}(X)
	:=\lim_{T\to\infty}{\frac{\H{\widetilde{\bf X}^{m,n}_{-T+1}, \cdots, \widetilde{\bf X}^{m,n}_{T}} }{2T}}, \label{eqn:Hmn-ent-def}
\end{equation}
where $\H{\cdot}$ stands for the discrete entropy. 
The above definition has an operational meaning in terms of the number of bits required for asymptotic lossless compression of the source as $T$ tends to $\infty$.
For fixed $m,n$ and varying $i$, since $\quant{X_i^{(n)}}{m}$'s depend on equilength and non-overlapping time intervals of the white noise, they are independent and identically distributed (Lemma \ref{lmm:Seperatedindependent} in Sec. \ref{subsec.WhiteLevyNoise}).
Therefore,
\begin{equation} \label{def.CompCrit}
	\mathcal{H}_{m,n}(X)
	=\H{\widetilde{\bf X}^{m,n}_{1}}
	=n\H{\quant{X_1^{(n)}}{m}}.
\end{equation}
To compensate for the quantization effect, we shall study the behavior of $\mathcal{H}_{m,n}(X)$ as $m,n\to\infty$.

The following theorems consider the behavior of $\mathcal{H}_{m,n}(X)$ by showing that for a wide class of white noises we have one of the following cases
\begin{align}
	&\lim_{n\to\infty} \sup_{m\geq m(n)}
	\left|\frac{\mathcal{H}_{m,n}(X)}{\kappa(n)} - \log(m) - \zeta(n)\right|
	=0, \label{def.CompCritGeneralCDC}\\
	&\lim_{n\to\infty} \sup_{m\geq m(n)}
	\left|\mathcal{H}_{m,n}(X) - \zeta(n)\right|
	=0, \label{def.CompCritGeneralD}
\end{align}
for appropriate functions $m(n)$, $\kappa(n)$, and $\zeta(n)$.
The following two theorems prove the existence and uniqueness of $\kappa(n)$ and $\zeta(n)$ such that \eqref{def.CompCritGeneralCDC} or \eqref{def.CompCritGeneralD} hold. Let us begin with the asymptotic uniqueness of $\kappa(n)$ and $\zeta(n)$ first.

% Theorem: Uniqueness of k(n) , z(n)
\begin{theorem}[Asymptotic Uniqueness of $\kappa(n)$ and $\zeta(n)$] \label{thm:Uniqueness}
	Let $X(t)$ be a white L\'evy noise. If one can find functions $\kappa_i(n)$, $\zeta_i(n)$, and $m_i(n)$ for $i=1,2$ such that 
	\begin{equation*}
		\lim_{n\to\infty}\sup_{m\geq m_i(n)}
		{\left|\frac{\mathcal{H}_{m,n}(X)}{\kappa_i(n)}-\log m - \zeta_i(n)\right|} = 0,
	\end{equation*}
	then,
	\begin{align}
		& \exists n_0:
		\kappa_1(n)=\kappa_2(n),
		\qquad\forall n \geq n_0 \label{eqn:thm:Uniqkn}\\
		&\lim_{n\to\infty}
		{\left|\zeta_1(n)-\zeta_2(n)\right|} = 0 \label{eqn:thm:Uniqzn}.
	\end{align}
	Likewise, if
	\begin{equation*}
		\lim_{n\to\infty} \sup_{m\geq m(n)}
		\left|\mathcal{H}_{m,n}(X) - \zeta_i(n)\right|
		=0,
	\end{equation*}
	then,
	\begin{equation*}
		\lim_{n\to\infty}
		{\left|\zeta_1(n)-\zeta_2(n)\right|} = 0.
	\end{equation*}
\end{theorem}

The theorem is proved in Section \ref{subsec:prf:thm:Uniqueness}.

The above theorem shows that $\kappa(n)$ and $\zeta(n)$ are essentially unique, if they exist. In the following theorem, we show the existence of $\kappa(n)$ and $\zeta(n)$ under certain conditions. We further express $\kappa(n)$ in terms of the parameters of the white noise. Next, we state some facts about $\zeta(n)$ for a class of white noise processes.

% Theorem: Entropy Dimension Generalization
\begin{theorem} \label{thm:GenerelEntropyDim}
	Let $X(t)$ be a white noise with the triplet $(\mu, \sigma, V = V_d + V_{ac} + V_{sc} )$.
Assume that $X_0$ defined by
	\begin{equation*}
		X_0 = \int_0^1{X(t) \ud t},
	\end{equation*}
	is either a discrete, continuous, or discrete-continuous random variable (as discussed  in Theorem \ref{thm:InfDivRVType}).
	Then, the following statements hold:
	\begin{itemize}
		\item
		\underline{If $X_0$ is discrete:}\\
		There exist functions $m(n)$ and $\zeta(n)$ such that 
		\begin{align*}
			&\lim_{{n\to\infty}}\sup_{m:~m\geq m(n)}\left|
			{\mathcal{H}_{m,n}(X)-\zeta(n)}\right|
			=0,
		\end{align*}
		and
		\begin{equation*}
			c_1 \log n \leq \zeta(n) \leq c_2 \log n,
		\end{equation*}
		for some constants $0 < c_1 \leq c_2 <\infty$ (not depending on $n$) provided that $\H{X_0}$ and $\H{A}$ are finite, where
		\begin{equation*}
			p_{A}(x)
			:= \frac{1}{\int_{\bR\setminus\{0\}}{\ud V(a)}} \tfrac{\ud }{\ud x}V(x).
		\end{equation*}
		Since $X_0$ is discrete $V(a)$ is discrete and bounded (Theorem \ref{thm:InfDivRVType}), so the above definition of $p_A(x)$ corresponds to a discrete random variable $A$ and the discrete entropy $\H{A}$ is meaningful.

		\item \underline{If $X_0$ is discrete-continuous:}\\
		Set
		\begin{equation*}
			\kappa(n) =
			n\left(1-\exp\left[-\frac{\lambda}{n}\left(1-\alpha\right)\right]\right),
		\end{equation*}
		where
		\begin{equation*}
			\lambda := \int_{\bR\setminus\{0\}}{\ud V(a)},
			\qquad \alpha := \frac{1}{\lambda}\int_{\bR\setminus\{0\}}{\ud V_{d}(a)}.
		\end{equation*}
		Then, there exist functions $m(n)$ and $\zeta(n)$ such that 
		\begin{align}
			&\lim_{{n\to\infty}}\sup_{m:~m\geq m(n)}\left|
			{\frac{\mathcal{H}_{m,n}(X)}{\kappa(n)} - \log(m) -\zeta(n)}\right|
			=0, \label{eqn:H/k-lnm-z=0}
		\end{align}
		and
		\begin{equation*}
			c_1 \log n \leq \zeta(n) \leq c_2 \log n,
		\end{equation*}
		for some constants $0 < c_1 \leq c_2 <\infty$ (not depending on $n$) if random variables ${X_{0,c}}$, ${X_{0,D}}$, $X_{1,c}^{(n)}$, ${A_c}$, and ${A_D}$ defined below satisfy the technical assumptions that $\H{\quant{X_0}{1}}$, $\h{X_{1,c}^{(n)}}$, $\h{A_c}$, $\H{X_{0,D}}$ and $\H{A_D}$ are well-defined for all $n$:
		random variables $X_{0,D}$ and $X_{0,c}$ are the discrete and continuous parts of $X_0$, respectively.
		Similarly, $X_{1,c}^{(n)}$ stands for the continuous part of $X_1^{(n)}$ (Theorem \ref{thm:InfDivRVType}, shows that $X_1^{(n)}$ is also discrete-continuous if $X_0$ is discrete-continuous).
		The random variables ${A_c}$ and ${A_D}$ correspond to the distributions
		\begin{equation*}
			\qquad p_{A_c}(x)
			:= \frac{1}{\int_{\bR\setminus\{0\}}{\ud V_{ac}(a)}} \tfrac{\ud }{\ud x} V_{ac}(x),\text{~~~and~~~}
			p_{A_D}(x)
			:= \frac{1}{\int_{\bR\setminus\{0\}}{\ud V_{d}(a)}} \tfrac{\ud }{\ud x} V_{d}(x),
		\end{equation*}
		respectively.
		
		\item \underline{If $X_0$ is continuous:}\\
		Set $\kappa(n) = n$ and $\zeta(n)=\h{X_1^{(n)}}$ where $X_1^{(n)}=\int_{0}^{\frac{1}{n}}{X(t) \ud t}$.
		If at least one of the assumptions in part (2) of Theorem \ref{thm:InfDivRVType} is satisfied, $\H{\quant{X_0}{1}}$ is finite and $\h{X_1^{(n)}}$ is well-defined for all $n$, then, we have that
		\begin{align}
			&\lim_{{n\to\infty}}\sup_{m:~m\geq m(n)}\left|
			{\frac{\mathcal{H}_{m,n}(X)}{\kappa(n)} - \log(m) -\zeta(n)}\right|
			=0 \label{eqn:H/k-lnm-z=0C}.
		\end{align}
		Furthermore, $\zeta(n)=\h{X_1^{(n)}}$ is a non-increasing function in $n$. 
	\end{itemize}
\end{theorem}

The proof can be found in Section \ref{prf:thm:GenerelEntropyDim}.

% Remark: Ac and AD
\begin{remark} \label{rmk:X0DC_AcAD}
	In order to define ${A_c}$ and ${A_D}$, which are used in the case that $X_0$ is a discrete-continuous random variable, note that Theorem \ref{thm:InfDivRVType} in conjunction with Lemma \ref{lmm:finite v->Poisson}  shows that $X_0$ is discrete-continuous only if the white noise is an impulsive Poisson white noise, plus a constant.
	Utilizing Theorem \ref{thm:InfDivRVType}, we can define continuous random variable $A_c$ with pdf $p_{A_c}$ such that:
	\begin{equation*}
		\qquad p_{A_c}(x)
		:= \frac{1}{\int_{\bR\setminus\{0\}}{\ud V_{ac}(a)}} \tfrac{\ud}{\ud x}V_{ac},
	\end{equation*}
	where $V_{ac}$ was defined in \eqref{eqn:v=vd+vac+vsc}.
	Note that, according to Theorem \ref{thm:InfDivRVType}, even if $V(a)$ does not have any discrete part, $X_0$ can still be discrete-continuous.
	Hence, if the discrete part of $V(a)$ exists, one can define $A_D$ with pmf $p_{A_D}$ such that
	\begin{equation*}
		p_{A_D}(x)
		:= \frac{1}{\int_{\bR\setminus\{0\}}{\ud V_{d}(a)}} \tfrac{\ud}{\ud x} V_{d}.
	\end{equation*}
\end{remark}

%%%%%%%%%%%%%%%
% Stable & Poisson White Noises
\subsection{Special Cases} \label{subsec.StablePoisson}
In this section, we evaluate functions $\zeta(n)$ and $m(n)$ in cases where white noise is impulsive Poisson, stable, or sum of impulsive Poisson and stable.

% Theorem: Stable Compressibility Criterion
\begin{theorem} \label{thm.AlphaCompCrit}
	Let $X(t)$ be a stable white noise with parameters $(\alpha,\beta,\sigma,\mu)$.
	We define
	\begin{equation*}
		X_0 := \InProd{X}{\phi}=\int_0^1{X(t) \ud t},
	\end{equation*}
	where $\phi$ is the function defined in \eqref{eq.PulseFunc}.
	Then, for all $m(\cdot):\bN\to \bN$ satisfying $\lim_{n\to \infty} \frac{m(n)}{\sqrt[\alpha]{n}}=\infty$ we have that
	\begin{equation*}
		\lim_{n\to\infty}\sup_{m:m\geq m(n)}{
		\bigg| \frac{\mathcal{H}_{m,n}(X)}{n}
		-\log \tfrac{m}{\sqrt[\alpha]{n}}} 
		- \h{X_0}\bigg| =0,
	\end{equation*}
	where $\mathcal{H}_{m,n}(X)$ is defined in \eqref{def.CompCrit}, and $\h{X_0}$ is the differential entropy of continuous random variable $X_0$.
\end{theorem}
The proof is given in Section \ref{subsec:ProofAlphaCompCrit}.

% Theorem: Poisson Compressibility Criterion
\begin{theorem} \label{thm.PoissonCompCrit}
	Let $X(t)$ be an impulsive Poisson white noise with rate $\lambda$ and amplitude pdf $p_A$
	such that:
	\begin{equation}
		p_A \in (\alpha, \ell, v)\text{--}\mathcal{AC},\label{eqn:defp_AAC}
	\end{equation}
	for some positive constants $\alpha, \ell, v$ and 
	\begin{align*}
		\int_\mathbb{R}{p_A(x) \left| \log \tfrac{1}{p_A(x)} \ud x \right|} < \infty,
		\qquad\H{\quant{A}{1}}<\infty,
	\end{align*}
	where $A \sim p_A$ and $[A]_1$ is the quantization of $A$ with step size one.
	Then, for any function $m(\cdot):\bN\to\bN$ satisfying $\lim_{n\rightarrow \infty}\frac{m(n)}{\log n}=\infty$ we have that
	\begin{equation*} 
		\lim_{{n\to\infty}}\sup_{m:~m\geq m(n)}\left|
		{\frac{\mathcal{H}_{m,n}(X)}{\kappa(n)} - \log (mn)  -\h{A} +\log{\lambda} - 1}\right|
		=0,
	\end{equation*}	
	where $\kappa(n)=n(1-\ue^{-\lambda/n})$. 
\end{theorem}
The proof is given in Section \ref{subsec:ProofPoissonCompCrit}.

% Corollary: Poisson
\begin{corollary}
	Let $X(t)$ be a white L\'evy noise with parameters $\sigma=0$, $\mu$, and absolutely continuous and bounded function $V(a) = V_{ac}(a)$, which was defined in \eqref{eqn:v=vd+vac+vsc}.
	Then, utilizing Lemma \ref{lmm:finite v->Poisson}, the compressibility of $X(t)$ is similar to the compressibility of the impulsive Poisson case with
	\begin{equation*}
		\lambda := \int_{\bR\setminus\{0\}}{\ud V(a)},
		\qquad A \sim p_A(x) := \tfrac{1}{\lambda}\tfrac{\ud }{\ud x} V(x),
	\end{equation*}
	provided that $p_A$ satisfies the conditions in Theorem \ref{thm.PoissonCompCrit}.
\end{corollary}

% Theorem: Combination of Poisson and stable white noises
\begin{theorem} \label{thm:LinCom}
	Let $X(t)$ be a stable white noise with parameters $(\alpha,\beta,\sigma,\mu)$ and $Y(t)$ be an impulsive Poisson white noise, independent of $X(t)$, with rate $\lambda$ and amplitude random variable $A$, such that there exists $\theta>0$ where
	\begin{equation*}
		\E{|A|^\theta}<\infty.
	\end{equation*}
	We define
	\begin{equation*}
		X_0 := \InProd{X}{\phi}=\int_0^1{X(t) \ud t},
	\end{equation*}
	where $\phi$ is the function defined in \eqref{eq.PulseFunc}.
	Then, for $Z(t)=X(t) + Y(t)$ and all $m(\cdot):\bN\to \bN$ satisfying $\lim_{n\to \infty} \frac{m(n)}{\sqrt[\alpha]{n}}=\infty$ we have that
	\begin{equation*}
		\lim_{n\to\infty}\sup_{m:m\geq m(n)}{
		\bigg| \frac{\mathcal{H}_{m,n}(Z)}{n}
		-\log \tfrac{m}{\sqrt[\alpha]{n}}} 
		- \h{X_0}\bigg| =0,
	\end{equation*}
	where $\mathcal{H}_{m,n}(Z)$ is defined in \eqref{def.CompCrit}.
\end{theorem}

The theorem is proved in Section \ref{subsec:pfr:thm:LinCom}.

% Remark: General Poisson in the combination
\begin{remark}
	Note that Theorem \ref{thm:LinCom}  assumes that the amplitude measure of the impulsive Poisson only has a finite moment $\theta$ for an arbitrary $\theta>0$, while Theorem \ref{thm.PoissonCompCrit} has the stronger assumption that amplitude measure is in $\left(\theta,\ell,v\right)\text{--}\mathcal{AC}$.
\end{remark}

\begin{corollary}
	Let $Z(t)$ be a white L\'evy noise with parameters $\sigma > 0$, $\mu$, and \emph{bounded} function $V(a)$.
	Then, $Z(t)$ can be decomposed as $Z(t)=X_1(t)+Y_1(t)$ where $X_1(t)$ is a Gaussian white noise with parameters $(0,\sigma)$ and $Y_1(t)$ is a white L\'evy noise with parameters $\sigma = 0$, $\mu$, and the {bounded} function $V(a)$.
	According to Lemma \ref{lmm:finite v->Poisson}, we have that $Y_1(t)=\mu'+Y(t)$ where $Y(t)$ is an impulsive Poisson white noise with parameters $\lambda$ and amplitude cdf $F_A$ where
	\begin{equation*}
		\lambda := \int_{\bR\setminus\{0\}}{\ud V(a)} < \infty,
		\qquad A \sim F_A(x) := \frac{1}{\lambda}V(x),
		\qquad\mu' := \mu-\lambda \int_{[-1,1]\setminus\{0\}}{a \ud V(a)}.
	\end{equation*}
	Then, the compressibility of $Z(t)$ is simillar to the above case with
	\begin{equation*}
		\alpha=2,
		\qquad \h{X_0} = \frac{1}{2} \log(2\pi\ue\sigma^2).
	\end{equation*}
	provided that $\E{|A|^\theta}<\infty$ for some $\theta>0$.
\end{corollary}

%%%%%%%
% Comparison
\subsection{Comparison and discussion}\label{subsec:Comparison}
In this section, we compare different white L\'evy noises based on Theorem \ref{thm:GenerelEntropyDim}.
Take two white L\'evy noises $X(t)$ and $Y(t)$  belonging to the general class of processes considered in Theorem \ref{thm:GenerelEntropyDim}.
Let $m_X(n)$ and $m_Y(n)$ be the functions associated with these two white noises in Theorem \ref{thm:GenerelEntropyDim}.
Observe that if Theorem \ref{thm:GenerelEntropyDim} holds for a choice of $m(n)$, it will also hold for any function $m'(n)\geq m(n)$.
As a result, we can look at the sequence of pairs $(n,m)$ where $m \geq\max\{m_X(n), m_Y(n)\}$ when we compare the two white noise processes.
Therefore to make the comparison we only need to compare the values of $\kappa(n)$ and $\zeta(n)$ for the two white noise processes asymptotically.

% Theorem: General Comparison
\begin{theorem} \label{thm:GeneralComparison}
	Let $X(t)$ and $Y(t)$ be two white L\'evy noises and define $X_0:=\int_0^1{X(t) \ud t}$ and $Y_0:=\int_0^1{Y(t) \ud t}$.
	Assume that conditions in Theorem \ref{thm:GenerelEntropyDim} hold for $X(t)$ and $Y(t)$ based on their type ($X_0$ and $Y_0$ might have different types).
	Let $\kappa_X(n), \kappa_Y(n)$, $\zeta_X(n), \zeta_Y(n)$ , and $m_X(n), m_Y(n)$ be the functions associated in Theorem \ref{thm:GenerelEntropyDim} to $X$ and $Y$, respectively. 	
	\begin{enumerate}
		\item \label{itm:D-DC}
		If $X_0$ is discrete and $Y_0$ is discrete-continuous, then
		\begin{equation*}
			\lim_{{n\to\infty}}\sup_{m\geq m(n)}
			{\frac{\mathcal{H}_{m,n}(X)}{\mathcal{H}_{m,n}(Y)}} = 0,
		\end{equation*}
		where $m(n)$ is a function such that
		\begin{equation*}
			m(n) \geq \max\{m_X(n),m_Y(n)\},
			\qquad \lim_{n\to\infty}\frac{\log m(n)}{\log n} = \infty.
		\end{equation*}
		
		\item \label{itm:DC-C}
		If $X_0$ is discrete-continuous and $Y_0$ is continuous, then
		\begin{equation*}
			\lim_{{n\to\infty}}\sup_{m\geq m(n)}
			{\frac{\mathcal{H}_{m,n}(X)}{\mathcal{H}_{m,n}(Y)}} = 0,
		\end{equation*}
		where $m(n)$ is a function such that
		\begin{equation*}
			m(n) \geq \max\{m_X(n),m_Y(n)\},
			\qquad \lim_{n\to\infty}\frac{\log m(n)}{\zeta_Y(n)} = \infty.
		\end{equation*}
		
		\item \label{itm:C-CG}
		If $X_0$ is continuous and $Y(t)$ is a Gaussian white noise with parameters $\sigma>0$ and $V(a)=0$ for all $a\in\bR$, then
		\begin{equation*}
			\limsup_{{n\to\infty}}\sup_{m: m\geq m(n)}
			{\frac{\mathcal{H}_{m,n}(X)}{\mathcal{H}_{m,n}(Y)}} \leq 1,
		\end{equation*}
		where $m(n)$ is a function such that
		\begin{equation*}
			m(n) \geq \max\{m_X(n),m_Y(n)\},
			\qquad \lim_{n\to\infty}{\frac{m(n)}{\sqrt{n}}} = \infty.
		\end{equation*}
	\end{enumerate}
\end{theorem}

The proof of the theorem is presented in Section \ref{subsec:prf:thm:GeneralComparison}.

The theorem, intuitively states that when $m$ grows sufficiently faster than $n$, the discrete white noises are asymptotically more compressible than discrete-continuous noises, and the discrete-continuous white noises are more compressible than continuous noises.
In addition, Gaussian white noise is the least compressible white noise.

In the following theorem, we compare impulsive Poisson and stable white noises amongst themselves.

% Theorem: Special cases comparison
\begin{theorem} \label{thm:Comparison}
The following statements apply to impulsive Poisson and stable white noises.
\begin{enumerate}[label=\roman*)]
	\item \label{part:ComparStable}
	A stable white noise becomes less compressible as its stability parameter $\alpha$ increases.
	Let $X_1$ and $X_2$ be stable white noises with stability parameters $0<\alpha_1<\alpha_2\leq 2$, and functions $m_1(n)$ and $m_2(n)$ defined in Theorem \ref{thm:GenerelEntropyDim}, respectively.
	Then, for $m(n) \geq \max\{m_1(n),m_2(n)\}$, we have that
	\begin{equation} \label{eqn:ComparStable/}
		\limsup_{n\to\infty} \sup_{m\geq m(n)}
		\frac{\mathcal{H}_{m,n}(X_1)}{\mathcal{H}_{m,n}(X_2)} \leq 1,
	\end{equation}
	\begin{equation} \label{eqn:ComparStable-}
		\lim_{n\to\infty} \sup_{m\geq m(n)}
		{\mathcal{H}_{m,n}(X_1)-\mathcal{H}_{m,n}(X_2)} = -\infty.
	\end{equation}

	\item \label{part:ComparPoisson}
	An impulsive Poisson white noise becomes less compressible as its rate parameter $\lambda$ increases.
	Let $X_1$ and $X_2$ be impulsive Poisson white noises with rate parameters $0<\lambda_1 <\lambda_2$, absolutely continuous amplitude distributions, and functions $m_1(n)$ and $m_2(n)$ defined in Theorem \ref{thm:GenerelEntropyDim}, respectively.
	Then, for $m(n) \geq \max\{m_1(n),m_2(n)\}$, we have that
	\begin{equation} \label{eqn:ComparPoisson/}
		\lim_{n\to\infty} \sup_{m\geq m(n)}
		\frac{\mathcal{H}_{m,n}(X_1)}{\mathcal{H}_{m,n}(X_2)} = \frac{\lambda_1}{\lambda_2}
		< 1,
	\end{equation}
	\begin{equation} \label{eqn:ComparPoisson-}
		\lim_{n\to\infty} \sup_{m\geq m(n)}
		{\mathcal{H}_{m,n}(X_1)-\mathcal{H}_{m,n}(X_2)}
		= -\infty.
	\end{equation}
\end{enumerate}
\end{theorem}
The proof can be found in Section \ref{subsec:Comparison:proof}.

We should mention that Theorem \ref{thm:Comparison} confirms with previously studied notions of compressibility in \cite{Amini11,Gribonval2012,Silva15}.
However, Part \ref{itm:DC-C} of Theorem \ref{thm:GeneralComparison} deviates from the available literature by identifying the impulsive Poisson white noises as more compressible than heavy-tailed stable white noises.
This difference is fundamental and is caused by basing the compressibility definition on the probability concentration properties of processes rather than their amplitude distribution.

%%%%%%%%
% Useful Lemmas
%%%%%%%%
\section{Some Useful Lemmas} \label{sec.UsefulLem}
In Section \ref{subsec:PoissonUseful}, we provide  a useful theorem for Poisson white noises, and in Section \ref{subsec:AmpQuant}, we give some lemmas about quantization of a random variable in amplitude domain. These results are used in the proof of the main results.

%%%%%%%%%%%
% Poisson White Noise 
\subsection{Poisson and Stable White Noises} \label{subsec:PoissonUseful}
% Theorem: Quantized Poisson Process Distribution
The following theorem states a feature of integrated Poisson white noise in a small interval.
\begin{lemma} \label{lmm.PoissonQuantDist}
	Let $X(t)$ be an impulsive Poisson white noise with rate $\lambda$ and amplitude pdf $p_A\in\mathcal{AC}$.
	Define $Y_n$ as
	\begin{equation}
		Y_n = \InProd{X}{\phi(nt)}=\int_0^\frac{1}{n}{X(t)\ud t}. \label{eq.Yn}
	\end{equation}
	Then, we have that
	\begin{equation} \label{eq.pYn}
		p_{Y_n}(x)=
		\ue^{-\frac{\lambda}{n}} \delta (x) + \left( 1-\ue^{-\frac{\lambda}{n}} \right) p_{A_n}(x),
	\end{equation}
	\begin{equation} \label{eq:pAn}
		p_{A_n}(x)=\frac{1}{\ue^{\frac{\lambda}{n}}-1}
		\sum_{k=1}^\infty
		{\frac{\left( \frac{\lambda}{n} \right)^k}{k!}\big(\overbrace{p_A*\cdots*p_A}^k\big) (x)}, 
	\end{equation}
	where $p_{Y_n}$ is the probability density function of $Y_n$, $*$ is the convolution operator, and $A_n$ is a random variable with probability density function $p_{A_n}$.
	Further, we have that
	\begin{equation} \label{eq.VarConv}
		\int_\mathbb{R}{\left|p_{A_n}(x)-p_A(x)\right|\ud x}
		\leq 2\frac{\ue^{\frac{\lambda}{n}}-\frac{\lambda}{n}-1}{\ue^{\frac{\lambda}{n}}-1}, 
	\end{equation}
	\begin{equation} \label{eqn:EAn->EA}
		\E{\left| A_n \right|^\alpha} \text{ exists for all $n$},
		\qquad\lim_{n\to\infty}{\E{\left| A_n \right|^\alpha}}
		=\E{\left| A \right|^\alpha},\qquad\forall \alpha\in(0,\infty),
	\end{equation}
	where $A$ is a random variable with the probability density function $p_A$.
	Also, the upper bound on the total variation distance in \eqref{eq.VarConv} vanishes as $n$ tends to infinity.
\end{lemma}

The main novelties of Lemma \ref{lmm.PoissonQuantDist}  are \eqref{eq.VarConv} and \eqref{eqn:EAn->EA}, as \eqref{eq.pYn} and \eqref{eq:pAn} are rather classical results. To make the paper self-contained, we prove all the identities in Section \ref{subsec:ProofPoissonQuantDist}.

% Lemma: Renyi Conditions for Stable White Noises
\begin{lemma} \label{lmm.AlphaStableRenyiCondition}
	Let $X$ be a stable white noise with parameters $(\alpha,\beta,\sigma,\mu)$.
	Define random variable $X_0$ with pdf $p_{X_0}$ as follows:
	\begin{equation*}
		X_0 = \InProd{X}{\phi}=\int_0^1{X(t) \ud t},
	\end{equation*}
	where $\phi$ is the function defined in \eqref{eq.PulseFunc}.
	Then,
	\begin{equation} \label{eqn:pX0AC}
		p_{X_0} \in \mathcal{AC},
	\end{equation}
	\begin{equation} \label{eqn:pX0<L}
		\mathrm{ess}\sup_{x\in\bR}{p_{X_0}(x)}=L < \infty,
	\end{equation}
	\begin{equation} \label{eqn:pX0PWC}
		p_{X_0}(x) \text{ is a piecewise continuous function},
	\end{equation}
	\begin{equation} \label{eqn:hpX0<inf}
		\int_\mathbb{R}{p_{X_0}(x) \left|\log \tfrac{1}{p_{X_0}(x)}\right| \ud x} < \infty,
	\end{equation}
	Moreover, we have that
	\begin{equation} \label{eqn:HpX0<inf}
		\H{\quant{X_0}{1}} < \infty.
	\end{equation}
\end{lemma}

Except \eqref{eqn:HpX0<inf}, other claims in Lemma \ref{lmm.AlphaStableRenyiCondition} are known results in the literature. Again, we include the proof of all parts in Section \ref{subsec:ProofX0features} for the sake of completeness.

% Lemma: distribution of a stable white noise sample
\begin{lemma} \label{lmm:StableXnDist}
	For a stable white noise with parameters $(\alpha,\beta,\sigma,\mu)$ we have that
	\begin{equation*}
		X_1^{(n)} \stackrel{d}{=} \frac{X_0-b_n}{\sqrt[\alpha]{n}},
	\end{equation*}
	where $X_1^{(n)}$ was defined in Definition \ref{def.TimeQuant},
	\begin{equation*}
		X_0 := \InProd{X}{\phi}=\int_0^1{X(t) \ud t},
	\end{equation*}
	and
	\begin{equation*}
		b_n =\left\lbrace
		\begin{array}{ll}
			\mu\left(1-\frac{\sqrt[\alpha]{n}}{n}\right) &\alpha\neq 1\\
			\frac{2}{\pi}\sigma\beta\ln{n} &\alpha=1
		\end{array}
		\right..
	\end{equation*}
\end{lemma}

The proof can be found in Section \ref{subsec:prf:lmm:StableXnDist}

%%%%%%%%%%%%
% Amplitude Quantization
\subsection{Amplitude Quantization} \label{subsec:AmpQuant}
% Lemma: infinite interval entropy
In the following lemma, we show that the total variation distance between two variables decreases by quantizing;
furthermore, moments of a quantized random variable tends to the moments of the original random variable, as the quantization step size vanishes.
\begin{lemma} \label{lmm.QuantNorm}
	Let random variables $X \sim p_X$ and $Y \sim p_Y$ be continuous, and $\quant{X}{m}\sim P_{X;m}$ and $\quant{Y}{m}\sim P_{Y;m}$ be their quantized version, defined in Definition \ref{def.AmpQuant}, respectively.
	Then for all $m \in (0,\infty)$ we have
	\begin{align}
		&\int_\mathbb{R}{\left| q_{X;m}(x) - q_{Y;m}(x) \right|\ud x}
		\leq \int_\mathbb{R}{\left| p_X(x) - p_Y(x) \right|\ud x}, \label{eq.QuantNorm1}
	\end{align}
	where $\widetilde{X}_m \sim q_{X;m}$ and $\widetilde{Y}_m \sim q_{Y;m}$ are random variables defined in Definition \ref{def.AmpQuant}.
	In addition, for any $\alpha\in(0,\infty)$, and $m\geq4$ we have
	\begin{equation} \label{eq.QuantNorm2}
		\E{\left|\widetilde{X}_m\right|^\alpha}
		\leq\left(\tfrac{2}{\sqrt{m}}\right)^\alpha+\ue^\frac{\alpha}{\sqrt{m}}\E{|X|^\alpha},
	\end{equation}
	\begin{equation} \label{eq.QuantNorm3}
		\Pr{|X|>\tfrac{1}{\sqrt{m}}} \ue^{-2\frac{\alpha}{\sqrt{m}}}\E{|X|^\alpha}
		\leq\E{\left|\widetilde{X}_m\right|^\alpha},
	\end{equation}
	provided that $\E{|X|^\alpha}$ exists.
\end{lemma}
The proof is given in Section \ref{subsec:ProofQuantNorm}.

% Corollary: EXm^a -> EX^a
\begin{corollary}
	Let $X \sim p_X$ be a continuous random variable, and let $\widetilde{X}_m\sim q_{X;m}$ be the random variable defined in Definition \ref{def.AmpQuant}.
	Then, we have that
	\begin{equation*}
		\lim_{m\to\infty}{\E{\left|\widetilde{X}_m\right|^\alpha}}
		=\E{|X|^\alpha},
	\end{equation*}
	if $\E{|X|^\alpha}$ exists.
\end{corollary}

% Proof of Corollary: EXm^a -> EX^a
\begin{proof}
	Since $p_X\in\mathcal{AC}$, we conclude that $\Pr{|X|\leq{1}/{\sqrt{m}}}$ vanishes as $m$ tends to $\infty$.
	Hence, the corollary achieved from Lemma \ref{lmm.QuantNorm}.
\end{proof}

% Lemma: H([aX])
The following lemma discusses the entropy of quantized version of multiplies of a continuous random variable.
\begin{lemma} \label{lmm.H([aX])}
	Let $X\sim p$ be a continuous random variable, and $m\in(0,\infty)$ be arbitrary.
	If $\H{\quant{X}{m}}$ exists, then for all $a\in (0,\infty)$, we have
	\begin{equation*}
		\H{\quant{aX}{\frac{m}{a}}}
		= \H{\quant{X}{m}}.
	\end{equation*}
\end{lemma}
The lemma is proved in Section \ref{subsec:ProofH([aX]}.

% Lemma: Entropy dimension for shifted sontinuous RVs
In the following lemma, we extend Remark \ref{rem.CEntDim} to an arbitrary shifted continuous random variables.
\begin{lemma} \label{lmm.DifEntShft}
	Let $X\sim p$ be a continuous random variable with a piecewise continuous pdf $p(x)$.
	For an arbitrary sequence $\lbrace c_m \rbrace_{m=1}^\infty\subset \bR$, we have 
	\begin{equation*}
		\lim_{m\to\infty}{\H{\quant{X+c_m}{m}}-\log m}=\h{X},
	\end{equation*}
	provided that
	\if@twocolumn
		\begin{equation*}
		\H{\quant{X}{1}} < \infty,
		\end{equation*}
		\begin{equation*}
		\int_\bR{p(x) \left| \log \tfrac{1}{p(x)} \right| \ud x} < \infty, ~~~\text{and}
		\end{equation*}
		\begin{equation*}
		\mathrm{ess}\sup_{x\in\bR}{p(x)}=L < \infty,
		\end{equation*}
	\else
		\begin{align*}
		\H{\quant{X}{1}} < \infty, ~~~
		\int_\bR{p(x) \left| \log \tfrac{1}{p(x)} \right| \ud x} < \infty, ~~~\text{and}~~~
		\mathrm{ess}\sup_{x\in\bR}{p(x)}=L < \infty,
		\end{align*}
	\fi
	where $\quant{X}{1}$ is the quantized version of $X$ with step size $1$.
\end{lemma}
The proof can be found in Section \ref{subsec:ProofDifEntShft}.

%%%%%%%
% (a,l,v)-AC
\subsection{On Sum of Independent Random Variables} \label{subsec:AC+RV}
% Lemma:
The following lemma is well-known and can be easily proved. Hence, we only mention it without proof.
\begin{lemma} \label{lmm:DC+DC}
	Let $X,Y$ be two independent discrete-continuous random variables where $p,q$ are the probability of $X,Y$ being discrete, respectively.
	Hence, we can write
	\begin{equation*}
		X=
		\begin{cases}
			X_D & p \\
			X_c & 1-p
		\end{cases},
		\qquad Y=
		\begin{cases}
			Y_D & q \\
			Y_c & 1-q
		\end{cases},
	\end{equation*}
	where $X_D, Y_D$ are discrete, and $X_c, Y_c$ are continuous random variables.
	Therefore, $Z=X+Y$ is also a discrete-continuous random variable such that
	\begin{equation*}
		Z=
		\begin{cases}
			X_D + Y_D & p q \\
			X_D + Y_c & p (1-q) \\
			X_c + Y_D & (1-p) q \\
			X_c + Y_c & (1-p) (1-q)
		\end{cases},
	\end{equation*}
	where the first case makes the discrete part while the other cases make the continuous part of $Z$.
\end{lemma}

% Lemma:
\begin{lemma} \label{lmm:h(X+Y)c}
	Let $X,Y$ be two independent discrete-continuous random variables where $X_D,Y_D$ are the discrete part, and $X_c,Y_c$ are the continuous part of $X$ and $Y$, respectively.
	Then, we have that
	\begin{equation*}
		\h{Z_c}
		\geq \min\{\h{X_c} , \h{Y_c}\},
	\end{equation*}
	where $Z=X+Y$.
\end{lemma}
The proof can be found in Section \ref{subsec:prf:lmm:h(X+Y)c}.

% Corollary
\begin{corollary} \label{cor:h(X+Y)c}
	Using induction, it can be proved that for \emph{i.i.d.} continuous-discrete random variables $X_1,\cdots,X_n$, we have that
	\begin{equation*}
		\h{S_{n,c}} \geq \h{X_{1,c}},
	\end{equation*}
	where $S_{n,c}$ and $X_{1,c}$ are the corresponding continuous random variables of $X_1+\cdots+X_n$ and $X_1$, respectively.
\end{corollary}

% Lemma: 
\begin{lemma} \label{lmm:AC+RV} 
	Let $X$ be a continuous random variable defined in Definition \ref{def.ACRV} with probability density function $p_X(x)$ and $Y$ be an arbitrary random variable with probability measure $\mu_Y$, independent of $X$. Let  $Z=X+Y$ with probability measure $\mu_Z$.
	Then, $Z$ is also continuous random variable with pdf $p_Z(z)$ defined as following:
	\begin{equation*}
		p_Z(z) = \E{p_X \left(z - Y\right)}.
	\end{equation*}
	Moreover,
	\begin{equation*}
		p_Z(z) \leq \ell,
		\qquad \forall z\in\bR,
	\end{equation*}
	provided that
	\begin{equation*}
		p_X(x) \leq \ell,
		\qquad \forall x\in\bR.
	\end{equation*}
\end{lemma}

The theorem is proved in Section \ref{subsec:prf:lmm:AC+RV}.

%%%%%%%%%%
% Maximum Entropy
\subsection{Maximum of $\zeta(n)$ for a Class of White Noise Processes}
In the following lemma we find maximum $\zeta(n)$, defined in Theorem \ref{thm:GenerelEntropyDim}  over a class of white noise processes. 

% Lemma: maximum zeta amon continuous
\begin{lemma} \label{lmm:zeta<EntP}
	Let $X(t)$ be a white L\'evy noise such that $X_0 := \int_0^1{X(t) \ud t}$ is a continuous random variable (as discussed in Theorem \ref{thm:InfDivRVType}).
	Furthermore, assume that $h(X_0)$ and $\H{\quant{X_0}{1}}$ are well-defined.
	Let $m'(n)$ and $\zeta(n)$ be any function satisfying the statement \eqref{eqn:H/k-lnm-z=0} in Theorem \ref{thm:GenerelEntropyDim} for the white noise $X(t)$.
	Define $m(n)$ an arbitrary function that $m(n)\geq m'(n)$ and $\lim_{n\to\infty}{m(n)/\sqrt{n}}=\infty$.
	Then, we have that
	\begin{equation*}
		\limsup_{n\to\infty}\left({\zeta(n) - \log\frac{1}{\sqrt{n}}}\right) \leq \h{X_0}.
	\end{equation*}
\end{lemma}

The theorem is proved in Section \ref{subsec:prf:lmm:zeta<EntP}.

%%%%%
% Proofs
%%%%%
\section{Proofs} \label{sec.Proofs}
%%%%%%%%%%%%%%%%%%%%%%%
% Proof of Theorem: Uniqueness on k(n) and z(n)
\subsection{Proof of Theorem \ref{thm:Uniqueness}} \label{subsec:prf:thm:Uniqueness} 
	\textbf{Proof of the first case}:

	\emph{Proof of \eqref{eqn:thm:Uniqkn}}: There is $n_0\in\mathbb{N}$ such that for any $n\geq n_0$ we have
	\begin{equation} \label{enq:thm:Uniq1}
		\sup_{m\geq m(n)}
		\left|\frac{\mathcal{H}_{m,n}(X)}{\kappa_i(n)} - \log m - \zeta_i(n) \right|
		\leq 1,
		\qquad i=1,2,
	\end{equation}
	where $m(n)$ can be any function larger than $\max\{m_1(n),m_2(n)\}$.
	We can write
	\begin{align*}
		&\left|\frac{\mathcal{H}_{m,n}(X)}{\kappa_1(n)} - \log m - \zeta_1(n) \right| \\
		&\qquad =\left|\frac{\kappa_2(n)}{\kappa_1(n)}
		\left[\frac{\mathcal{H}_{m,n}(X)}{\kappa_2(n)} - \log m -\zeta_2(n)\right]
		-\left(1-\frac{\kappa_2(n)}{\kappa_1(n)}\right)\log m 
		- \zeta_1(n) + \zeta_2(n)\frac{\kappa_2(n)}{\kappa_1(n)} \right| \\
		&\qquad \geq \left|1-\frac{\kappa_2(n)}{\kappa_1(n)}\right| \log m
		- \frac{\kappa_2(n)}{\kappa_1(n)}
		\left|\frac{\mathcal{H}_{m,n}(X)}{\kappa_2(n)} - \log m -\zeta_2(n)\right|
		-\left|\zeta_1(n) - \zeta_2(n)\frac{\kappa_2(n)}{\kappa_1(n)} \right|.
	\end{align*}
	Hence, utilizing \eqref{enq:thm:Uniq1} for $i=2$, we have that for all $m\geq m(n)$, we have that
	\begin{align*}
		&\left|\frac{\mathcal{H}_{m,n}(X)}{\kappa_1(n)} - \log m - \zeta_1(n) \right|
		\geq \left|1-\frac{\kappa_2(n)}{\kappa_1(n)}\right| \log m
		- \frac{\kappa_2(n)}{\kappa_1(n)}
		-\left|\zeta_1(n) - \zeta_2(n)\frac{\kappa_2(n)}{\kappa_1(n)} \right|.
	\end{align*}
If $\kappa_1(n)\neq\kappa_2(n)$, we can select $m>m(n)$ large enough such that the right-hand side of the above equation becomes larger than $1$.
	Therefore, \eqref{enq:thm:Uniq1} cannot be satisfied for $i=1$. 
	Thus $\kappa_1(n)=\kappa_2(n)$ for any $n\geq n_0$.
	
	\emph{Proof of \eqref{eqn:thm:Uniqzn}}:
	From \eqref{eqn:thm:Uniqkn} we obtain that for large enough $n$ we have that $\kappa_1(n)=\kappa_2(n)$.
	Hence, for large enough $n$, we can write
	\begin{align*}
		& \sup_{m\geq m(n)}
		\left|\frac{\mathcal{H}_{m,n}(X)}{\kappa_1(n)} - \log m - \zeta_1(n) \right|
		+\sup_{m\geq m(n)}
		\left|\frac{\mathcal{H}_{m,n}(X)}{\kappa_2(n)} - \log m - \zeta_2(n) \right| \\
	&\qquad= \sup_{m\geq m(n)}
		\left|\frac{\mathcal{H}_{m,n}(X)}{\kappa_1(n)} - \log m - \zeta_1(n) \right|
		+\sup_{m\geq m(n)}
		\left|-\frac{\mathcal{H}_{m,n}(X)}{\kappa_2(n)} + \log m + \zeta_2(n) \right| \\
		& \qquad \geq \sup_{m\geq m(n)}
		\left| \zeta_1(n) - \zeta_2(n) \right|
		=\left| \zeta_1(n) - \zeta_2(n) \right|,
	\end{align*}
	where $m(n):=\max\{m_1(n) , m_2(n)\}$.
	Therefore, according to the assumption, we obtain that
	\begin{align*}
		\lim_{n\to\infty}
		\left| \zeta_1(n) - \zeta_2(n) \right|
		\leq\; & \lim_{n\to\infty} \sup_{m\geq m(n)}
		\left|\frac{\mathcal{H}_{m,n}(X)}{\kappa_1(n)} - \log m - \zeta_1(n) \right| \\
		&+\lim_{n\to\infty} \sup_{m\geq m(n)}
		\left|\frac{\mathcal{H}_{m,n}(X)}{\kappa_2(n)} - \log m - \zeta_2(n) \right| \\
		=\;&0.
	\end{align*}
	Thus, the statement was proved.
	
	\textbf{Proof of the second case}:
	For this case, similarly, we can write
	\begin{align*}
		\lim_{n\to\infty}
		\left| \zeta_1(n) - \zeta_2(n) \right|
		\leq\; & \lim_{n\to\infty} \sup_{m\geq m(n)}
		\left|\mathcal{H}_{m,n}(X) - \zeta_1(n) \right| \\
		&+\lim_{n\to\infty} \sup_{m\geq m(n)}
		\left|\mathcal{H}_{m,n}(X) - \zeta_2(n) \right| \\
		=\;&0.
	\end{align*}
	Therefore the statement is proved.
\qed

%%%%%%%%%%%%%%%%%%%%%%%%%%
% Proof of Theorem: Entropy Dimension Generalization
\subsection {Proof of Theorem \ref{thm:GenerelEntropyDim}} \label{prf:thm:GenerelEntropyDim}	From \eqref{def.CompCrit} we obtain that
	\begin{equation*}
		\mathcal{H}_{m,n}(X)
		= n\H{\quant{X_1^{(n)}}{m}}.
	\end{equation*}
	Now, we are going to use Theorem \ref{thm.DCEntDim} in order to show that
	\begin{equation*}
		n\H{\quant{X_1^{(n)}}{m}}
		= n d_n \log m + n h_n + n e_{m,n},
	\end{equation*}
	where $d_n$ is the entropy dimension of $X_1^{(n)}$ and $e_{m,n}$ vanishes for every fixed $n$ when $m$ tends to $\infty$.
	To this end, first, we prove the conditions of Theorem \ref{thm.DCEntDim} which are the same for all types of discrete, continuous and discrete-continuous:
	\begin{enumerate}
		\item \label{eqn:X1nD,C,DC}
		$X_1^{(n)}$ is a discrete, continuous, or discrete-continuous random variable, and has the same type as $X_0$,
		\item \label{eqn:H[X1n]<inf}
		$\H{\quant{X_1^{(n)}}{1}}$ is finite.
	\end{enumerate}
	In the next step, we prove the exclusive conditions for each case.
	By choosing $\kappa(n)=1$ in the discrete or $\kappa(n)=n d_n$ in other two cases, and $\zeta(n)=n h_n / \kappa(n)$, we only need to find $m(n)$ such that
	\begin{equation*}
		\lim_{n\to\infty} \sup_{m: m\geq m(n)}{\left|\frac{e_{m,n}}{\kappa(n)}\right|}=0.
	\end{equation*}
Since $e_{m,n}$ vanishes as $m$ tends to $\infty$ for any fixed $n$, there is some $m(n)$ such that for any $m\geq m(n)$ we have $$\left|\frac{e_{m,n}}{\kappa(n)}\right|\leq \frac 1n.$$
This proves the existence of the function $m(n)$.

	It remains to prove items \ref{eqn:X1nD,C,DC} and \ref{eqn:H[X1n]<inf}, and then, find $\kappa(n)$ and $\zeta(n)$ for cases of $X_0$ being discrete, continuous, or discrete-continuous.

	\emph{Proof of item \ref{eqn:X1nD,C,DC}}:
Random variable $X_1^{(n)}$ is the time quantization with step size $1/n$ of a white L\'evy noise with parameters $\sigma$, $\mu$ and $V(a)$.
	It can be alternatively expressed as the time quantization with unit step size of a new white L\'evy noise with scaled parameters $\sigma/\sqrt{n}$, $\mu/n$ and $V(a)/n$.
	Theorem \ref{thm:InfDivRVType} gives conditions for $X_0$ (time quantization with unit step size) being discrete, continuous, or discrete-continuous random variable in terms of $\sigma$, $\mu$, and $V(a)$. These conditions for $\sigma$, $\mu$, and $V(a)$ are equivalent for the conditions for $\sigma/\sqrt{n}$, $\mu/n$, and $V(a)/n$.
	Therefore, random variable $X_1^{(n)}$ has the same type as $X_0$. 

	\emph{Proof of item \ref{eqn:H[X1n]<inf}}:
	From Definition \ref{def.WhiteNoise}, observe that $X_0$ can be written as the sum of $n$ iid random variables $X_i^{(n)}$ for $i\in\{1,\cdots,n\}$.
	Hence, we have that
	\begin{equation*}
		\quant{\sum_{i=1}^n{X_i^{(n)}}}{1}
		= \quant{X_0}{1}.
	\end{equation*}
	Observe that for any quantization of sum of variables we can write
	\begin{equation*}
		\quant{\sum_{i=1}^n{X_i^{(n)}}}{1}
		= \sum_{i=1}^n{\quant{X_i^{(n)}}{1}} + E_n,
	\end{equation*}
	where $E_n$ is a random variable taking values from $\{0,\cdots,n-1\}$.
	Observe that
	\begin{align*}
		\H{\quant{X_1^{(n)}}{1}}
		=\; & \H{\sum_{i=1}^n\quant{X_i^{(n)}}{1} \Big| \quant{X_2^{(n)}}{1}, \cdots, \quant{X_n^{(n)}}{1}} \\
		\leq\; & \H{\sum_{i=1}^n{\quant{X_i^{(n)}}{1}}} 
		= \H{\quant{X_0}{1} - E_n} \\
		\leq\; & \H{\quant{X_0}{1}, E_n}
		\leq \H{\quant{X_0}{1}} + \log n
	\end{align*}
	Therefore, $\H{\quant{X_1^{(n)}}{1}}<\infty$.
	
	Now, we aim to check the exclusive constraints and find $\kappa(n)$ and $\zeta(n)$ for each case.

	% Continuous
	\textbf{Case 1: $X_0$ is continuous and at least one of the assumptions in part (2) of Theorem \ref{thm:InfDivRVType} is satisfied}:
	For the first claim, as it was proved in item \ref{eqn:X1nD,C,DC} at the beginning of the proof, observe that $X_1^{(n)}$ is continuous.
	Therefore, according to Theorem \ref{thm.DCEntDim}, $d_n$ is $1$ for every $n$, and
	\begin{equation*}
		\zeta(n) := h_n = \h{X_1^{(n)}},
	\end{equation*}
	provided that $\h{X_1^{(n)}}$ is well-defined, which is true based on the assumption.

	Based on Definitions \ref{def.WhiteNoise} and \ref{def.TimeQuant}, we can write
	\begin{equation*}
		X_1^{(n-1)}
		=X_1^{(n)}+E,
	\end{equation*}
	where $E$ and $X_1^{(n)}$ are independent and
	\begin{equation*}
		\ue^{\uj \omega E}
		= \exp\left(\omega \int_{\frac{1}{n}}^{\frac{1}{n-1}}{X(t) \ud t} \right).
	\end{equation*}
Now observe that for any two independent continuous random variables $X,Y$ we have that
	\begin{equation*}
		\h{X+Y}
		\geq \h{X+Y|Y}
		=\h{X}.
	\end{equation*}
	Therefore,
	\begin{equation*}
		\zeta(n)
		\leq \zeta(n-1).
	\end{equation*}
As a result, $\zeta(n)$ is a non-increasing function.

	% Discrete-Continuous
	\textbf{Case 2: $X_0$ is discrete-continuous}:
	From item \ref{eqn:X1nD,C,DC} from the beginning of the proof,  $X_1^{(n)}$ is discrete-continuous if $X_0$ is discrete-continuous.
	From Lemma \ref{lmm:finite v->Poisson} we have that $X_1^{(n)}$ is sum of a constant $\mu'/n:=(\mu/n)-(\lambda/n) \int_{[-1,1]\setminus\{0\}}{a \ud V(a)}$ and the integral over $[0, \frac{1}{n}]$ of an impulsive Poisson white noise with rate $\lambda$ and  discrete-continuous amplitude density $p_A(x)$ where
	\begin{align}
		\lambda := \int_{\bR\setminus\{0\}}{\ud V(a)} < \infty,
		\qquad A \sim p_A(x) := \frac{1}{\lambda}\tfrac{\ud}{\ud x}V(x). \label{eqn:def-lambda-use}
	\end{align}
	The reason that $A$ is discrete-continuous is that, according to Theorem \ref{thm:InfDivRVType}, when $X_0$ is discrete-continuous, the function $V(a)$ must be bounded and discrete-continuous \emph{i.e.} $V_{cs}(a)\equiv 0$.
	Note that, according to Remark \ref{rmk:X0DC_AcAD}, the discrete and the continuous parts of $A$ are $A_D$ and $A_c$ with distributions $p_{A_D}$ and $p_{A_c}$, respectively.

	The distribution of the integral of an impulsive Poisson white noise over $[0,\frac{1}{n}]$ is given in  Lemma \ref{lmm.PoissonQuantDist}.
	Consider a Poisson random variable $Q_1$ with rate $\lambda/n$.
	Then, utilizing Lemma \ref{lmm.PoissonQuantDist}, we can define the following conditional distribution:
	\begin{equation} \label{eqn:PoissonTMP}
		p_{X_1^{(n)}|Q_1}(x|k)
		= p_A^{*k}(x)
		:=
		\begin{cases}
			\big( \overbrace{p_A*\cdots*p_A}^k \big)(x) & k\in\bN \\
			\delta(x) & k=0
		\end{cases}.
	\end{equation}
	According to Theorem \ref{thm.DCEntDim}, we have that
	\begin{equation*}
		d_n = 1 - \Pr{X_1^{(n)} \text{ is discrete}}.
	\end{equation*}
	Therefore, by the definition of $Q$ we can write
	\begin{align*}
		\Pr{X_1^{(n)} \text{ is discrete}}
		=& \sum_{k=0}^\infty{p_{Q_1}(k) \Pr{X_1^{(n)} \text{ is discrete} \Big| Q_1=k}} \\
		=& \ue^{-\frac{\lambda}{n}} \times 1
		+ \sum_{k=1}^\infty{p_{Q_1}(k) \Pr{\sum_{i=1}^k{A_i} \text{ is discrete}}},
	\end{align*}
	where $A_i$ are iid random variable with probability density $p_A(x)$.
	Utilizing the result of Lemma \ref{lmm:DC+DC} that
	\begin{equation*}
		\Pr{X+Y \text{ is discrete}}
		=\Pr{X \text{ is discrete}}\Pr{Y \text{ is discrete}},
	\end{equation*}	
	where $X,Y$ are independent discrete-continuous random variables, we have that
	\begin{align}
		d_n = & 1 - \Pr{X_1^{(n)} \text{ is discrete}} \nonumber\\
		=& 1- \ue^{-\frac{\lambda}{n}}
		- \ue^{-\frac{\lambda}{n}}
		\sum_{k=2}^\infty {\frac{(\lambda/n)^k}{k!}\Pr{p_{A} \text{ is discrete}}^k} \nonumber\\
		=& 1- \ue^{-\frac{\lambda}{n}}
		-\ue^{-\frac{\lambda}{n}} (\ue^{\frac{\lambda}{n}\Pr{A\text{ is discrete}}}-1)\nonumber \\
		=& 1-\exp\left[-\frac{\lambda}{n}\left(1-\Pr{A \text{ is discrete}}\right)\right].\label{last-equation-added-am}
	\end{align}
	Therefore, the problem of $\kappa(n)$ is settled.
	
	In order to find $\zeta(n)$, we apply Theorem \ref{thm.DCEntDim} to $X_1^{(n)}$.
	Let $X_c^{(n)}$ and $X_D^{(n)}$ be the continuous and the discrete part of $X_1^{(n)}$, respectively.
	Then, to apply Theorem \ref{thm.DCEntDim}, we need to show that $|\h{X_c^{(n)}}|$ and $\H{X_D^{(n)}}$ are well-defined.
	Then, Theorem \ref{thm.DCEntDim} implies that 
	\begin{equation*}
		\zeta(n) = \h{X_c^{(n)}}
		+ \frac{1-d_n}{d_n} \H{X_D^{(n)}}
		+ \frac{d_n\log\frac{1}{d_n} + (1-d_n) \log\frac{1}{1-d_n}}{d_n}.
	\end{equation*}
	Since $d_n$ given in \eqref{last-equation-added-am} vanishes as $n$ goes to infinity, to show the claim of the theorem for this part, it suffices to prove that 
	\begin{align}
		&|\h{X_c^{(n)}}| \leq c_3 < \infty, \label{eqn:mainthm:hXcn<inf} \\
		&\frac{\H{X_D^{(n)}}}{d_n} \leq c_4 \log n, \label{eqn:mainthm:HXDn<inf} \\
		c_5\log n
		\leq& \log\frac{1}{d_n}+\frac{1}{d_n}\log\frac{1}{1-d_n}
		\leq c_6 \log n. \label{eqn:mainthm:H2d/d<inf}
	\end{align}
	for some constants $0 \leq c_3 , c_4 < \infty$ and $0 < c_5 \leq c_6 <\infty$ not depending on $n$.
	Observe that \eqref{eqn:mainthm:hXcn<inf}  and \eqref{eqn:mainthm:HXDn<inf} imply that $|\h{X_c^{(n)}}|$ and $\H{X_D^{(n)}}$ are well-defined.
	Thus, the conditions of Theorem \ref{thm.DCEntDim} are satisfied if we can show \eqref{eqn:mainthm:hXcn<inf}-\eqref{eqn:mainthm:H2d/d<inf}.
	It remains to show \eqref{eqn:mainthm:hXcn<inf}-\eqref{eqn:mainthm:H2d/d<inf}.

	\emph{Proof of \eqref{eqn:mainthm:hXcn<inf}}:
	From  Definition \ref{def.WhiteNoise}, for $X_0 := \int_0^1{X(t) \ud t}$, we have that $X_0=\sum_{i=1}^n{X_i^{(n)}}$, where $X_i^{(n)}$ are \emph{i.i.d.}.
	We denote the discrete and continuous part of $X_0$ and $X_1^{(n)}$ with $X_{0,D}, X_{0,c}$ and $X_{1,D}^{(n)}, X_{1,c}^{(n)}$, respectively.
	Therefore, according to Corollary \ref{cor:h(X+Y)c} we have that
	\begin{equation*}
		\h{X_{0,c}}
		\geq \h{X_{1,c}^{(n)}}.
	\end{equation*}
	Hence, if $\h{X_{0,c}}<\infty$, then $\h{X_{1,c}^{(n)}}<\infty$ uniformly on $n$.
	Note that based on the assumption, $\h{X_{1,c}^{(n)}}$ well-defined. 
	
	Now, it only remains to prove $\h{X_{1,c}^{(n)}}>-\infty$.
	To this end, similar to what we did in \eqref{eqn:PoissonTMP}.
	Let
	\begin{align}
		S_{k}=\sum_{i=1}^k{A_i} \label{def-S-k-am}
	\end{align}
	where $A_i$ are iid random variable with probability density $p_A(x)$.
	Let $(S_{k})_c$ be the continuous part of $S_k$. Utilizing the definitions given in \eqref{eqn:def-lambda-use}, Lemma \ref{lmm.PoissonQuantDist} and Lemma \ref{lmm:DC+DC}, one can define random variable $Q_2$ as follows:
	\begin{equation*}
		p_{Q_2}(k)
		= \frac{1}{\ue^{\lambda/n}-\ue^{\Pr{A \text{ is discrete}}\lambda /n}}\frac{(\lambda/n)^k}{k!}(1-\Pr{A \text{ is discrete}}^k),
		\qquad \forall k\in\bN,
	\end{equation*}
	\begin{equation*}
		p_{X_{1,c}^{(n)}|Q_2}(x|k)
		= p_{(S_{k})_c}(x),
		\qquad \forall k\in\bN.
	\end{equation*}
	Hence, we can write
	\begin{equation*}
		\h{X_{1,c}^{(n)}}
		\geq \h{X_{1,c}^{(n)} \big| Q_2}
		= \sum_{k=1}^\infty
		{p_{Q_2}(k) \h{(S_{k})_c}}.
	\end{equation*}
	Now, from Corollary \ref{cor:h(X+Y)c}, we obtain that
	\begin{equation*}
		\h{X_{1,c}^{(n)}}
		\geq \sum_{k=1}^\infty
		{p_{Q_2}(k) \h{A_c}}
		= \h{A_c},
	\end{equation*}
	where $A_c$ is the continuous part of $A$. 
	Therefore, if $\h{A_c}>-\infty$, the statement is proved.
	
	\emph{Proof of \eqref{eqn:mainthm:HXDn<inf}}:
	Similar to what we did above, utilizing the definitions given in \eqref{eqn:def-lambda-use}, Lemma \ref{lmm.PoissonQuantDist} and Lemma \ref{lmm:DC+DC}, one can define random variable $Q_3$ as following:
	\begin{equation*}
		p_{Q_3}(k)
		=\frac{1}{\ue^{-\frac{\lambda}{n}}+\ue^{\frac{\lambda}{n}\Pr{A \text{ is discrete}}}-1}
		\begin{cases}
			\frac{\left(\frac{\lambda}{n}\Pr{A \text{ is discrete}}\right)^k}{k!} & k\in\bN \\
			\ue^{-\frac{\lambda}{n}} & k=0
		\end{cases},
	\end{equation*}
	\begin{equation*}
		p_{X_{1,D}^{(n)}|Q_3}(x|k)
		=
		\begin{cases}
			p_{(S_{k})_D}(x) & k\in\bN \\
			\delta(x) & k=0
		\end{cases}
	\end{equation*}
	where $(S_{k})_D$ is the discrete part of $S_k$ defined in \eqref{def-S-k-am} for $A_i$ begin iid random variable with probability density $p_A(x)$.  If we let $(A_i)_D$ to be the discrete part of $A_i$, Lemma \ref{lmm.PoissonQuantDist} shows that
	\begin{align}
		(S_{k})_D=\sum_{i=1}^k (A_i)_D. \label{s_kd-ai-d}
	\end{align}

	We can write that
	\begin{align*}
		\H{X_D^{(n)}}
		\leq & \H{X_D^{(n)}, Q_3} \\
		= & \H{X_D^{(n)} \Big| Q_3} + \H{Q_3}.
	\end{align*}
	Now, in order to prove \eqref{eqn:mainthm:HXDn<inf}, it suffices to show that
	\begin{align*}
		& n\H{X_D^{(n)} \Big| Q_3} \leq c_7 < \infty, \\
		& n\H{Q_3} \leq c_8 \log n,
	\end{align*}
	for some constants $c_7$ and $c_8$. 

	For the first inequality, since
	\begin{equation*}
		\H{X+Y}
		\leq\H{X,Y}
		\leq\H{X}+\H{Y},
	\end{equation*}
	using \eqref{s_kd-ai-d} we have
	\begin{align*}
		\H{X_{1,D}^{(n)} | Q}&=\sum_{k=1}^\infty
		{p_{Q_3}(k) \H{\sum_{i=1}^k (A_i)_D}}
\\&
		\leq \sum_{k=1}^\infty
		{p_{Q_3}(k) k\H{A_D}}
\\& =\mathbb{E}[Q_3]\H{A_D}.
	\end{align*}
	It can be verified that $n\E{Q_3}$ converges to $\lambda (1-\Pr{A \text{ is discrete}})$ as $n$ tends to infinity. As a result, $n \H{X_{1,D}^{(n)} | Q}$ is bounded uniformly on $n$.

	For the second inequality, note that among all the discrete distributions defined on $\{0,1,\cdots\}$ with a specified mean $\mu$, we have that\cite[Theorem 13.5.4]{Cover06}:
	\begin{equation*}
		\H{X} \leq \log(1+\mu) + \mu\log\left(1+\frac{1}{\mu}\right).
	\end{equation*}
	Therefore,
	\begin{equation*}
		\H{Q_3} \leq \log(1+\E{Q_3}) + \E{Q_3}\log\left(1+\frac{1}{\E{Q_3}}\right).
	\end{equation*}
	It  suffices to prove that
	\begin{equation} \label{eqn:EQ<1/n}
		\frac{c_9}{n}
		\leq \E{Q_3}
		\leq \frac{c_{10}}{n},
	\end{equation}
	for some positive constants $c_9$ and $c_{10}$. The reason is that the above equation shows existence of $n_0$ such that for $n>n_0$,
	\begin{equation*}
		1+\frac{1}{\E{Q_3}}
		\leq \left(\frac{n}{c_9}\right)^2.
	\end{equation*}
	Then, we have that
	\begin{align*}
		n\H{Q_3}
		&\leq n\E{Q_3} + n\E{Q_3} \log\left(1+\frac{1}{\E{Q_3}}\right)
\\&\leq c_{10} + 2 c_{10} \log(n) - 2 c_{10} \log(c_9).
	\end{align*}
It only remains to prove \eqref{eqn:EQ<1/n}.
	To this end, we calculate $\E{Q_3}$ as following:
	\begin{align*}
		\E{Q_3}
		= & \frac{1}{\ue^{-\frac{\lambda}{n}}+\ue^{\frac{\lambda}{n}\Pr{A \text{ is discrete}}}-1}
		\sum_{k=1}^\infty
		{k \frac{\left(\frac{\lambda}{n}\Pr{A \text{ is discrete}}\right)^k}{k!}} \\
		= & \frac{\frac{\lambda}{n}\Pr{A \text{ is discrete}}}{\ue^{-\frac{\lambda}{n}}+\ue^{\frac{\lambda}{n}\Pr{A \text{ is discrete}}}-1}
		\sum_{k=0}^\infty
		{\frac{\left(\frac{\lambda}{n}\Pr{A \text{ is discrete}}\right)^k}{k!}} \\
		= & \frac{\Pr{A \text{ is discrete}}\frac{\lambda}{n} \ue^{\frac{\lambda}{n}\Pr{A \text{ is discrete}}}}
		{\ue^{-\frac{\lambda}{n}}+1-\ue^{-\Pr{A \text{ is discrete}}\frac{\lambda}{n}}}
	\end{align*}
	As a result, for large enough $n$ there exists global constant $0<c_9 \leq c_{10} <\infty$ such that \eqref{eqn:EQ<1/n} is satisfied.

	\emph{Proof of \eqref{eqn:mainthm:H2d/d<inf}}:
	It can be proved that
	\begin{equation*}
		\lim_{n\to\infty}
		{\log\frac{1}{d_n}+\frac{1}{d_n}\log\frac{1}{1-d_n}
		-\log\frac{n}{\lambda(1-\Pr{A \text{ is discrete}})}}
		= 1.
	\end{equation*}
	Therefore, \eqref{eqn:mainthm:H2d/d<inf} is proved.

	% Discrete
	\textbf{Case 3: $X_0$ is discrete}:
	As it was proved in item \ref{eqn:X1nD,C,DC}, observe that $X_1^{(n)}$ is discrete if $X_0$ is discrete.
	Therefore, from Theorem \ref{thm.DCEntDim}, $d_n$ is $0$ for every $n$, and there is no term of $\log m$ in $\mathcal{H}_{m,n}(X)$. However, regardless of the results of Theorem \ref{thm.DCEntDim}, we choose $\kappa(n)=1$, and,
	\begin{equation*}
		\zeta(n) = n \H{X_1^{(n)}}.
	\end{equation*}
	From Theorem \ref{thm:InfDivRVType}, we obtain that if $X_0$ is discrete, then $V(a)=V_{d}(a)$ and is bounded.
	Hence, according to Lemma \ref{lmm:finite v->Poisson}, $X_1^{(n)}$ is sum of the  integral of a Poisson impulsive white noise, plus a constant.	Therefore, the calculation of finding $\zeta(n)$ is completely similar to \eqref{eqn:mainthm:HXDn<inf} in the discrete-continuous case.
	Thus, the theorem is proved.
\qed

%%%%%%%%%%%%%%%%%%%%%%%%%%%%%%
% Proof of Theorem: Stable White Noises Compressibility Criterion
\subsection{Proof of Theorem \ref{thm.AlphaCompCrit}} \label{subsec:ProofAlphaCompCrit}
	From Lemma \ref{lmm:StableXnDist} we obtain that
	\begin{equation*}
		X_1^{(n)} \stackrel{d}{=} \frac{X_0-b_n}{\sqrt[\alpha]{n}},
	\end{equation*}
	where $b_n$ is defined in Lemma \ref{lmm:StableXnDist}.
	Hence,
	\begin{equation} \label{eqn:IndepRV}
		\H{\quant{X_1^{(n)}}{m}}
		=\H{\quant{\frac{X_0-b_n}{\sqrt[\alpha]{n}}}{m}}
	\end{equation}
	Therefore, from Lemma \ref{lmm.H([aX])} we obtain that
	\begin{equation} \label{eq.AlphaScaling}
		\H{\quant{X_1^{(n)}}{m}}
		=\H{\quant{X_0-b_n}{\frac{m}{\sqrt[\alpha]{n}}}}.
	\end{equation}
	Lemma \ref{lmm.AlphaStableRenyiCondition} shows that the distribution of $X_0$ satisfies the properties of Lemma \ref{lmm.DifEntShft}.
	Thus, from Lemma \ref{lmm.DifEntShft} we obtain that when $m\geq m(n)$ and $\frac{m(n)}{\sqrt[\alpha]{n}}$ tends to $\infty$, it can be shown that
	\begin{equation} \label{eq.EntDim}
		\lim_{n\to\infty}\sup_{m: m\geq m(n)}
		{\left|\H{\quant{X_0-b_n}{\frac{m}{\sqrt[\alpha]{n}}}}
		- \log \frac{m}{\sqrt[\alpha]{n}}
		- \h{X_0} \right|}
		= 0.
	\end{equation}
	Statement of the theorem follows from \eqref{eqn:IndepRV}-\eqref{eq.EntDim}.
\qed

%%%%%%%%%%%%%%%%%%%%%%%%%%%%%%%%
% Proof of Theorem: The limit of the quantized Poisson process entropy
\subsection{Proof of Theorem \ref{thm.PoissonCompCrit}} \label{subsec:ProofPoissonCompCrit}
From Lemma \ref{lmm:finite v->Poisson} we obtain that, for an impulsive Poisson white noise $X(t)$, $X_0:=\int_0^1{X(t) \ud t}$ is a discrete-continuous random variable; as a result, utilizing Theorem \ref{thm:InfDivRVType}, $X_1^{(n)}$ is also discrete-continuous.
In order to find the probability of $X_1^{(n)}$ being continouous, observe that, in \eqref{eq.pYn} in Lemma \ref{lmm.PoissonQuantDist}, $p_{A_n}$ is an absolutely continuous distribution due to \eqref{eq:pAn}, the fact that $p_A$ is absolutely continuous, and Lemma \ref{lmm:DC+DC}.
Hence, from Theorem \ref{thm:GenerelEntropyDim}, we obtain that $\Pr{X_1^{(n)} \text{ is discrete}=\ue^{-\frac{\lambda}{n}}}$, thereby
\begin{equation} \label{eqn:k(n) Poisson}
	\kappa(n) = n \left(1-\ue^{-\frac{\lambda}{n}}\right).
\end{equation}

In order to find $\zeta(n)$ and $m(n)$ we need to write $\mathcal{H}_{m,n}(X)$ as following.
Let $P_{X_1^{(n)};m}$ be the pmf of $\quant{X_1^{(n)}}{m}$ as in Definition \ref{def.AmpQuant}.
Therefore,
\if@twocolumn
	\begin{align*}
	n \H{\quant{X_1^{(n)}}{m}}
	&= n \sum_{i\in\mathbb{Z}}{P_{X_1^{(n)};m}[i]\log\frac{1}{P_{X_1^{(n)};m}[i]}} \\
	&= n P_{X_1^{(n)};m}[0]\log\frac{1}{P_{X_1^{(n)};m}[0]} \\
	&\quad+n\sum_{i\in\mathbb{Z}\setminus\lbrace0\rbrace}
	{P_{X_1^{(n)};m}[i]\log\frac{1}{P_{X_1^{(n)};m}[i]}}.
	\end{align*}
\else
	\begin{align*}
	n \H{\quant{X_1^{(n)}}{m}}
	&= n \sum_{i\in\mathbb{Z}}{P_{X_1^{(n)};m}[i]\log\frac{1}{P_{X_1^{(n)};m}[i]}} \\
	&= n P_{X_1^{(n)};m}[0]\log\frac{1}{P_{X_1^{(n)};m}[0]} 
	+n\sum_{i\in\mathbb{Z}\setminus\lbrace0\rbrace}
	{P_{X_1^{(n)};m}[i]\log\frac{1}{P_{X_1^{(n)};m}[i]}}.
	\end{align*}
\fi
Proving the following limits will imply the first statement of the theorem:
\begin{equation} \label{eq.kPoissonQuantH1}
	\lim_{m,n\to\infty}
	{\frac{n}{\kappa(n)} P_{X_1^{(n)};m}[0]\log \frac{1}{P_{X_1^{(n)};m}[0]}}
	=1,
\end{equation}
\begin{align}
	&\lim_{n\to \infty} \sup_{m:m\geq m(n)}
	{\left|\frac{n}{\kappa(n)}\sum_{i\in\mathbb{Z}\setminus\lbrace 0\rbrace}
	{P_{X_1^{(n)};m}[i]\log\frac{1}{P_{X_1^{(n)};m}[i]}}-\log(mn)
	-\h{A}+\log\lambda\right|}
	=0, \label{eq.kPoissonQuantH2}
\end{align}
where $ \h{A}$ is the differential entropy of random variable $A$.
Proving the above equalities, will imply the statement of the theorem.

\emph{Proof of \eqref{eq.kPoissonQuantH1}:}
Observe that from \eqref{eqn:k(n) Poisson} we have that
\begin{equation} \label{eqn:PoissonH2.1.1}
	\lim_{n\to\infty}{\kappa(n)}
	=\lim_{n\to\infty}{n\left(1-\ue^{-\frac{\lambda}{n}}\right)}
	= \lambda.
\end{equation}
Therefore, it is enough to only prove
\begin{equation} \label{eq.PoissonQuantH1}
	\lim_{m,n\to \infty}
	{n P_{X_1^{(n)};m}[0]\log \frac{1}{P_{X_1^{(n)};m}[0]}}
	=\lambda.
\end{equation}
From Lemma \ref{lmm.PoissonQuantDist}, we obtain that
\begin{equation}
	P_{X_1^{(n)};m}[0]=\ue^{-\frac{\lambda}{n}}+(1-\ue^{-\frac{\lambda}{n}}) P_{A_n;m}[0],\label{eq.PX1[0]}
\end{equation}
for some random variable $A_n$ with features mentioned in \eqref{eq:pAn}-\eqref{eqn:EAn->EA}.
For every $m$ we have
\begin{equation*}
	\int_{-\frac{1}{2m}}^{\frac{1}{2m}}{|p_{A_n}(x) - p_A(x)| \ud x}
	\leq \int_{\mathbb{R}}{|p_{A_n}(x) - p_A(x)| \ud x}.
\end{equation*}
Hence, 
\if@twocolumn
	\begin{align*}
	\int_{-\frac{1}{2m}}^{\frac{1}{2m}}{p_{A_n}(x) \ud x}
	\leq& \int_{-\frac{1}{2m}}^{\frac{1}{2m}} \hspace{-1mm} p_A(x) \ud x
	+\hspace{-1mm}  \int_{\mathbb{R}}{|p_{A_n}(x) - p_A(x)| \ud x}\\
	\leq& \frac{M}{m}+2\frac{\ue^{\frac{\lambda}{n}}-\frac{\lambda}{n}-1}{\ue^{\frac{\lambda}{n}}-1}, \EQnum\label{eqn:int pAn[0]}
	\end{align*}
\else
	\begin{align*}
	\int_{-\frac{1}{2m}}^{\frac{1}{2m}}{p_{A_n}(x) \ud x}
	\leq \int_{-\frac{1}{2m}}^{\frac{1}{2m}}{p_A(x) \ud x}
	+\int_{\mathbb{R}}{|p_{A_n}(x) - p_A(x)| \ud x} 
	\leq \frac{\ell}{m}+2\frac{\ue^{\frac{\lambda}{n}}-\frac{\lambda}{n}-1}{\ue^{\frac{\lambda}{n}}-1}, \EQnum\label{eqn:int pAn[0]}
	\end{align*}
\fi
where \eqref{eqn:int pAn[0]} is true due to \eqref{eq.VarConv}, and the fact that $p_A(x) \leq M$ almost everywhere.
Thus, we obtain that
\begin{equation*}
	\lim_{m,n\to\infty}
	{\int_{-\frac{1}{2m}}^{\frac{1}{2m}}{p_{A_n}(x)\ud x}}
	=0.
\end{equation*}
Therefore, from the definition of $P_{A_n;m}[0]$ defined in Definition \ref{def.AmpQuant}, we obtain that 
\begin{equation} \label{eq.PA_n[0]=0}
	\lim_{m,n\to\infty}{P_{A_n;m}[0]} = 0.
\end{equation}
Therefore, from \eqref{eq.PX1[0]}, we can write that
\begin{equation} \label{eq.PX1[0]=1}
	\lim_{m,n\to \infty}{P_{X_1^{(n)};m}[0]} = 1.
\end{equation}
Hence, if we prove that
\begin{equation} \label{eq.PoissonQuantH1.2}
	\lim_{m,n\to\infty}
	{n \log \frac{1}{P_{X_1^{(n)};m}[0]}}=\lambda,
\end{equation}
we can conclude from \eqref{eq.PX1[0]=1} that
\if@twocolumn
	\begin{align*}
	&\lim_{m,n\to \infty}
	{n P_{X_1^{(n)};m}[0]\log \frac{1}{P_{X_1^{(n)};m}[0]}} \\
	&\qquad= \lim_{m,n\to \infty}
	{n \log \frac{1}{P_{X_1^{(n)};m}[0]}}=\lambda.
	\end{align*}
\else
\begin{align*}
	&\lim_{m,n\to \infty}
	{n P_{X_1^{(n)};m}[0]\log \frac{1}{P_{X_1^{(n)};m}[0]}}
	= \lim_{m,n\to \infty}
	{n \log \frac{1}{P_{X_1^{(n)};m}[0]}}=\lambda.
\end{align*}
\fi
This will complete the proof of \eqref{eq.PoissonQuantH1}.
Thus, it only remains to prove \eqref{eq.PoissonQuantH1.2}.
Again, by substituting the value of $P_{X_1;m}[0]$ from \eqref{eq.PX1[0]}, we have
\if@twocolumn
	\begin{align*}
	&\lim_{m,n\to \infty}
	{n \log \frac{1}{P_{X_1^{(n)};m}[0]}} \\
	&\qquad=-\lim_{m,n\to \infty}
	{n\log{\left(\ue^{-\frac{\lambda}{n}}+(1-\ue^{-\frac{\lambda}{n}}) P_{A_n;m}[0] \right)}} \\
	&\qquad=-\lim_{m,n\to \infty}
	{n\log{\left(1+(\ue^{-\frac{\lambda}{n}}-1)\left(1-P_{A_n;m}[0]\right) \right)}} \\
	&\qquad=-\lim_{m,n\to \infty}
	{n(\ue^{-\frac{\lambda}{n}}-1)\left(1-P_{A_n;m}[0]\right)} \EQnum \label{PoissonQuantH1.2.1} \\
	&\qquad=-\lim_{m,n\to \infty}
	{n(\ue^{-\frac{\lambda}{n}}-1)} = \lambda. \EQnum \label{PoissonQuantH1.2.2}
	\end{align*}
\else
	\begin{align*}
	\lim_{m,n\to \infty}
	{n \log \frac{1}{P_{X_1^{(n)};m}[0]}}
	=&-\lim_{m,n\to \infty}
	{n\log{\left(\ue^{-\frac{\lambda}{n}}+(1-\ue^{-\frac{\lambda}{n}}) P_{A_n;m}[0] \right)}} \\
	=&-\lim_{m,n\to \infty}
	{n\log{\left(1+(\ue^{-\frac{\lambda}{n}}-1)\left(1-P_{A_n;m}[0]\right) \right)}} \\
	=&-\lim_{m,n\to \infty}
	{n(\ue^{-\frac{\lambda}{n}}-1)\left(1-P_{A_n;m}[0]\right)} \EQnum \label{PoissonQuantH1.2.1} \\
	=&-\lim_{m,n\to \infty}
	{n(\ue^{-\frac{\lambda}{n}}-1)} = \lambda. \EQnum \label{PoissonQuantH1.2.2}
	\end{align*}
\fi
where \eqref{PoissonQuantH1.2.1} follows from Taylor series of $\log (1+x)$ near $0$ for $x=(\ue^{-{\lambda}/{n}}-1)\left(1-P_{A_n;m}[0]\right)$; observe that $\ue^{-{\lambda}/{n}}-1$ is close to zero for large values of $n$, and \eqref{PoissonQuantH1.2.2} is obtained from \eqref{eq.PA_n[0]=0}.

\emph{Proof of \eqref{eq.kPoissonQuantH2}:}
From Lemma \ref{lmm.PoissonQuantDist} we obtain that
\begin{equation*}
	P_{X_1^{(n)};m}[i] = \left(1-\ue^{-\frac{\lambda}{n}}\right)P_{A_n;m}[i],
	\qquad \forall i \in\mathbb{Z}\setminus\lbrace 0 \rbrace.
\end{equation*}
By substituting the value of $P_{X_1^{(n)};m}[i]$ in terms of $P_{A_n;m}[i]$, we have
\if@twocolumn
	\begin{align*}
	&n \sum_{i\in\mathbb{Z}\setminus\lbrace0\rbrace}
	{P_{X_1^{(n)};m}[i]\log\frac{1}{P_{X_1^{(n)};m}[i]}} \\
	&\qquad=n \left(1-\ue^{-\frac{\lambda}{n}}\right) \left(1-P_{A_n;m}[0] \right)
	\log\frac{1}{1-\ue^{-\frac{\lambda}{n}}} \nonumber\\
	&\qquad\quad-n \left(1-\ue^{-\frac{\lambda}{n}}\right){P_{A_n;m}[0] \log\frac{1}{P_{A_n;m}[0]}}\nonumber\\
	&\qquad\quad+n \left(1-\ue^{-\frac{\lambda}{n}}\right)
	\sum_{i\in\mathbb{Z}}{P_{A_n;m}[i] \log\frac{1}{P_{A_n;m}[i]}}.
	\end{align*}
\else
	\begin{align*}
	n \sum_{i\in\mathbb{Z}\setminus\lbrace0\rbrace}
	{P_{X_1^{(n)};m}[i]\log\frac{1}{P_{X_1^{(n)};m}[i]}}
	=&n \left(1-\ue^{-\frac{\lambda}{n}}\right) \left(1-P_{A_n;m}[0] \right)
	\log\frac{1}{1-\ue^{-\frac{\lambda}{n}}} \nonumber\\
	&-n \left(1-\ue^{-\frac{\lambda}{n}}\right){P_{A_n;m}[0] \log\frac{1}{P_{A_n;m}[0]}}\nonumber\\
	&+n \left(1-\ue^{-\frac{\lambda}{n}}\right)
	\sum_{i\in\mathbb{Z}}{P_{A_n;m}[i] \log\frac{1}{P_{A_n;m}[i]}}.
	\end{align*}
\fi
Therefore, in order to prove \eqref{eq.kPoissonQuantH2}, utilizing \eqref{eqn:k(n) Poisson}, and it suffices to show that
\begin{equation} \label{eqn:kPoissonH2.1}
	\lim_{m,n\to \infty}
	{P_{A_n;m}[0] \log{P_{A_n;m}[0]}} = 0,
\end{equation}
\begin{equation} \label{eqn:kPoissonH2.2}
	\lim_{n\to\infty}\sup_{m:m\geq m(n)}
	{\left|\left(1-P_{A_n;m}[0] \right)
	\log\frac{1}{1-\ue^{-\frac{\lambda}{n}}}-\log n
	+\log\lambda\right|}
	= 0,
\end{equation}
\begin{equation} \label{eqn:PoissonH2.3.3}
	\lim_{m,n\to\infty}
	{\sum_{i\in\mathbb{Z}}{P_{A_n;m}[i] \log\frac{1}{P_{A_n;m}[i]}}-\log m}
	= \h{A}.
\end{equation}

\emph{Proof of \eqref{eqn:kPoissonH2.1}}:
It is obtained from \eqref{eq.PA_n[0]=0}.

\emph{Proof of \eqref{eqn:kPoissonH2.2}}:
Let us add and subtract the term $\left(1-P_{A_n;m}[0] \right)\log n$ inside of the absolute value of the left hand side of \eqref{eqn:kPoissonH2.2}.
Hence, we have
\begin{equation*}
	\left(1-P_{A_n;m}[0] \right) \log\tfrac{1}{1-\ue^{-\frac{\lambda}{n}}}-\log n
	=\left(1-P_{A_n;m}[0] \right)\log\tfrac{1}{n\left(1-\ue^{-\frac{\lambda}{n}}\right)} 
	-P_{A_n;m}[0] \log n,
\end{equation*}
From \eqref{eq.PA_n[0]=0}, and \eqref{eqn:PoissonH2.1.1}, we can write
\begin{equation*}
	\lim_{m,n\to\infty}
	{\left(1-P_{A_n;m}[0] \right)
	\log\tfrac{1}{n\left(1-\ue^{-\frac{\lambda}{n}}\right)}}
	= -\log \lambda.
\end{equation*}
Hence, we only need to show that
\begin{equation} \label{eqn:kPoissonH2.2.1}
	\lim_{n\to\infty}
	\sup_{m:m\geq m(n)}{P_{A_n;m}[0]\log n}
	= 0.
\end{equation}

For \eqref{eqn:kPoissonH2.2.1}, from \eqref{eqn:int pAn[0]}, we have that
\begin{equation*}
	0\leq P_{A_n;m}[0] \log n 
	\leq \left(\frac{\ell}{m}+2\frac{\ue^\frac{\lambda}{n}-\frac{\lambda}{n}-1}{\ue^\frac{\lambda}{n}-1} \right) \log n.
\end{equation*}
Thus, \eqref{eqn:kPoissonH2.2.1} is proved due to the choice of $m(n)$, which causes $(\log n)/{m}$ tends to $0$, and the fact that 
\begin{equation} \label{eqn:tmp321}
	\lim_{n\to\infty}
	{\frac{\ue^\frac{\lambda}{n}-\frac{\lambda}{n}-1}{\ue^\frac{\lambda}{n}-1}\log n} = 0.
\end{equation}

\emph{Proof of \eqref{eqn:PoissonH2.3.3}:}
Remember that in Definition \ref{def.AmpQuant}, for every arbitrary random variable $X$, a random variable $\widetilde{X}_m$ with an absolutely continuous distribution $q_{X;m}(x)$ was defined. Thus, corresponding to $A_n$, we can define random variables $\widetilde{A}_{nm}$ with an absolutely continuous distribution $q_{A_n;m}(x)$.
Observe that, from Lemma \ref{lmm:H([X])-ln m=h(q)},
\if@twocolumn
	\begin{align*}
	&\sum_{i\in\mathbb{Z}}{P_{A_n;m}[i] \log\frac{1}{P_{A_n;m}[i]}} \\
	&\qquad=\sum_{i\in\mathbb{Z}}
	{\frac{1}{m}q_{A_n;m}\left(\tfrac{i}{m}\right)
	\log\frac{m}{q_{A_n;m}\left(\frac{i}{m}\right)}}\\
	&\qquad=\int_\mathbb{R}{q_{A_n;m}(x)\log\frac{m}{q_{A_n;m}(x)} \ud x}\\
	&\qquad=\log m + \int_\mathbb{R}{q_{A_n;m}(x)\log\frac{1}{q_{A_n;m}(x)} \ud x} \\
	&\qquad=\log m + \h{\widetilde{A}_{nm}}, \EQnum\label{eqn:PoissonH2.3.4}
	\end{align*}
\else
	\begin{equation} \label{eqn:PoissonH2.3.4}
	\sum_{i\in\mathbb{Z}}{P_{A_n;m}[i] \log\frac{1}{P_{A_n;m}[i]}} 
	=\log m + \h{\widetilde{A}_{nm}},
	\end{equation}
\fi
where $\widetilde{A}_{nm}$ is a continuous random variable with pdf $q_{A_n;m}(x)$.
Similarly, for random variable $A$, using  Definition \ref{def.AmpQuant}, we can define continuous random variable $\widetilde{A}_{m}$ with pdf $q_{A;m}(x)$.
Again, from Lemma \ref{lmm:H([X])-ln m=h(q)}
\begin{align*}
	\H{\quant{A}{m}}:=&\sum_{i\in\mathbb{Z}}{P_{A;m}[i] \log\frac{1}{P_{A;m}[i]}}
	=\log m + \h{\widetilde{A}_{m}}\EQnum \label{eqn:PoissonH2.3.2}.
\end{align*}
From \eqref{eqn:PoissonH2.3.2} and Remark \ref{rem.CEntDim}, we obtain
\begin{equation} \label{eqn:hA=hAt}
	\h{A}=\lim_{m\to\infty}{\h{\widetilde{A}_m}}
\end{equation}
Hence, from \eqref{eqn:PoissonH2.3.4} and \eqref{eqn:hA=hAt}, we obtain that in order to prove \eqref{eqn:PoissonH2.3.3}, it suffices to show that
\begin{equation} \label{eqn:to-show-conv-ent}
	\lim_{m,n\to\infty}{\h{\widetilde{A}_{nm}}}
	=\lim_{m\to\infty}{\h{\widetilde{A}_m}}.
\end{equation}
To show this, we utilize Theorem \ref{thm.EntConv}. This theorem reduces convergence in differential entropy to convergence in total variation distance for a restricted class of distributions. In other words, to show \eqref{eqn:to-show-conv-ent}, it suffices to show 
\begin{equation} \label{eq.Poisson|qAn-qA|=0}
	\lim_{m,n\to\infty}
	\int_{\mathbb{R}}{\left| q_{A_n;m}(x)-q_{A;m}(x) \right|\ud x} =0,
\end{equation}
as long as we can show that $q_{A_n;m}(x)$ and $q_{A;m}(x)$ for all $m,n\in\mathbb{N}$ belong to the class of distributions given in Definition \ref{def.(a,m,v)-AC}. Remember that in the statement of the theorem, in equation \eqref{eqn:defp_AAC}, we had assumed that 
\begin{equation} \label{eqn:defp_AAC2}
	p_A \in (\alpha, \ell, v)\text{--}\mathcal{AC},
\end{equation}
for some positive values for $\alpha, \ell, v$.
We show that for all $m,n\in\mathbb{N}$
\begin{equation*} \label{eq.PoissonH2.5}
	q_{A_n;m}(x),q_{A;m}(x) \in (\alpha, \ell, v')\text{--}\mathcal{AC},
\end{equation*}
where $v'=2v+3$.
In other words, for all $m,n\in\mathbb{N}$ we have
\begin{equation} \label{eqn:PoissonH2.5.3}
	q_{A_n;m}(x),q_{A;m}(x) \in\mathcal{AC},
\end{equation}
\begin{equation} \label{eqn:PoissonH2.5.1}
	q_{A_n;m}(x),q_{A;m}(x)\leq \ell, \qquad (\text{a.e}),
\end{equation}
\begin{equation} \label{eqn:PoissonH2.5.2}
	\E{\left|\widetilde{A}_{nm}\right|^\alpha}, \E{\left|\widetilde{A}_m\right|^\alpha}
	\leq 2v+3.
\end{equation}
As a result, it remains to show \eqref{eq.Poisson|qAn-qA|=0}, \eqref{eqn:PoissonH2.5.3}, \eqref{eqn:PoissonH2.5.1} and \eqref{eqn:PoissonH2.5.2}.

\emph{Proof of \eqref{eq.Poisson|qAn-qA|=0}}:
From \eqref{eq.QuantNorm1} in Lemma \ref{lmm.QuantNorm}, for all $m,n\in\mathbb{N}$ we have
\begin{equation*}
	\int_{\mathbb{R}}{\left| q_{A_n;m}(x)-q_{A;m}(x) \right| \ud x}
	\leq \int_{\mathbb{R}}{\left| p_{A_n}(x)-p_A(x) \right| \ud x}.
\end{equation*}
Thus,
\if@twocolumn
	\begin{align}
	&\lim_{m,n\to\infty}
	\int_{\mathbb{R}}{\left| q_{A_n;m}(x)-q_{A;m}(x) \right| \ud x}  \nonumber\\
	&\qquad\leq \lim_{n\to\infty}
	\int_{\mathbb{R}}{\left| p_{A_n}(x)-p_A(x) \right| \ud x} =0. \label{eq.PoissonQuantH2.5.4}
	\end{align}
\else
	\begin{align}
	&\lim_{m,n\to\infty}
	\int_{\mathbb{R}}{\left| q_{A_n;m}(x)-q_{A;m}(x) \right| \ud x}
	\leq \lim_{n\to\infty}
	\int_{\mathbb{R}}{\left| p_{A_n}(x)-p_A(x) \right| \ud x} =0. \label{eq.PoissonQuantH2.5.4}
	\end{align}
\fi
where \eqref{eq.PoissonQuantH2.5.4} follows from Lemma \ref{lmm.PoissonQuantDist}.
Hence, \eqref{eq.Poisson|qAn-qA|=0} is proved.

\emph{Proof of \eqref{eqn:PoissonH2.5.3}:}
Since the pdfs $q_{A_n;m}$ and $q_{A;m}$ are combination of step functions, hence, their cdf are absolutely continuous.

\emph{Proof of \eqref{eqn:PoissonH2.5.1}}:
Since there exists some $\ell$ such that $p_A(x)\leq \ell$ for almost all $x\in\bR$, we obtain that
\if@twocolumn
	\begin{align*}
	q_{A;m}(x)=& \, m P_{A;m}[i] = m\int_\frac{i-\frac{1}{2}}{m}^\frac{i+\frac{1}{2}}{m}{p_A(x) \ud x}\\
	\leq& m\int_\frac{i-\frac{1}{2}}{m}^\frac{i+\frac{1}{2}}{m}{\ell \ud x}
	=\ell.
	\end{align*}
\else
	\begin{align*}
	q_{A;m}(x)
	=\,m P_{A;m}[i]
	= m\int_\frac{i-\frac{1}{2}}{m}^\frac{i+\frac{1}{2}}{m}{p_A(x) \ud x}
	\leq m\int_\frac{i-\frac{1}{2}}{m}^\frac{i+\frac{1}{2}}{m}{\ell \ud x}
	=\ell.
	\end{align*}
\fi
It remains to show that $q_{A_n;m}(x)$ is bounded too.
Similar to the above equation, it suffices to show that $p_{A_n}(x)\leq \ell$.
From \eqref{eq:pAn} in Lemma \ref{lmm.PoissonQuantDist}, we have that
\begin{align*}
	p_{A_n}(x)
	=&\frac{\ue^{-\frac{\lambda}{n}}}{1-\ue^{-\frac{\lambda}{n}}}
	\sum_{k=1}^\infty
	{\frac{\left( \frac{\lambda}{n} \right)^k}{k!}\big(\overbrace{p_A*\cdots*p_A}^k\big) (x)}\\
	\leq&\frac{\ue^{-\frac{\lambda}{n}}}{1-\ue^{-\frac{\lambda}{n}}}
	\sum_{k=1}^\infty
	{\frac{\left( \frac{\lambda}{n} \right)^k}{k!}\ell}
	=\ell, \EQnum \label{eqn:pAnUpBnd}
\end{align*}
where \eqref{eqn:pAnUpBnd} follows from the fact that $p_A$ is bounded, and and Lemma \ref{lmm:AC+RV}.

\emph{Proof of \eqref{eqn:PoissonH2.5.2}}:
From \eqref{eq.QuantNorm2} in Lemma \ref{lmm.QuantNorm}, we obtain that there exists $M_1\in\mathbb{N}$ such that for all $m>M_1$ and $n\in\mathbb{N}$ we have
\begin{equation} \label{eqn:EAm->EA}
	\E{\left|\widetilde{A}_m\right|^\alpha}
	\leq 2\E{|A|^\alpha}+1,
\end{equation}
\begin{equation} \label{eqn:EAnm->EAn}
	\E{\left|\widetilde{A}_{nm}\right|^\alpha}
	\leq 2\E{|A_n|^\alpha}+1.
\end{equation}
From \eqref{eqn:EAm->EA} we obtain that
\begin{equation}
	\E{\left|\widetilde{A}_m\right|^\alpha}
	\leq 2v+1< 2v+3.
\end{equation}
From \eqref{eqn:EAn->EA} in Lemma \ref{lmm.PoissonQuantDist} we obtain that there exists $N_1\in\mathbb{N}$ such that for all $n>N_1$ we have
\begin{equation} \label{eqn:PoissonEAn->EA}
	\E{\left|\widetilde{A}_n\right|^\alpha}
	\leq \E{|A|^\alpha}+1.
\end{equation}
Therefore, due to \eqref{eqn:EAnm->EAn} and \eqref{eqn:PoissonEAn->EA}, we obtain that
\begin{align*}
	m>M_1,n>N_1\Longrightarrow\E{\left|\widetilde{A}_{nm}\right|^\alpha}
	\leq 2v+3,
\end{align*}
Thus, \eqref{eqn:PoissonH2.5.2} is proved.
\qed

%%%%%%%%%%%%%%%%%%%%%%%%%%%%%%%
% Proof of Theorem: Linear Combination of Poisson and alpha stable
\subsection{Proof of Theorem \ref{thm:LinCom}} \label{subsec:pfr:thm:LinCom}
	From \eqref{def.CompCrit} and using the fact that $Z(t)=X(t)+Y(t)$, we obtain that
	\begin{equation*}
		\mathcal{H}_{m,n}(Z)
		= n\H{\quant{Z_1^{(n)}}{m}}
		= n\H{\quant{X_1^{(n)}+Y_1^{(n)}}{m}}.
	\end{equation*}
	Therefore, we aim to prove that
		\begin{equation} \label{eqn:am-addeda1}
		\lim_{m,n,\frac{m}{\sqrt[\alpha]{n}}\to\infty}{
		\left( \H{\quant{X_1^{(n)}+Y_1^{(n)}}{m}}
		-\log \frac{m}{\sqrt[\alpha]{n}} \right)}
		= \h{X_0},
	\end{equation}
	In the proof of Theorem \ref{thm.AlphaCompCrit}, we compute the quantization entropy of a stable white noise process by showing that
	\begin{equation} \label{eqn:am-addeda2}
		\lim_{m,n,\frac{m}{\sqrt[\alpha]{n}}\to\infty}{
		\left( \H{\quant{X_1^{(n)}}{m}}
		-\log \frac{m}{\sqrt[\alpha]{n}} \right)}
		= \h{X_0},
	\end{equation}
	Here, we essentially want to prove that we can ignore $Y_1^{(n)}$ in the sum $X_1^{(n)}+Y_1^{(n)}$ when we are computing the quantization entropy.
	To show the theorem, we only need to prove that
	\begin{equation} \label{eqn:am-addeda3}
		\lim_{m,n,\frac{m}{\sqrt[\alpha]{n}}\to\infty}{
		\left( \H{\quant{X_1^{(n)}+Y_1^{(n)}}{m}}
		-\H{\quant{X_1^{(n)}}{m}}	 \right)}=0.
	\end{equation}

	From Lemma \ref{lmm:StableXnDist} we obtain that
	\begin{equation*}
		X_1^{(n)} \stackrel{d}{=} \frac{X_0-b_n}{\sqrt[\alpha]{n}},
	\end{equation*}
	where $b_n$ is defined in Lemma \ref{lmm:StableXnDist}.
	Hence,
	\begin{align}
		\H{\quant{X_1^{(n)}+Y_1^{(n)}}{m}}\nonumber
		=& \H{\quant{\frac{X_0-b_n}{\sqrt[\alpha]{n}}+Y_1^{(n)}}{m}} \nonumber\\
		=& \H{\quant{X_0-b_n+\sqrt[\alpha]{n}Y_1^{(n)}}{\frac{m}{\sqrt[\alpha]{n}}}}, \label{a-n=-1}
	\end{align}
	where the last equality follows from Lemma \ref{lmm.H([aX])}. Similarly,
	\begin{equation} \label{a-n=-2}
		\H{\quant{X_1^{(n)}}{m}}
		= \H{\quant{X_0-b_n}{\frac{m}{\sqrt[\alpha]{n}}}},
	\end{equation}
	For some technical reason that will be needed in the proof, we want to show that $-b_n$ in the \eqref{a-n=-1} and \eqref{a-n=-2} can be replaced by a number $r_n<\frac{\sqrt[\alpha]{n}}{m}$ without changing the value of the entropies.  
	Write $-b_n$ as
	\begin{equation*}
		b_n=r_n+k\frac{\sqrt[\alpha]{n}}{m}
	\end{equation*}
	for some $k\in\mathbb{Z}$ and 
	\begin{equation} \label{eqn:dn}
		0 \leq r_n < \frac{\sqrt[\alpha]{n}}{m}.
	\end{equation}
	Then, 
	\begin{align*}
		\H{\quant{X_1^{(n)}+Y_1^{(n)}}{m}}
		=& \H{k+\quant{X_0+r_n+\sqrt[\alpha]{n}Y_1^{(n)}}{\frac{m}{\sqrt[\alpha]{n}}}}
		=& \H{\quant{X_0+r_n+\sqrt[\alpha]{n}Y_1^{(n)}}{\frac{m}{\sqrt[\alpha]{n}}}}.
	\end{align*}
	Similarly, we have
	\begin{equation*}
		\H{\quant{X_1^{(n)}}{m}}
		= \H{\quant{X_0+r_n}{\frac{m}{\sqrt[\alpha]{n}}}}.
	\end{equation*}

	Let $V_n := X_0+r_n$ and $W_n := V_n+\sqrt[\alpha]{n}Y_1^{(n)}$.
	Then, to prove \eqref{eqn:am-addeda3}, we only need to show that
	\begin{equation} \label{eqn:am-addeda4}
		\lim_{m,n,\frac{m}{\sqrt[\alpha]{n}}\to\infty}
		\left(\H{\quant{V_n}{\frac{m}{\sqrt[\alpha]{n}}}} - \H{\quant{W_n}{\frac{m}{\sqrt[\alpha]{n}}}} \right)
		=0
	\end{equation}

	From Lemma \ref{lmm.AlphaStableRenyiCondition}, and Lemma \ref{lmm:AC+RV} it can be obtained that $V_n, W_n$ are both continuous random variables with pdfs $q_{V_n}, q_{W_n}$, respectively. 

	Define $\widetilde{V}_{n, m/\sqrt[\alpha]{n}} \sim q_{V_n; m/\sqrt[\alpha]{n}}$ and $\widetilde{W}_{n, m/\sqrt[\alpha]{n}} \sim q_{W_n; m/\sqrt[\alpha]{n}}$ according to Definition \ref{def.AmpQuant} for random variables $V_n$ and $W_n$, respectively.
	From Lemma \ref{lmm:H([X])-ln m=h(q)}, we have that
	\begin{equation} \label{eqn:Hq-ln=hq}
		\H{\quant{V_n}{\frac{m}{\sqrt[\alpha]{n}}}} - \log{\frac{m}{\sqrt[\alpha]{n}}}
		=\h{q_{V_n; m/\sqrt[\alpha]{n}}},
		\qquad \H{\quant{W_n}{\frac{m}{\sqrt[\alpha]{n}}}} - \log{\frac{m}{\sqrt[\alpha]{n}}}
		=\h{q_{W_n; m/\sqrt[\alpha]{n}}}.
	\end{equation}
	Therefore, 
	\begin{equation} \label{eqn:am-addeda55}
		\H{\quant{V_n}{\frac{m}{\sqrt[\alpha]{n}}}} -  \H{\quant{W_n}{\frac{m}{\sqrt[\alpha]{n}}}}
		=\h{q_{V_n; m/\sqrt[\alpha]{n}}} -  \h{q_{W_n; m/\sqrt[\alpha]{n}}}.
	\end{equation}
	To show \eqref{eqn:am-addeda4}, we need to show that the difference $\h{q_{V_n; m/\sqrt[\alpha]{n}}} - \h{q_{W_n; m/\sqrt[\alpha]{n}}}$ vanishes in the limit.
	To show this, we utilize Theorem \ref{thm.EntConv}: we prove existence of constants $0<\gamma,\ell,v<\infty$ such that for any $m$ and $n$ satisfying $m/\sqrt[\alpha]{n}\geq 1$ we have
	\begin{equation} \label{eqn:|qVn-qWn|->0}
		\int_{\bR}{\left| q_{V_n; m/\sqrt[\alpha]{n}}(x) - q_{W_n; m/\sqrt[\alpha]{n}}(x) \right| \ud x}
		\leq 2\left( 1-\ue^{-\frac{\lambda}{n}} \right),
	\end{equation}
	\begin{equation} \label{eqn:Vn,Wn-amvAC}
		q_{V_n; m/\sqrt[\alpha]{n}} \in \left( \gamma, \ell , v \right) \text{--}\mathcal{AC},
		\qquad q_{W_n; m/\sqrt[\alpha]{n}} \in \left( \gamma, \ell , v \right) \text{--}\mathcal{AC}.
	\end{equation}
	Then, Theorem \ref{thm.EntConv} yields
	\begin{equation} \label{eqn:hw-hv}
		\left|\h{q_{W_n; m/\sqrt[\alpha]{n}}}-\h{q_{V_n; m/\sqrt[\alpha]{n}}}\right|
		\leq c_1 \Delta + c_2 \Delta \log\tfrac{1}{\Delta},
	\end{equation}
	for some constants $c_1$ and $c_2$ (not depending on $m$ and $n$) and $\Delta =  2\left( 1-\ue^{-\frac{\lambda}{n}} \right)$.
	The differential entropy difference in \eqref{eqn:hw-hv} vanishes as $\Delta$ vanishes when $n$ goes to infinity.
	This completes the proof.  It only remains to prove \eqref{eqn:|qVn-qWn|->0} and \eqref{eqn:Vn,Wn-amvAC}.

	\emph{Proof of \eqref{eqn:|qVn-qWn|->0}}:
	Lemma \ref{lmm.QuantNorm} shows that it is enough to prove that
	\begin{equation*}
		\int_{\bR}{\left| p_{V_n}(x) - p_{W_n}(x) \right| \ud x}
		\leq 2\left( 1-\ue^{-\frac{\lambda}{n}} \right).
	\end{equation*}
	From \eqref{eq.pYn} in Lemma \ref{lmm.PoissonQuantDist} we obtain that $\sqrt[\alpha]{n}Y_1^{(n)}=W_n-V_n$ has the following pdf:
	\begin{equation*}
		\sqrt[\alpha]{n} Y_1^{(n)} \sim
		\ue^{-\frac{\lambda}{n}}\delta(x)
		+\left( 1-\ue^{-\frac{\lambda}{n}} \right) p_{A'_n}(x),
	\end{equation*}
	where
	\begin{equation*}
		p_{A'_n}(x)
		= \frac{1}{\sqrt[\alpha]{n}} p_{A_n} \left(\frac{x}{\sqrt[\alpha]{n}} \right)
	\end{equation*}
	is absolutely continuous.
	Therefore, $p_{W_n}$ is
	\begin{align*}
		p_{W_n}(x) =&
		\left( p_{V_n} * p_{\sqrt[\alpha]{n}Y_1^{(n)}} \right) (x) \\
		=& \ue^{-\frac{\lambda}{n}} p_{V_n} (x)
		+\left( 1-\ue^{-\frac{\lambda}{n}} \right) \left( p_{V_n} * p_{A'_n} \right) (x).
	\end{align*}
	Then, in order to find the total variation, we can write
	\begin{equation*}
		\left |p_{V_n}(x)-p_{W_n}(x) \right|
		\leq \left(1-\ue^{-\frac{\lambda}{n}}\right) p_{V_n}(x)
		+ \left(1-\ue^{-\frac{\lambda}{n}}\right) \left( p_{V_n} * p_{A'_n} \right) (x)
	\end{equation*}
	Hence, we can write that
	\begin{align*}
		&\int_{\bR}{\left| p_{V_n}(x) - p_{W_n}(x) \right| \ud x} \\
		&\qquad\leq\left(1-\ue^{-\frac{\lambda}{n}}\right) \int_{\bR}{p_{V_n}(x) \ud x}
		+\left( 1-\ue^{-\frac{\lambda}{n}} \right) \int_{\bR}{\left( p_{V_n} * p_{A'_n} \right) (x) \ud x} \\
		&\qquad =2\left( 1-\ue^{-\frac{\lambda}{n}} \right).
	\end{align*}

	\emph{Proof of \eqref{eqn:Vn,Wn-amvAC}}:
	From Definition \ref{def.AmpQuant}, it is clear that $\widetilde{V}_{n,m/\sqrt[\alpha]{n}}, \widetilde{W}_{n,m/\sqrt[\alpha]{n}}$ have density and are continuous.
	It remains to prove the other requirements:

	\begin{itemize}
		\item We show that there exists $\ell<\infty$ that $q_{V_n;m/\sqrt[\alpha]{n}}(x), q_{W_n;m/\sqrt[\alpha]{n}}(x) \leq \ell$ for all $m,n$ and $x\in\bR$.
		
		Lemma \ref{lmm.AlphaStableRenyiCondition} shows that $X_0$ has a density which is bounded by some $\ell<\infty$.
		The same bound applies to the density of $V_n := X_0+r_n$, which is a shifted version of $X_0$.
		Utilizing Lemma \ref{lmm:AC+RV}, $W_n$ also has density bounded by $\ell<\infty$.
		Thus, we obtain that for some $\ell<\infty$ (not depending on $n$)
		\begin{equation*}
			q_{V_n}(x), q_{W_n}(x) \leq \ell < \infty,
			\qquad \forall x\in\bR.
		\end{equation*}
		From the definition of $q_{V_n;m/\sqrt[\alpha]{n}}(x), q_{W_n;m/\sqrt[\alpha]{n}}(x)$ in Definition \ref{def.AmpQuant}, the densities of $q_{V_n;m/\sqrt[\alpha]{n}}(x)$ and $ q_{W_n;m/\sqrt[\alpha]{n}}(x)$ are the average of the densities of $q_{V_n}(x)$ and $q_{W_n}(x)$ over the quantization interval.
		Therefore, they are also bounded from above by $\ell$:
		\begin{equation*}
			q_{V_n;m/\sqrt[\alpha]{n}}(x), q_{W_n;m/\sqrt[\alpha]{n}}(x) \leq \ell < \infty,
			\qquad \forall x\in\bR.
		\end{equation*}

		\item There exists $v<\infty$ and $\gamma>0$ such that $\E{\left|\widetilde{V}_{n;m/\sqrt[\alpha]{n}}\right|^\gamma}, \E{\left|\widetilde{W}_{n;m/\sqrt[\alpha]{n}}\right|^\gamma} \leq v$ for all $m,n$ satisfying $m/\sqrt[\alpha]{n}\geq 1$:

		Utilizing Lemma \ref{lmm.QuantNorm} and the fact that $m/\sqrt[\alpha]{n}\geq 1$, we have that
		\begin{align*} 
			\E{\left|\widetilde{V}_{n;m/\sqrt[\alpha]{n}}\right|^\gamma}
			&\leq\left(\tfrac{2}{\sqrt{m/\sqrt[\alpha]{n}}}\right)^\gamma+\ue^\frac{\gamma}{\sqrt{m/\sqrt[\alpha]{n}}}\E{|V_n|^\gamma}\\
			&\leq 2^\gamma+\ue^\gamma\E{|V_n|^\gamma}
		\end{align*}
		A similar equation holds for $\E{\left|\widetilde{W}_{n;m/\sqrt[\alpha]{n}}\right|^\gamma}$.
		Therefore, we will be done if we can find  $v'<\infty$ and $\gamma>0$ such that for all $m,n$ satisfying $m/\sqrt[\alpha]{n}\geq 1$, we have
		\begin{equation*}
			\E{\left|V_n\right|^\gamma},\E{\left|W_n\right|^\gamma} < v'.
		\end{equation*}
		To this end, we can write
		\begin{equation*}
			|V_n|
			=|X_0+r_n|
			\leq |X_0| + r_n
			\leq 2 \max{\{|X_0|, r_n\}}.
		\end{equation*}
		Since $r_n<\frac{\sqrt[\alpha]{n}}{m}\leq 1$, we obtain
		\begin{align*}
			|V_n|^\gamma&
			\leq 2^\gamma \max{\{|X_0|^\gamma,r_n^\gamma\}}
			\leq 2^\gamma \left( |X_0|^\gamma + r_n^\gamma \right)
			\leq 2^\gamma \left( |X_0|^\gamma + 1\right).
		\end{align*}
		Note that \cite[Property 1.2.16]{Samor94}
		\begin{equation*}
			\E{|X_0|^\gamma}<\infty,
			\qquad\forall \gamma<\alpha.
		\end{equation*}
		Therefore, the assertion is proved for $V_n$.
		Following a similar argument for $W_n$, we obtain 
		\begin{equation} \label{eqn:|Wn|^gamma}
			|W_n|^\gamma
			\leq 3^\gamma \left( |X_0|^\gamma + r_n^\gamma + n^\frac{\gamma}{\alpha}|Y_1^{(n)}|^\gamma\right)
			\leq 3^\gamma \left( |X_0|^\gamma + 1 + n^\frac{\gamma}{\alpha}|Y_1^{(n)}|^\gamma\right).
		\end{equation}
		Expected value of $|X_0|^\gamma$ is finite for any $\gamma<\alpha$.
		Thus, we only need to find some $\gamma<\alpha$ such that $\E{|Y_1^{(n)}|^\gamma}$ is finite. 
		According to \eqref{eq.pYn} in Lemma \ref{lmm.PoissonQuantDist}, we have
		\begin{align*}
			\E{|Y_1^{(n)}|^\gamma}
			=& \Pr{Y_1^{(n)}=0} \E{|Y_1^{(n)}|^\gamma|Y_1^{(n)}=0} \\
			&+ \Pr{Y_1^{(n)}\neq 0} \E{|Y_n|^\gamma | Y_1^{(n)} \neq 0} \\
			=& 0 + \left(1-\ue^{-\frac{\lambda}{n}}\right) \E{|A_n|^\gamma}
		\end{align*}
		Therefore,
		\begin{align*}
			n^{\frac{\gamma}{\alpha}}\E{|Y_1^{(n)}|^\gamma}
			=& n^{\frac{\gamma}{\alpha}}
			\left(1-\ue^{-\frac{\lambda}{n}}\right) \E{|A_n|^\gamma}
		\end{align*}
		From \eqref{eqn:EAn->EA} in Lemma \ref{lmm.PoissonQuantDist}, one can find some $\tilde{v}<\infty$ such that $\E{|A_n|^\gamma}<\tilde{v}$ for all $n$ provided that $\E{|A|^\gamma} < \infty$. 
		Furthermore, $\sup_{n\in\mathbb{N}}n^{\frac{\gamma}{\alpha}}\left(1-\ue^{-\frac{\lambda}{n}}\right)<\infty$ for any $\gamma<\alpha$.
		Therefore, we may choose any $0<\gamma<\alpha$.
\end{itemize}
\qed

%%%%%%%%%%%%%%%%%%%
% Proof of Theorem: General Comparison
\subsection{Proof of Theorem \ref{thm:GeneralComparison}} \label{subsec:prf:thm:GeneralComparison}
	% Discrete - Discrete-continuous
	\textbf{Proof of item \ref{itm:D-DC}}:
	We can write that
	\begin{align*}
		\frac{\mathcal{H}_{m,n}(X)}{\mathcal{H}_{m,n}(Y)}
		=\frac{\log n}{\kappa_Y(n) (\zeta_Y(n)+\log m)}
		\frac{\frac{\mathcal{H}_{m,n}(X)}{\log n}}{\frac{\mathcal{H}_{m,n}(Y)}{\kappa_Y(n) (\zeta_Y(n)+\log m)}}.
	\end{align*}
	According to Theorem \ref{thm:GenerelEntropyDim}, there exists $c_1,c_2>0$ such that
	\begin{equation*}
		\zeta_Y(n) \geq c_1 \log n,
		\qquad\kappa_Y(n) \geq c_2,
	\end{equation*}
	where the second inequality is obtained from the fact that $\kappa_Y(n)=n\left(1-\exp\left[-\frac{\lambda_Y}{n}\left(1-\alpha_Y\right)\right]\right) $ tends to the constant $\lambda_Y(1-\alpha_Y)$ as $n$ tends to infinity. 
	Therefore,
	\begin{align*}
		\frac{\mathcal{H}_{m,n}(X)}{\mathcal{H}_{m,n}(Y)}
		\leq \frac{\log n}{c_2 (c_1 \log n+\log m)}
		\frac{\frac{\mathcal{H}_{m,n}(X)}{\log n}}{\frac{\mathcal{H}_{m,n}(Y)}{\kappa_Y(n) (\zeta_Y(n)+\log m)}}.
	\end{align*}
	Hence, from the assumption of the theorem, we know that $\log m/\log n$ tends to $\infty$ as $n$ tends to infinity. 
	As a result, it only suffices to prove that
	\begin{align}
		&\lim_{n\to\infty}\sup_{m\geq m'(n)}
		{\frac{\mathcal{H}_{m,n}(X)}{\log n}} < \infty, \label{eqLa3}\\
		&\lim_{n\to\infty} \inf_{m\geq m'(n)}
		{\frac{\mathcal{H}_{m,n}(Y)}{\kappa_Y(n)(\zeta_Y(n) + \log m)}} = 1.\label{eqLa4}
	\end{align}
	Equation \eqref{eqLa3} is immediate from  Theorem \ref{thm:GenerelEntropyDim} since for a discrete $X_0$ we have
	\begin{equation}
		\lim_{{n\to\infty}}\sup_{m:~m\geq m(n)}\left|
		{{\mathcal{H}_{m,n}(X)} -\zeta_X(n)}\right|
		=0.
	\end{equation}
	Furthermore, $\zeta_X(n)$ has a logarithmic growth from the second part of Theorem \ref{thm:GenerelEntropyDim}.  
	To show \eqref{eqLa4}, observe that  from Theorem \ref{thm:GenerelEntropyDim}, we have
	\begin{equation*}
		\lim_{{n\to\infty}}\sup_{m>m'(n)}\bigg|
		{\frac{\mathcal{H}_{m,n}(Y)}{\kappa_Y(n)}-\log m -\zeta_Y(n)\bigg| = 0},
	\end{equation*}
	which means that for any $\epsilon>0$, there exists $n_1$ such that
	\begin{equation*}
		\left|\frac{\mathcal{H}_{m,n}(Y)}{\kappa_Y(n)}-\log m -\zeta_Y(n)\right| \leq \epsilon,
		\qquad\forall n>n_1,\forall m>m'(n).
	\end{equation*}
	In addition, note that $\zeta_Y(n)\geq c_1 \log n$; 
	therefore $\zeta_Y(n) + \log m$ converges to infinity as $n$ converges to infinity.
	Thus, for any $\epsilon>0$ there exists $n_2$ such that
	\begin{equation*}
		\zeta_Y(n) + \log m > \frac{1}{\epsilon},
		\qquad\forall n>n_2,\forall m>m'(n).
	\end{equation*}
	Hence, by taking $n=\max\{n_1,n_2\}$, we have that
	\begin{equation*}
		\left|{\frac{\mathcal{H}_{m,n}(Y)}{\kappa_Y(n)(\zeta_Y(n) + \log m)}}
		-\frac{\log m}{\zeta_Y(n) + \log m} -\frac{\zeta_Y(n)}{\zeta_Y(n) + \log m}\right|
		\leq \epsilon^2,
		\qquad\forall n>n_1,\forall m>m'(n).
	\end{equation*}
	Since $c_1 \log n \leq \zeta_Y(n) \leq c_2 \log n$, and $\log m/ \log n$ tends to infinity if $m>m'(n)$,  we have
	\begin{equation*}
		\lim_{n\to\infty}\sup_{m>m'(n)}\frac{\log m}{\zeta_Y(n) + \log m}
		=\lim_{n\to\infty}\inf_{m>m'(n)}\frac{\log m}{\zeta_Y(n) + \log m}
		=1.
	\end{equation*}
	Likewise,
	\begin{equation*}
		\lim_{n\to\infty}\sup_{m>m'(n)}\frac{\zeta_Y(n)}{\zeta_Y(n) + \log m}
=		\lim_{n\to\infty}\inf_{m>m'(n)}\frac{\zeta_Y(n)}{\zeta_Y(n) + \log m}
		=0.
	\end{equation*}
	Therefore, 
	\begin{equation*}
		\lim_{n\to\infty} \sup_{m\geq m'(n)}
		{\frac{\mathcal{H}_{m,n}(Y)}{\kappa_Y(n)(\zeta_Y(n) + \log m)}}
		=\lim_{n\to\infty} \inf_{m\geq m'(n)}
		{\frac{\mathcal{H}_{m,n}(Y)}{\kappa_Y(n)(\zeta_Y(n) + \log m)}}
		= 1
	\end{equation*}
	Hence, \eqref{eqLa4} is proved.

	% Discrete-continuous - continuous
	\textbf{Proof of item \ref{itm:DC-C}}:
	We can write that
	\begin{equation*}
		\frac{\mathcal{H}_{m,n}(X)}{\mathcal{H}_{m,n}(Y)}
		=\frac{\frac{\mathcal{H}_{m,n}(X)}{n \log m}}
		{\frac{\mathcal{H}_{m,n}(Y)}{n\log m}}.
	\end{equation*}
It suffices to prove that
	\begin{align}
		&\lim_{{n\to\infty}}\sup_{m>m'(n)}
		{\frac{\mathcal{H}_{m,n}(X)}{n \log m}} = 0, \label{eqn:Hdc/nlogm=0}\\
		&\lim_{{n\to\infty}}\inf_{m>m'(n)}
		{\frac{\mathcal{H}_{m,n}(Y)}{n\log m}} = 1. \label{eqn:Hc/nlogm=1}
	\end{align}

	\emph{Proof of \eqref{eqn:Hdc/nlogm=0}}:
	Note that from Theorem \ref{thm:GenerelEntropyDim}, we obtain that for any $\epsilon>0$, there exists $n_1$ such that
	\begin{equation*}
		\left|\frac{\mathcal{H}_{m,n}(X)}{\kappa_X(n)}-\log m -\zeta_X(n)\right| \leq \epsilon,
		\qquad\forall n>n_1,\forall m>m'(n).
	\end{equation*}
	In addition,  $1/(n\log m)\leq \epsilon$ for any $m\geq 2$ and $n\geq \epsilon^{-1}$.  Therefore, for any $\epsilon>0$ we have
	\begin{equation*}
		\left|\frac{\mathcal{H}_{m,n}(X)}{\kappa_X(n) n \log m}
		-\frac{1}{n} -\frac{\zeta_X(n)}{n \log m}\right| \leq \epsilon^2,
		\qquad\forall n>n_2,\forall m>m'(n),
	\end{equation*}
where $n_2=\max(n_1,  \epsilon^{-1})$. 
	Utilizing the fact that from Theorem \ref{thm:GenerelEntropyDim}, $\zeta_X(n) \leq c_2 \log n$, we have
	\begin{equation*}
		\lim_{n\to\infty}\sup_{m>m'(n)}\frac{\zeta_X(n)}{n \log m}
		\leq \lim_{n\to\infty}\sup_{m>m'(n)}\frac{c_2 \log n}{n \log m}
\leq \lim_{n\to\infty}\frac{c_2 \log n}{n}
		=0.
	\end{equation*}
As a result,
	\begin{equation*}
		\lim_{{n\to\infty}}\sup_{m>m'(n)}
		{\frac{\mathcal{H}_{m,n}(X)}{\kappa_X(n)n \log m}} = 0.
	\end{equation*}
 The value of $\kappa_X(n)$ given in Theorem \ref{thm:GenerelEntropyDim} tends to a constant as $n$ tends to infinity. This completes the proof for \eqref{eqn:Hdc/nlogm=0}.

	\emph{Proof of \eqref{eqn:Hc/nlogm=1}}: From the theorem's assumption, the ratio of $\log m'(n)$ and $\zeta_Y(n)$ goes to infinity. We know from Theorem \ref{thm:GenerelEntropyDim} that $\zeta_Y(n)$ is non-decreasing. As a result, \begin{equation}\lim_{n\rightarrow\infty}\log m'(n)=\infty.\label{eq;nreq32}\end{equation} 
Utilizing Theorem \ref{thm:GenerelEntropyDim} and \eqref{eq;nreq32}, we obtain that for any $\epsilon>0$, there exists $n_1$ such that
	\begin{equation*}
		\left|\frac{\mathcal{H}_{m,n}(Y)}{n}
		-\log m -\zeta_Y(n) \right| \leq \epsilon,
		\qquad\forall n>n_1,\forall m>m'(n).
	\end{equation*}
and 
	\begin{equation*}
		\frac{1}{\log(m)} \leq \epsilon,
		\qquad\forall n>n_1,\forall m>m'(n).
	\end{equation*}
As a result,
	\begin{equation*}
		\left|\frac{\mathcal{H}_{m,n}(Y)}{n \log m}
		-1 -\frac{\zeta_Y(n)}{\log m} \right| \leq \epsilon^2,
		\qquad\forall n>n_1,\forall m>m'(n).
	\end{equation*}
We have that $0\leq \zeta_Y(n)/\log m\leq \zeta_Y(n)/\log m'(n)$ for any $m>m'(n)$. Furthermore, $\zeta_Y(n)/\log m'(n)$
 vanishes as $n$ tends to infinity according to the theorem's assumption. Thus, 
	\begin{equation*}
		\lim_{n\to\infty} \sup_{m\geq m'(n)}
		{\frac{\mathcal{H}_{m,n}(Y)}{n \log m}}
		=\lim_{n\to\infty} \inf_{m\geq m'(n)}
		{\frac{\mathcal{H}_{m,n}(Y)}{n \log m}}
		= 1.
	\end{equation*}
	Hence, the statement is proved.
	
	% Continuous - Gaussian
	\textbf{Proof of item \ref{itm:C-CG}}:
	We can write
	\begin{align*}
		\frac{\mathcal{H}_{m,n}(X)}{\mathcal{H}_{m,n}(Y)} - 1 
		=& \frac{\left[\frac{\mathcal{H}_{m,n}(X)}{n} - \log m - \zeta_X(n) \right]
		-\left[ \frac{\mathcal{H}_{m,n}(Y)}{n} - \log\frac{m}{\sqrt{n}}-\h{Y_0}\right]}
		{\frac{\mathcal{H}_{m,n}(Y)}{n}} \\
		&+\frac{\zeta_X(n)-\h{X_0}-\log\frac{1}{\sqrt n}}{\frac{\mathcal{H}_{m,n}(Y)}{n}}
		+\frac{\h{X_0}}{{\frac{\mathcal{H}_{m,n}(Y)}{n}}}
		-\frac{\h{Y_0}}{{\frac{\mathcal{H}_{m,n}(Y)}{n}}}.
	\end{align*}
	In order to prove the theorem, it suffices to show that 
	\begin{align}
		&\lim_{n\to\infty}\sup_{m>m(n)}
		{\left|\frac{\mathcal{H}_{m,n}(X)}{n}-\log m-\zeta_X(n)\right|}
		= 0, \label{eqn:HG,HX/n-lnm<inf1}\\
		&\lim_{n\to\infty}\sup_{m>m(n)}
		{\left|\frac{\mathcal{H}_{m,n}(Y)}{n}-\log\frac{m}{\sqrt{n}}-\h{Y_0}\right|}
		= 0, \label{eqn:HG,HX/n-lnm<inf2}
	\end{align}
	\begin{equation} \label{HG/n=inf}
		\lim_{n\to\infty}\inf_{m>m(n)}{\frac{\mathcal{H}_{m,n}(Y)}{n}}
		= \infty,
	\end{equation}
	\begin{equation} \label{eqn:zeta+alnn>0}
		\limsup_{n\to\infty}{\zeta_X(n) -\h{X_0} - \log\frac{1}{\sqrt n}}
		\leq 0.
	\end{equation}
	
	From Theorem \ref{thm:GenerelEntropyDim} and Theorem \ref{thm.AlphaCompCrit}, the limits of \eqref{eqn:HG,HX/n-lnm<inf1} and \eqref{eqn:HG,HX/n-lnm<inf2} are proved, respectively.

	To prove \eqref{HG/n=inf}, observe that from Theorem \ref{thm.AlphaCompCrit}
	\begin{align*}
		\frac{\mathcal{H}_{m,n}(Y)}{n}
		=\left[\frac{\mathcal{H}_{m,n}(Y)}{n}-\log\frac{m}{\sqrt{n}}\right]
		+\log\frac{m}{\sqrt{n}},
	\end{align*}
	where the first expression converges to $\h{Y_0}$ according to Theorem \ref{thm.AlphaCompCrit}. The second one tends to infinity because of the assumption of the theorem.
	
	Finally, \eqref{eqn:zeta+alnn>0} is the direct consequence of Lemma \ref{lmm:zeta<EntP}.

	Hence, the theorem is proved.
\qed

%%%%%%%%%%%%%%%%%%%%%%%
% Proof of Theorem: Compressibility Comparison
\subsection{Proof of Theorem \ref{thm:Comparison}} \label{subsec:Comparison:proof}
\textbf{Case 1: Stable}

\emph{Proof of \eqref{eqn:ComparStable/}}:
According to Theorem \ref{thm.AlphaCompCrit}, observe that for $i=1,2$
\begin{equation*}
	\lim_{n\to\infty} \sup_{m \geq m(n)}{\log\frac{m}{\sqrt[\alpha_i]{n}}}
	 = +\infty.
\end{equation*}
Therefore, we can write
\begin{equation*}
	\lim_{n\to\infty}\sup_{m \geq m(n)}
	{\left|\frac{\mathcal{H}_{m,n}(X_i)}{n\log\frac{m}{\sqrt[\alpha_i]{n}}}-1\right|}
	= 0.
\end{equation*}

Note that
\begin{equation*}
	\frac{\mathcal{H}_{m,n}(X_1)}{\mathcal{H}_{m,n}(X_2)}
	=\frac{\frac{\mathcal{H}_{m,n}(X_1)}{n\log({m}/{\sqrt[\alpha_1]{n}})}}
	{\frac{\mathcal{H}_{m,n}(X_2)}{n\log({m}/{\sqrt[\alpha_2]{n}})}}
	\frac{\log({m}/{\sqrt[\alpha_1]{n}})}{\log({m}/{\sqrt[\alpha_2]{n}})}.
\end{equation*}
Since the limit of the second fraction depends on how $m$ changes as $n$ tends to $\infty$, it does not have a limit.
However, due to $\alpha_1<\alpha_2$ we can write
\begin{equation*}
	\frac{\log({m}/{\sqrt[\alpha_1]{n}})}{\log({m}/{\sqrt[\alpha_2]{n}})}
	< 1.
\end{equation*}
As a result, the statement is proved.

\emph{Proof of \eqref{eqn:ComparStable-}}:
We can write
\if@twocolumn
	\begin{align*}
	&\mathcal{H}_{m_k,n_k}(X_1)-\mathcal{H}_{m_k,n_k}(X_2)\\
	&~=n_k\Bigg[\frac{\mathcal{H}_{m_k,n_k}(X_1)-n_k\log\frac{m_k}{\sqrt[\alpha_1]{n_k}}}{n_k} \nonumber\\
	&\qquad\quad-\frac{\mathcal{H}_{m_k,n_k}(X_2)-n_k\log\frac{m_k}{\sqrt[\alpha_2]{n_k}}}{n_k}
	+\log\frac{\sqrt[\alpha_2]{n_k}}{\sqrt[\alpha_1]{n_k}} \Bigg].
	\end{align*}
\else
	\begin{equation*}
	\mathcal{H}_{m,n}(X_1)-\mathcal{H}_{m,n}(X_2)
	=n \Bigg[\frac{\mathcal{H}_{m,n}(X_1)}{n} - \log\frac{m}{\sqrt[\alpha_1]{n}}
	-\frac{\mathcal{H}_{m,n}(X_2)}{n} +\log\frac{m}{\sqrt[\alpha_2]{n}} \Bigg]
	+n \log\frac{\sqrt[\alpha_2]{n}}{\sqrt[\alpha_1]{n}}.
	\end{equation*}
\fi
From Theorem \ref{thm.AlphaCompCrit} we obtain that there exists $n_0$ such that for $n>n_0$ we have
\begin{equation*}
	\frac{\mathcal{H}_{m,n}(X_1)}{n} - \log\frac{m}{\sqrt[\alpha_1]{n}}
	-\frac{\mathcal{H}_{m,n}(X_2)}{n} +\log\frac{m}{\sqrt[\alpha_2]{n}}
	\geq C,
\end{equation*}
where $C:=\h{X_0^{(\alpha_1)}}-\h{X_0^{(\alpha_2)}}-1$.
Therefore, for $n>n_0$, we have that
\begin{equation*}
	\mathcal{H}_{m,n}(X_1)-\mathcal{H}_{m,n}(X_2)
	\geq n\left(C+\log\frac{\sqrt[\alpha_2]{n}}{\sqrt[\alpha_1]{n}}\right).
\end{equation*}
Hence, $\mathcal{H}_{m,n}(X_1)-\mathcal{H}_{m,n}(X_2)\to+\infty$.

% Impulsive Poisson
\textbf{Case 2: Impulsive Poisson}

\emph{Proof of \eqref{eqn:ComparPoisson/}}:
Observe that for $i=1,2$ we have
\begin{equation*}
	\lim_{n\to\infty} \sup_{m \geq m(n)}{n\left(1-\ue^{-{\lambda_i}/{n}}\right) \log(mn)}
	 = +\infty.
\end{equation*}
Hence, from Theorem \ref{thm.PoissonCompCrit}, we can write
\begin{equation*}
	\lim_{n\to\infty}\sup_{m \geq m(n)}
	{\left|\frac{\mathcal{H}_{m,n}(X_i)}{n\left(1-\ue^{-{\lambda_i}/{n}}\right)\log(mn)}-1\right|}
	= 0.
\end{equation*}

Note that
\begin{equation*}
	\frac{\mathcal{H}_{m,n}(X_1)}{\mathcal{H}_{m,n}(X_2)}
	=\frac{\frac{\mathcal{H}_{m,n}(X_1)}{n\left(1-\ue^{-{\lambda_1}/{n}}\right)\log(mn)}}{\frac{\mathcal{H}_{m,n}(X_2)}{n\left(1-\ue^{-{\lambda_2}/{n}}\right)\log(mn)}}
	\frac{1-\ue^{-\frac{\lambda_1}{n}}}{1-\ue^{-\frac{\lambda_2}{n}}}.
\end{equation*}

Hence, we obtain that
\begin{equation*}
	\lim_{n\to\infty}\sup_{m \geq m(n)}
	{\frac{\mathcal{H}_{m,n}(X_1)}{\mathcal{H}_{m,n}(X_2)}}
	=\frac{\lambda_1}{\lambda_2} < 1.
\end{equation*}

\emph{Proof of \eqref{eqn:ComparPoisson-}}:
From Theorem \ref{thm.PoissonCompCrit} and the fact that
\begin{equation*}
	\lim_{n\to\infty}{n \left(1-\ue^{-\frac{\lambda}{n}}\right)}
	= \lambda,
\end{equation*}
we can write for $i=1,2$ that
\begin{align*}
	&\lim_{n\to\infty}
	{\sup_{m\geq m(n)}
	{\mathcal{H}_{m,n}(X_i)
	- n \left(1-\ue^{-\frac{\lambda_i}{n}}\right) \log(mn)}} \\
	&\qquad=\lim_{n\to\infty}
	{\inf_{m\geq m(n)}
	{\mathcal{H}_{m,n}(X_i)
	- n \left(1-\ue^{-\frac{\lambda_i}{n}}\right) \log(mn)}} \\
	&\qquad= \lambda_i (\h{A_i} + \log \lambda_i - 1).
\end{align*}
Then, we can write that
\begin{align*}
	\mathcal{H}_{m,n}(X_1)-\mathcal{H}_{m,n}(X_2)
	=&\left[\mathcal{H}_{m,n}(X_1)-n\left(1-\ue^{-\frac{\lambda_1}{n}}\right) \log(m n)\right] \\
	&-\left[\mathcal{H}_{m,n}(X_2)-n\left(1-\ue^{-\frac{\lambda_2}{n}}\right)\log(m n)\right] \\
	&+n\left(\ue^{-\frac{\lambda_2}{n}}-\ue^{-\frac{\lambda_1}{n}}\right)\log(m n).
\end{align*}
Therefore, the first two parts are converging to a constant.
For the last one, note that there exists $n_0$ such that for $n>n_0$ we have that while the last one converges to $-\infty$.
Therefore, the theorem is proved.
\qed

%%%%%%%%%%%%%%%%%%%%
% Proof of Lemma: 
\subsection{Proof of Lemma \ref{lmm:finite v->Poisson}} \label{subsec:prf:lmm:finite v->Poisson}
	From Definition \ref{def.WhiteNoise} and Theorem \ref{thm.LevyKhintchin} we obtain that for any function $\varphi(t)\in\mathcal{S}(\bR)$
	\begin{align*}
		\E{\ue^{\uj \InProd{X}{\varphi}}}
		=& \exp\left(
		j \mu \int_\bR{\varphi(t) \ud t}
		+ \int_{\bR\setminus\{0\}}{\int\left(\ue^{\uj\varphi(t) a}-1-\uj\varphi(t) a \mathds{1}_{(-1,1)}(a) \ud t \right) v(a) \ud a} \right) \\
		=& \exp\left(
		j \mu' \int_\bR{\varphi(t) \ud t}
		+ \lambda \int_{\bR\setminus\{0\}}\int{\left(\ue^{\uj\varphi(t) a}-1\right) p_A(a) \ud t \ud a} \right)
		 \\
		=& \exp\left(
		j \mu' \int_\bR{\varphi(t) \ud t}\right)\times  \exp\left(
		 \lambda \int_{\bR\setminus\{0\}}\int{\left(\ue^{\uj\varphi(t) a}-1\right) p_A(a) \ud t \ud a} \right).
	\end{align*}
	Hence, from Definition \ref{def.Poisson}, the assertion is proved.
\qed

%%%%%%%%%%%%%%%%%%%%%%%%%%%%%%
% Proof of Theorem: Time quantized Poisson process distribution
\subsection{Proof of Lemma \ref{lmm.PoissonQuantDist}} \label{subsec:ProofPoissonQuantDist}
\emph{Proof of \eqref{eq.pYn}}:
First, we find the characteristic function of $Y_n$ in terms of the characteristic function of $A$.
From the definition of $Y_n$ in \eqref{eq.Yn}, we have that
\begin{equation*}
	\widehat{p}_{Y_n}(\omega)
	= \E{\ue^{\uj\omega Y_n}}
	= \E{\ue^{\uj\omega\InProd{X_n}{\phi_{1,n}}}},
\end{equation*}
where $\phi_{1,n}$ defined in Definition \ref{def.TimeQuant}.
Due to the definition of white noise in Definition \ref{def.WhiteNoise}, we can write
\begin{equation*}
	\widehat{p}_{Y_n}(\omega)=\E{\ue^{\uj\InProd{X_n}{\omega\phi_{1,n}}}}
	= \ue^{\int_0^\frac{1}{n}{{f(\omega)}\ud t}}
	= \ue^{\frac{1}{n}f(\omega)}.
\end{equation*}
Because of the definition of Poisson process in Definition \ref{def.Poisson}, we have that
\begin{equation*}
	f(\omega)
	= \lambda\int_\bR{(\ue^{\uj\omega x}-1)p_A(x)\ud x}
	= \lambda\widehat{p}_A(\omega)-\lambda,
\end{equation*}
where $\widehat{p}_A(\omega)$ is the characteristic function of $A$. Hence, we conclude that
\if@twocolumn
\begin{align*}
	\widehat{p}_{Y_n}(\omega)
	=&\ue^{-\frac{\lambda}{n}}\ue^{\frac{\lambda}{n} \widehat{p}_A(\omega)}\\
	=&\ue^{-\frac{\lambda}{n}}\left(1+\sum_{k=1}^\infty
	{\frac{\left( \frac{\lambda}{n}\widehat{p}_A(\omega) \right)^k}{k!}}\right)\\
	=&\ue^{-\frac{\lambda}{n}}+\left(1-\ue^{-\frac{\lambda}{n}}\right)\widehat{p}_{A_n}(\omega), \EQnum\label{eqn:pYw pYnw}
\end{align*}
\else
\begin{align}
	\widehat{p}_{Y_n}(\omega)
	=&\ue^{-\frac{\lambda}{n}}\ue^{\frac{\lambda}{n} \widehat{p}_A(\omega)}
	=\ue^{-\frac{\lambda}{n}}\left(1+\sum_{k=1}^\infty
	{\frac{\left( \frac{\lambda}{n}\widehat{p}_A(\omega) \right)^k}{k!}}\right) \nonumber\\
	=&\ue^{-\frac{\lambda}{n}}+\left(1-\ue^{-\frac{\lambda}{n}}\right)\widehat{p}_{A_n}(\omega), \label{eqn:pYw pYnw}
\end{align}
\fi
where
\begin{equation} \label{eq.HpYn}
	\widehat{p}_{A_n}(\omega):=
	\frac{1}{\ue^{\frac{\lambda}{n}}-1}\sum_{k=1}^\infty
	{\frac{\left( \frac{\lambda}{n}\widehat{p}_A(\omega) \right)^k}{k!}}.
\end{equation}
Therefore, by taking the inverse Fourier transform of \eqref{eqn:pYw pYnw}, \eqref{eq.pYn} is proved.

\emph{Proof of \eqref{eq:pAn}}:
It is achieved by taking the inverse Fourier transform from \eqref{eq.HpYn}.

\emph{Proof of \eqref{eq.VarConv}:}
We can write
\if@twocolumn
	\begin{align*}
	&p_{A_n}(x) =
	\frac{1}{\ue^{\frac{\lambda}{n}}-1}
	\sum_{k=1}^\infty
	{\frac{\left( \frac{\lambda}{n} \right)^k}{k!}\big(\overbrace{p_A*\cdots*p_A}^k \big)(x)} \\
	&~~~ =\frac{\frac{\lambda}{n}}{\ue^{\frac{\lambda}{n}}-1} p_A(x)
	+\frac{1}{\ue^{\frac{\lambda}{n}}-1}
	\sum_{k=2}^\infty
	{\frac{\left( \frac{\lambda}{n} \right)^k}{k!}\big(\overbrace{p_A*\cdots*p_A}^k \big)(x)}. 	
	 \end{align*}
\else
	\begin{align*}
	p_{A_n}(x) 
	&= \frac{1}{\ue^{\frac{\lambda}{n}}-1}
	\sum_{k=1}^\infty
	{\frac{\left( \frac{\lambda}{n} \right)^k}{k!}\big(\overbrace{p_A*\cdots*p_A}^k \big)(x)} \\
	&=\frac{\frac{\lambda}{n}}{\ue^{\frac{\lambda}{n}}-1} p_A(x)
	+\frac{1}{\ue^{\frac{\lambda}{n}}-1}
	\sum_{k=2}^\infty
	{\frac{\left( \frac{\lambda}{n} \right)^k}{k!}\big(\overbrace{p_A*\cdots*p_A}^k \big)(x)}. 	
	 \end{align*}
\fi
Hence, we can write:
\if@twocolumn
	\begin{align*}
	&| p_{A_n}(x)-p_A(x) | \\
	&\qquad\leq\left| \tfrac{\frac{\lambda}{n}}{\ue^{\frac{\lambda}{n}}-1} -1\right| p_A(x) \nonumber\\
	&\qquad\quad+\frac{1}{\ue^{\frac{\lambda}{n}}-1}
	\sum_{k=2}^\infty
	{\frac{\left( \frac{\lambda}{n} \right)^k}{k!} \big(\overbrace{p_A*\cdots*p_A}^k\big) (x)}.
	\end{align*}
\else
	\begin{align*}
	| p_{A_n}(x)-p_A(x) |
	\leq&\left| \tfrac{\frac{\lambda}{n}}{\ue^{\frac{\lambda}{n}}-1} -1\right| p_A(x)
	+\frac{1}{\ue^{\frac{\lambda}{n}}-1}
	\sum_{k=2}^\infty
	{\frac{\left( \frac{\lambda}{n} \right)^k}{k!} \big(\overbrace{p_A*\cdots*p_A}^k\big) (x)}.
	\end{align*}
\fi
Since $p_A$ is a probability density, $\overbrace{p_A*\cdots*p_A}^k$ is also a probability density; hence,
\begin{equation*}
	\int_{\mathbb{R}}{\big(\overbrace{p_A*\cdots*p_A}^k\big)(x) \ud x} = 1.
\end{equation*}
Therefore,
\if@twocolumn
	\begin{align*}
	\int_\mathbb{R}{| p_{A_n}(x)-p_A(x) |\ud x}
	\leq& \left| \tfrac{\frac{\lambda}{n}}{\ue^{\frac{\lambda}{n}}-1} \hspace{-1mm}-\hspace{-1mm} 1\right|
	+\frac{1}{\ue^{\frac{\lambda}{n}} \hspace{-1mm}-\hspace{-1mm} 1}
	\sum_{k=2}^\infty
	{\frac{\left( \frac{\lambda}{n} \right)^k}{k!}} \\
	=& 2\frac{\ue^{\frac{\lambda}{n}} -\frac{\lambda}{n}-1}{\ue^{\frac{\lambda}{n}}-1},
	\end{align*}
\else
	\begin{align*}
	\int_\mathbb{R}{| p_{A_n}(x)-p_A(x) |\ud x}
	\leq \left| \tfrac{\frac{\lambda}{n}}{\ue^{\frac{\lambda}{n}}-1} -1\right|
	+\frac{1}{\ue^{\frac{\lambda}{n}}-1}
	\sum_{k=2}^\infty
	{\frac{\left( \frac{\lambda}{n} \right)^k}{k!}} 
	= 2\frac{\ue^{\frac{\lambda}{n}}-\frac{\lambda}{n}-1}{\ue^{\frac{\lambda}{n}}-1},
	\end{align*}
\fi
which vanishes as $n$ tends to $\infty$.

\emph{Proof of \eqref{eqn:EAn->EA}}:
Let the sequence of independent and identically distributed random variables $\left\lbrace A^{(i)}\right\rbrace_{i=1}^\infty$ have pdf $p_A$, and let the discrete random variable $B$ be independent to $\left\lbrace A^{(i)}\right\rbrace_{i=1}^\infty$ and have the following pmf:
\begin{equation*}
	P_B[k]:=\Pr{B=k}
	= \frac{1}{\ue^{\frac{\lambda}{n}}-1}\frac{\left( \frac{\lambda}{n} \right)^k}{k!},
	\qquad \forall k\in\mathbb{N}:=\{1,2,\cdots\}.
\end{equation*}
From \eqref{eq:pAn}, we obtain that
\begin{equation*}
	A_n \stackrel{d}{=} \sum_{i=1}^B{A^{(i)}}.
\end{equation*}
Now, to find $\E{\left|A_n\right|^\alpha}$, we can write
\if@twocolumn
	\begin{align*}
	&\E{\left|A_n\right|^\alpha}
	=\sum_{k=1}^\infty{P_B[k]\E{\Big| \sum_{i=1}^k{A^{(i)}} \Big|^\alpha}}\\
	&~~~ =\tfrac{\frac{\lambda}{n}}{\ue^\frac{\lambda}{n}-1}\E{|A|^\alpha}
	+\tfrac{1}{\ue^\frac{\lambda}{n}-1}
	\sum_{k=2}^\infty
	{\tfrac{\left( \frac{\lambda}{n} \right)^k}{k!}\E{\Big|\sum_{i=1}^k{A^{(i)}}\Big|^\alpha}}.
	\end{align*}
\else
	\begin{align*}
	\E{\left|A_n\right|^\alpha}
	=\sum_{k=1}^\infty{P_B[k]\E{\Big| \sum_{i=1}^k{A^{(i)}} \Big|^\alpha}}
	=\tfrac{\frac{\lambda}{n}}{\ue^\frac{\lambda}{n}-1}\E{|A|^\alpha}
	+\tfrac{1}{\ue^\frac{\lambda}{n}-1}
	\sum_{k=2}^\infty
	{\tfrac{\left( \frac{\lambda}{n} \right)^k}{k!}\E{\Big|\sum_{i=1}^k{A^{(i)}}\Big|^\alpha}},
	\end{align*}
\fi
provided that the limit exists.
From
\begin{equation*}
	\lim_{n\to\infty}\tfrac{\frac{\lambda}{n}}{\ue^\frac{\lambda}{n}-1}=1,
\end{equation*}
we have to prove that
\begin{equation}
	\lim_{n\to\infty}{\tfrac{1}{\ue^\frac{\lambda}{n}-1}
	\sum_{k=2}^\infty{\tfrac{\left( \frac{\lambda}{n} \right)^k}{k!}\E{\Big|\sum_{i=1}^k{A^{(i)}}\Big|^\alpha}}}
	=0. \label{eq:pAn1}
\end{equation}
In order to prove \eqref{eq:pAn1}, it suffices to show that 
\begin{equation} \label{eq:pAn2}
	\E{\Big|\sum_{i=1}^k{A^{(i)}}\Big|^\alpha} \text{ exists},
	\qquad\E{\Big|\sum_{i=1}^k{A^{(i)}}\Big|^\alpha}\leq c^k,
\end{equation}
where $c=2^\alpha(1+\E{|A|^\alpha})$ is a constant.
This is because
\if@twocolumn
	\begin{align*}
	\tfrac{1}{\ue^\frac{\lambda}{n}-1}
	\sum_{k=2}^\infty{\tfrac{\left( \frac{\lambda}{n} \right)^k}{k!}
	\E{\Big|\sum_{i=1}^k{A^{(i)}}\Big|^\alpha}} 
	&\leq\tfrac{1}{\ue^\frac{\lambda}{n}-1}
	\sum_{k=2}^\infty{\tfrac{\left( \frac{\lambda}{n} \right)^k}{k!}c^k}\\
	&=\tfrac{\ue^{\frac{\lambda}{n}c}-\frac{\lambda}{n}c-1}{\ue^\frac{\lambda}{n}-1},
	\end{align*}
\else
	\begin{align*}
	\tfrac{1}{\ue^\frac{\lambda}{n}-1}
	\sum_{k=2}^\infty{\tfrac{\left( \frac{\lambda}{n} \right)^k}{k!}
	\E{\Big|\sum_{i=1}^k{A^{(i)}}\Big|^\alpha}}
	\leq&\tfrac{1}{\ue^\frac{\lambda}{n}-1}
	\sum_{k=2}^\infty{\tfrac{\left( \frac{\lambda}{n} \right)^k}{k!}c^k}
	=\tfrac{\ue^{\frac{\lambda}{n}c}-\frac{\lambda}{n}c-1}{\ue^\frac{\lambda}{n}-1},
	\end{align*}
\fi
vanishes as $n$ tends to infinity.
Therefore, it only remains to prove \eqref{eq:pAn2}.
We can bound $\big|A^{(1)}+\cdots+A^{(k)}\big|$ as follows:
\if@twocolumn
\begin{align*}
	\Big|\sum_{i=1}^k{A^{(i)}}\Big|
	&\leq \sum_{i=1}^k{\left|A^{(i)}\right|}
	\leq 1+\sum_{i=1}^k{\left|A^{(i)}\right|}\\
	& \leq \prod_{i=1}^k{\left(1+\left|A^{(i)}\right|\right)}
	\leq \prod_{i=1}^k\left({2\max{\left\lbrace1,\left|A^{(i)}\right|\right\rbrace}}\right).
\end{align*}
\else
\begin{align*}
	\Big|\sum_{i=1}^k{A^{(i)}}\Big|
	\leq \sum_{i=1}^k{\left|A^{(i)}\right|}
	\leq 1+\sum_{i=1}^k{\left|A^{(i)}\right|}
	\leq \prod_{i=1}^k{\left(1+\left|A^{(i)}\right|\right)}
	\leq \prod_{i=1}^k\left({2\max{\left\lbrace1,\left|A^{(i)}\right|\right\rbrace}}\right).
\end{align*}
\fi
By finding the expected value of the $\alpha$ power of both sides, we have
\begin{align}
	\E{\Big|\sum_{i=1}^k{A^{(i)}}\Big|^\alpha}
	\leq& 2^{\alpha k}\, \E{\prod_{i=1}^k{\max{\left\lbrace1,\left|A^{(i)}\right|^\alpha\right\rbrace}}} \label{eq:pAn3.1}\\
	=&2^{\alpha k}\, \prod_{i=1}^k\E{\max{\left\lbrace1,\left|A^{(i)}\right|^\alpha\right\rbrace}} \label{eq:pAn3.2}\\
	=&\left(2^\alpha \, \E{\max{\left\lbrace1,\left|A\right|^\alpha\right\rbrace}}\right)^k \label{eq:pAn3.3},
\end{align}
where \eqref{eq:pAn3.1} is true because $\max\lbrace1,|x|\rbrace^\alpha=\max\lbrace1,|x|^\alpha\rbrace$ for all $x\in\mathbb{R}$, \eqref{eq:pAn3.2} is true because $A^{(1)},\cdots,A^{(n)}$ are independent, and \eqref{eq:pAn3.3} is true because $A^{(1)},\cdots,A^{(n)}$ are identically distributed.
In order to find an upper bound for \eqref{eq:pAn3.3}, we can write
\begin{align*}
	\E{\max{\left\lbrace1,\left|A\right|^\alpha\right\rbrace}}
	\leq  ~\E{1+\left|A\right|^\alpha} 
	= 1+\E{\left|A\right|^\alpha}
\end{align*}
Hence, \eqref{eq:pAn2}, and the theorem are proved.
\qed

%%%%%%%%%%%%%%%%%%%%%%%%%%%%
% Proof of Lemma: Renyi Conditions for Stable White Noises
\subsection{Proof of Lemma \ref{lmm.AlphaStableRenyiCondition}} \label{subsec:ProofX0features}
\emph{Proof of \eqref{eqn:pX0AC}, \eqref{eqn:pX0<L}, and \eqref{eqn:pX0PWC}}:
If we prove that $X_0$ is a stable random variable, then \eqref{eqn:pX0AC}-\eqref{eqn:pX0PWC} are proved. \cite{Samor94}
In order to do so, we find the characteristic function of $X_0$:
\begin{align*}
	\widehat{p}_{X_0}(\omega):=&\E{\ue^{\uj\omega\InProd{X}{\phi}}}
	=\ue^{\int_0^1{f(\omega)\ud t}}\EQnum\label{eq.X0 1.1}
	=\ue^{f(\omega)},
\end{align*}
where, $f(\omega)$ is a valid L\'evy exponent, defined in \eqref{eq.StableLevyExp}, and \eqref{eq.X0 1.1} is true because of the definition of stable white noise in Definition \ref{def.StableProcess}.
Hence, from Definition \ref{def.StableRV}, we obtain that $X_0$ is a stable random variable.

\emph{Proof of \eqref{eqn:hpX0<inf}}:
Since the functions $x\mapsto p_{X_0}(x)$ and $p\mapsto p|\log(1/p)|$ are continuous, the function $x\mapsto p_{X_0}(x) \left|\log(1/p_{X_0}(x))\right|$ is also continuous.
Thus, for any arbitrary $x_0>0$, we have
\begin{equation*}
	\int_{-x_0}^{+x_0}{p_{X_0}(x) \left|\log \tfrac{1}{p_{X_0}(x)}\right| \ud x} < \infty.
\end{equation*}
Therefore, in order to prove \eqref{eqn:hpX0<inf}, we need to prove the boundedness of the tail of the integral.
For sufficiently large $x_0$, we have
\begin{equation} \label{eq.prf.lmm.AlphaStableBound}
	p_{X_0}(x) \leq\frac{c}{|x|^{\alpha+1}} \qquad \forall x>x_0,
\end{equation}
for some positive constant $c$ depending on parameters $(\alpha,\beta,\sigma,\mu)$ \cite{Samor94}.
Since $p\mapsto p|\log({1}/{p})|$ is increasing for $p\in [0, 1/\ue]$, it suffices to show that for sufficiently large $x_0$, the following integral is bounded:
\begin{equation} \label{eq.prf.lmm.AlphaStableRenyiCondition}
	\int_{|x|>x_0}{\frac{c}{|x|^{\alpha+1}}\log{\tfrac{|x|^{\alpha+1}}{c}}\ud x} <\infty.
\end{equation}
By changing variable $y={x}/{c^{{1}/(\alpha+1)}}$, it suffices to show that 
\begin{equation}
	\int_{|y|>y_0}\frac{1}{|y|^{\alpha+1}} \log{|y|^{\alpha+1}}\ud y <\infty,
\end{equation}
where $y_0={x_0}/{c^{{1}/(\alpha+1)}}$, which holds.

\emph{Proof of \eqref{eqn:HpX0<inf}}:
To show that $\H{\quant{X_0}{1}} < \infty$, it suffices to look at the tail of the probability sequence of $\quant{X_0}{1}$.
From \eqref{eq.prf.lmm.AlphaStableBound}, for sufficiently large $x_0$, we have that for all $|m|>x_0$
\begin{align}
	\int_{m-\frac{1}{2}}^{m+\frac{1}{2}}{p_{X_0}(x)\ud x}
	\leq& \int_{m-\frac{1}{2}}^{m+\frac{1}{2}}{\frac{a}{|x|^{\alpha+1}} \ud x} 
	< \tfrac{a}{\left(|m|-\frac{1}{2}\right)^{\alpha+1}}. \label{eq.prf.lmm.AlphaStableBound2}
\end{align}
Hence, it suffices to show that
\if@twocolumn
\begin{align*}
	&\sum_{|m|-\frac{1}{2}>x_0}{\tfrac{a}{\left(|m|-\frac{1}{2}\right)^{\alpha+1}}\log \tfrac{\left(|m|-\frac{1}{2}\right)^{\alpha+1}}{a}} \\
	&\qquad=2\sum_{m>x_0+\frac{1}{2}}{\tfrac{a}{\left(m-\frac{1}{2}\right)^{\alpha+1}}\log \tfrac{\left(m-\frac{1}{2}\right)^{\alpha+1}}{a}} 
	<\infty.
\end{align*}
\else
	\begin{align*}
	&\sum_{|m|-\frac{1}{2}>x_0}{\tfrac{a}{\left(|m|-\frac{1}{2}\right)^{\alpha+1}}\log \tfrac{\left(|m|-\frac{1}{2}\right)^{\alpha+1}}{a}} =2\sum_{m>x_0+\frac{1}{2}}{\tfrac{a}{\left(m-\frac{1}{2}\right)^{\alpha+1}}\log \tfrac{\left(m-\frac{1}{2}\right)^{\alpha+1}}{a}} 
	<\infty.
	\end{align*}
\fi
due that ${a}/{x^{\alpha+1}}\log({x^{\alpha+1}}/{a})$ is a decreasing function for sufficiently large $x$, we obtain that
\if@twocolumn
	\begin{align*}
	&\sum_{m\geq x_0+1} \tfrac{a}{\left(m-\frac{1}{2}\right)^{\alpha+1}}\log \tfrac{\left(m-\frac{1}{2}\right)^{\alpha+1}}{a} \\
	&\qquad<\int_{x>x_0}{\tfrac{a}{x^{\alpha+1}}\log{\tfrac{x^{\alpha+1}}{a}}\ud x} 
	<\infty,
	\end{align*}
\else
	\begin{align*}
	&\sum_{m\geq x_0+1} \tfrac{a}{\left(m-\frac{1}{2}\right)^{\alpha+1}}\log \tfrac{\left(m-\frac{1}{2}\right)^{\alpha+1}}{a} 
	<\int_{x>x_0}{\tfrac{a}{x^{\alpha+1}}\log{\tfrac{x^{\alpha+1}}{a}}\ud x} 
	<\infty,
	\end{align*}
\fi
where the last inequality is true because of \eqref{eq.prf.lmm.AlphaStableRenyiCondition}.
\qed

%%%%%%%%%%%%%%%%%%%%%%%%%%%%%%%
% Proof of Lemma: Distribution of the Sample of Stable Innovations
\subsection{Proof of Lemma \ref{lmm:StableXnDist}} \label{subsec:prf:lmm:StableXnDist}
	Using Definition \ref{def.StableProcess}, the characteristic function of $X_1^{(n)}$ is as follows:
	\begin{equation*}
		\widehat{p}_{X_1^{(n)}}(\omega)
		=\left\lbrace
		\begin{array}{ll}
			\ue^{-\sigma^\alpha|\omega|^\alpha\left(\frac{1}{n}-\uj\frac{1}{n}\beta \sgn{\omega} \tan\frac{\pi\alpha}{2}\right)+\uj\mu\frac{\omega}{n}} & \alpha\neq 1\\
			\ue^{-\sigma|\omega|\left(\frac{1}{n}+\uj\frac{1}{n}\beta \frac{2}{\pi}\sgn\omega\ln{|\omega|}\right)+\uj\mu\frac{\omega}{n}} & \alpha=1
		\end{array}\right..
	\end{equation*}
	From the definition of $X_0$, we can write
	\begin{equation*}
		\widehat{p}_{X_0}(\omega)
		=\left\lbrace
		\begin{array}{ll}
			\ue^{-\sigma^\alpha|\omega|^\alpha\left(1-\uj\beta \sgn{\omega}\tan\frac{\pi\alpha}{2}\right)+\uj\mu\omega} & \alpha\neq 1\\
			\ue^{-\sigma|\omega|\left(1+\uj\beta \frac{2}{\pi}\sgn\omega\ln{|\omega|}\right)+\uj\mu\omega} & \alpha=1
		\end{array}\right..
	\end{equation*}
	Thus, we obtain that
	\begin{equation*}
		\widehat{p}_{X_1^{(n)}}(\omega)
		=\widehat{p}_{X_0}\left(\frac{\omega}{\sqrt[\alpha]{n}}\right)
		\ue^{-\uj\omega c_n}
	\end{equation*}
	where
	\begin{equation*}
		c_n =\left\lbrace
		\begin{array}{ll}
			\mu\left(\frac{1}{\sqrt[\alpha]{n}}-\frac{1}{n}\right) &\alpha\neq 1\\
			\frac{2}{\pi}\sigma\beta\frac{\ln n}{n} &\alpha=1
		\end{array}
		\right.,
	\end{equation*}
	Therefore, $X_1^{(n)}$ can be written with respect to $X_0$ as follows:
	\begin{equation*}
		X_1^{(n)} \stackrel{d}{=} \frac{X_0-b_n}{\sqrt[\alpha]{n}},
	\end{equation*}
	where $b_n=c_n\sqrt[\alpha]{n}$.
\qed

%%%%%%%%%%%%%%%%%%%%
% Proof of Lemma: infinite interval entropy
\subsection{Proof of Lemma \ref{lmm.QuantNorm}} \label{subsec:ProofQuantNorm}
\emph{Proof of \eqref{eq.QuantNorm1}}:
From the definition of $q_{X;m}$ and $q_{Y;m}$, defined in Definition \ref{def.AmpQuant}, it follows that they are constant in intervals $\left[ (i-{1}/{2})/{m}, (i+{1}/{2})/{m}\right]$ for all $i\in\mathbb{Z}$.
Therefore,
\if@twocolumn
	\begin{align*}
	&\int_{\frac{i-\frac{1}{2}}{m}}^{\frac{i+\frac{1}{2}}{m}}
	{\left| q_{Y;m}(x)-q_{X;m}(x) \right| \ud x} \\
	&\qquad=\frac{1}{m}\left| q_{Y;m}\left(\tfrac{i}{m}\right)-q_{X;m}\left(\tfrac{i}{m}\right) \right| \\
	&\qquad=\left| \int_{\frac{i-\frac{1}{2}}{m}}^{\frac{i+\frac{1}{2}}{m}}{p_Y(x) \ud x}
	-\int_{\frac{i-\frac{1}{2}}{m}}^{\frac{i+\frac{1}{2}}{m}}{p_X(x) \ud x} \right| \EQnum\label{eq.QuantNorm1.1}\\
	&\qquad\leq\int_{\frac{i-\frac{1}{2}}{m}}^{\frac{i+\frac{1}{2}}{m}}
	{\left| p_Y(x)-p_X(x) \right| \ud x},
	\end{align*}
\else
	\begin{align*}
	\int_{\frac{i-\frac{1}{2}}{m}}^{\frac{i+\frac{1}{2}}{m}}
	{\left| q_{Y;m}(x)-q_{X;m}(x) \right| \ud x}
	=&\frac{1}{m}\left| q_{Y;m}\left(\tfrac{i}{m}\right)-q_{X;m}\left(\tfrac{i}{m}\right) \right| \\
	=&\left| \int_{\frac{i-\frac{1}{2}}{m}}^{\frac{i+\frac{1}{2}}{m}}{p_Y(x) \ud x}
	-\int_{\frac{i-\frac{1}{2}}{m}}^{\frac{i+\frac{1}{2}}{m}}{p_X(x) \ud x} \right| \EQnum\label{eq.QuantNorm1.1}\\
	\leq&\int_{\frac{i-\frac{1}{2}}{m}}^{\frac{i+\frac{1}{2}}{m}}
	{\left| p_Y(x)-p_X(x) \right| \ud x},
	\end{align*}
\fi
where \eqref{eq.QuantNorm1.1} is from the definition of $q$.
Hence by summation over $i\in\mathbb{Z}$, we have
\begin{equation*}
	\int_{\mathbb{R}}{\left| q_{Y;m}(x)-q_{X;m}(x) \right| \ud x}
	\leq \int_\mathbb{R}{\left| p_Y(x)-p_X(x) \right| \ud x}.
\end{equation*}
Thus, \eqref{eq.QuantNorm1} is proved.

\emph{Proof of \eqref{eq.QuantNorm2} and \eqref{eq.QuantNorm3}}:
We claim that, it suffices to show that
\begin{equation}
	\big|\widetilde{X}_m - X\big| \leq \frac{1}{m} \label{eq.QuantNorm2.1}.
\end{equation}
Because if so, we can write
\begin{align*}
	&\big| |X| - \tfrac{1}{m} \big|^\alpha
	\leq \big|\widetilde{X}_m\big|^\alpha
	\leq \big( |X| + \tfrac{1}{m} \big)^\alpha,
\end{align*}
and as a result
\begin{align*}
	 & \E{\left| |X| - \tfrac{1}{m} \right|^\alpha}
	\leq \E{\big|\widetilde{X}_m\big|^\alpha}
	\leq \E{\left( |X| + \tfrac{1}{m} \right)^\alpha}.
\end{align*}
Therefore, it suffices to show that
\begin{equation} \label{eqn:QuantNormUBnd}
	\E{\left( |X| + \tfrac{1}{m} \right)^\alpha}
	\leq\left(\tfrac{2}{\sqrt{m}}\right)^\alpha+\ue^\frac{\alpha}{\sqrt{m}}\E{|X|^\alpha},
\end{equation}
\begin{equation} \label{eqn:QuantNormLBnd}
	\E{\left| |X| - \tfrac{1}{m} \right|^\alpha}
	\geq\Pr{|X|>\tfrac{1}{\sqrt{m}}}\ue^{\frac{-2\alpha}{\sqrt{m}}}
	\E{|X|^\alpha}.
\end{equation}

\emph{Proof of \eqref{eqn:QuantNormUBnd}}:
By conditioning whether $|X|>1/\sqrt{m}$ or not, we can write
\if@twocolumn
	\begin{align*}
	&\E{\left( |X| + \tfrac{1}{m} \right)^\alpha}\\
	&\quad=\Pr{|X|>\tfrac{1}{\sqrt{m}}}\E{\left( |X| + \tfrac{1}{m} \right)^\alpha\left||X|>\tfrac{1}{\sqrt{m}} \right.}\\
	&\qquad+\Pr{|X|\leq\tfrac{1}{\sqrt{m}}}\E{\left(|X|+\tfrac{1}{m}\right)^\alpha\left||X|\leq\tfrac{1}{\sqrt{m}} \right.}\\
	&\quad\leq \Pr{|X|>\tfrac{1}{\sqrt{m}}}\E{\left( |X| + \tfrac{1}{m} \right)^\alpha\left||X|>\tfrac{1}{\sqrt{m}} \right.}\nonumber
\\&\qquad
	+\left(\tfrac{2}{\sqrt{m}}\right)^\alpha, \EQnum\label{eqn:QuantNormU1}
	\end{align*}
\else
	\begin{align*}
	\E{\left( |X| + \tfrac{1}{m} \right)^\alpha}
	=&\Pr{|X|>\tfrac{1}{\sqrt{m}}}\E{\left( |X| + \tfrac{1}{m} \right)^\alpha\left||X|>\tfrac{1}{\sqrt{m}} \right.}\\
	&+\Pr{|X|\leq\tfrac{1}{\sqrt{m}}}\E{\left(|X|+\tfrac{1}{m}\right)^\alpha\left||X|\leq\tfrac{1}{\sqrt{m}} \right.}\\
	\leq& \Pr{|X|>\tfrac{1}{\sqrt{m}}}\E{\left( |X| + \tfrac{1}{m} \right)^\alpha\left||X|>\tfrac{1}{\sqrt{m}} \right.}
	+\left(\tfrac{2}{\sqrt{m}}\right)^\alpha, \EQnum\label{eqn:QuantNormU1}
	\end{align*}
\fi
where \eqref{eqn:QuantNormU1} is true since
\begin{equation*}
	\Pr{|X|\leq\tfrac{1}{\sqrt{m}}} \leq 1,
	\qquad\tfrac{1}{m}+\tfrac{1}{\sqrt{m}} \leq \tfrac{2}{\sqrt{m}}.
\end{equation*}
Note that if $\Pr{|X|>\tfrac{1}{\sqrt{m}}}=0$, the conditional expected value $\mathbb{E}\big[\big( |X| + \tfrac{1}{m} \big)^\alpha\big| |X|>\tfrac{1}{\sqrt{m}} \big]$ is not well-defined, but in this case we take the product $\Pr{|X|>\tfrac{1}{\sqrt{m}}} \mathbb{E}\big[\big( |X| + \tfrac{1}{m} \big)^\alpha\big| |X|>\tfrac{1}{\sqrt{m}} \big]$ to be zero, without any need for specifying an exact value for $\mathbb{E}\big[\big( |X| + \tfrac{1}{m} \big)^\alpha\big| |X|>\tfrac{1}{\sqrt{m}} \big]$.  

From \eqref{eqn:QuantNormU1}, it only remains to prove that
\begin{align*}
	\Pr{|X| \hspace{-0.8mm}>\hspace{-0.8mm} \tfrac{1}{\sqrt{m}}}
	\E{\left( |X| + \tfrac{1}{m} \right)^\alpha\left||X| \hspace{-0.8mm}>\hspace{-0.8mm} \tfrac{1}{\sqrt{m}} \right.}
	\leq \ue^\frac{\alpha}{\sqrt{m}}\E{|X|^\alpha}.
\end{align*}
In order to do that, we can write
\begin{equation*}
	\left( |X| + \tfrac{1}{m} \right)^\alpha
	=|X|^\alpha \ue^{\alpha\log\left(1+\frac{1}{m|X|}\right)}.
\end{equation*}
Since $|X|>1/\sqrt{m}$, we can write
\if@twocolumn
\begin{align*}
	|X|^\alpha \ue^{\alpha\log\left(1+\frac{1}{m|X|}\right)}
	\leq& |X|^\alpha \ue^{\alpha\log\left(1+\frac{1}{\sqrt{m}}\right)} \\
	\leq& |X|^\alpha \ue^{\alpha\frac{1}{\sqrt{m}}},
\end{align*}
\else
\begin{align*}
	|X|^\alpha \ue^{\alpha\log\left(1+\frac{1}{m|X|}\right)}
	\leq |X|^\alpha \ue^{\alpha\log\left(1+\frac{1}{\sqrt{m}}\right)} 
	\leq |X|^\alpha \ue^{\alpha\frac{1}{\sqrt{m}}},
\end{align*}
\fi
where the last inequality is true because $\log(1+x)\leq x$. As a result
\if@twocolumn
	\begin{align*}
	&\Pr{|X|>\tfrac{1}{\sqrt{m}}}\E{\left( |X| + \tfrac{1}{m} \right)^\alpha\left||X|>\tfrac{1}{\sqrt{m}} \right.}\\
	&\qquad\leq \Pr{|X|>\tfrac{1}{\sqrt{m}}}\E{|X|^\alpha \ue^{\alpha\frac{1}{\sqrt{m}}}\left||X|>\tfrac{1}{\sqrt{m}} \right.}\\
	&\qquad\leq \ue^\frac{\alpha}{\sqrt{m}}\E{|X|^\alpha}.
	\end{align*}
\else
	\begin{align*}
	\Pr{|X|>\tfrac{1}{\sqrt{m}}}\E{\left( |X| + \tfrac{1}{m} \right)^\alpha\left||X|>\tfrac{1}{\sqrt{m}} \right.}
	\leq& \Pr{|X|>\tfrac{1}{\sqrt{m}}}\E{|X|^\alpha \ue^{\alpha\frac{1}{\sqrt{m}}}\left||X|>\tfrac{1}{\sqrt{m}} \right.}\\
	\leq& \ue^\frac{\alpha}{\sqrt{m}}\E{|X|^\alpha}.
	\end{align*}
\fi

\emph{Proof of \eqref{eqn:QuantNormLBnd}}:
Similar to the previous proof, by conditioning whether $|X|>1/\sqrt{m}$ or not, we can write
\if@twocolumn
	\begin{align*}
	&\E{\left| |X|  \hspace{-0.8mm}-\hspace{-0.8mm} \tfrac{1}{m} \right|^\alpha}
	=\Pr{|X| \hspace{-0.8mm}>\hspace{-0.8mm} \tfrac{1}{\sqrt{m}}}\E{\left| |X| \hspace{-0.8mm}-\hspace{-0.8mm} \tfrac{1}{m} \right|^\alpha\left||X| \hspace{-0.8mm}>\hspace{-0.8mm} \tfrac{1}{\sqrt{m}} \right.}\\
	&~~~\phantom{\geq} +\Pr{|X|\leq\tfrac{1}{\sqrt{m}}}\E{\left||X|-\tfrac{1}{m}\right|^\alpha\left||X|\leq\tfrac{1}{\sqrt{m}} \right.}\\
	&~~~\geq \phantom{+} \Pr{|X|>\tfrac{1}{\sqrt{m}}}
	\E{\left| |X|-\tfrac{1}{m} \right|^\alpha\left||X|>\tfrac{1}{\sqrt{m}} \right.}, \EQnum\label{eqn:QuantNormL1}
	\end{align*}
\else
	\begin{align*}
	\E{\left| |X| - \tfrac{1}{m} \right|^\alpha}
	 =& \Pr{|X|>\tfrac{1}{\sqrt{m}}}\E{\left| |X|-\tfrac{1}{m} \right|^\alpha\left||X|>\tfrac{1}{\sqrt{m}} \right.}\\
	&+\Pr{|X|\leq\tfrac{1}{\sqrt{m}}}\E{\left||X|-\tfrac{1}{m}\right|^\alpha\left||X|\leq\tfrac{1}{\sqrt{m}} \right.}\\
	\geq&\Pr{|X|>\tfrac{1}{\sqrt{m}}}
	\E{\left| |X|-\tfrac{1}{m} \right|^\alpha\left||X|>\tfrac{1}{\sqrt{m}} \right.}, \EQnum\label{eqn:QuantNormL1}
	\end{align*}
\fi
If $\Pr{|X|>1/\sqrt{m}}=0$, \eqref{eqn:QuantNormLBnd} is clearly correct. 
Thus, assume that $\Pr{|X|>1/\sqrt{m}}>0$, meaning that $\E{| |X|-{1}/{m} |^\alpha \big| |X|>{1}/{\sqrt{m}} }$ is well-defined. We need to show that
\begin{equation*}
	\E{\left| |X| - \tfrac{1}{m} \right|^\alpha\left||X|>\tfrac{1}{\sqrt{m}} \right.}
	\geq \ue^{-\frac{\alpha}{\sqrt{m}}}\E{|X|^\alpha}.
\end{equation*}
Observe that when $|x|>1/\sqrt{m}$ and $m\geq4$, we can write
\if@twocolumn
\begin{align*}
	\left| |x| - \tfrac{1}{m} \right|^\alpha
	=&|x|^\alpha \ue^{\alpha\log\left|1-\frac{1}{m|x|}\right|}
	\geq |x|^\alpha \ue^{\alpha\log\left(1-\frac{1}{\sqrt{m}}\right)} \\
	\geq& |x|^\alpha \ue^{-\alpha\frac{2}{\sqrt{m}}},
\end{align*}
\else
\begin{align*}
	\left| |x| - \tfrac{1}{m} \right|^\alpha
	=|x|^\alpha \ue^{\alpha\log\left|1-\frac{1}{m|x|}\right|}
	\geq |x|^\alpha \ue^{\alpha\log\left(1-\frac{1}{\sqrt{m}}\right)} 
	\geq |x|^\alpha \ue^{-\alpha\frac{2}{\sqrt{m}}}, 
\end{align*}
\fi
where the last inequality is true because $\log(1\hspace{-0.5mm}-\hspace{-0.5mm}x)\geq \hspace{-0.5mm}-\hspace{-0.5mm}2x$ for $0\leq x\leq1/2$. As a result,
\begin{equation*}
	\E{\left| |X| - \tfrac{1}{m} \right|^\alpha\left||X|>\tfrac{1}{\sqrt{m}} \right.}
	\geq \ue^{-\frac{\alpha}{\sqrt{m}}}\E{|X|^\alpha \big| |X|>\tfrac{1}{\sqrt{m}}}.
\end{equation*}
Now, we need to show that
\begin{equation} \label{eqn:EX>c>EX}
	\E{|X|^\alpha \Big| |X|>\tfrac{1}{\sqrt{m}}}
	\geq\E{|X|^\alpha}.
\end{equation}
Without loss of generality, we may assume that $\Pr{|X|>1/\sqrt{m}}<1$. Now, we have
\if@twocolumn
	\begin{align*}
	&\E{|X|^\alpha} 
	= \Pr{|X|>\tfrac{1}{\sqrt{m}}}
	\E{|X|^\alpha \Big| |X|>\tfrac{1}{\sqrt{m}}} \EQnum\label{eqn:EX<maxEX>c,EX<c}\\
	&\qquad\quad+ \Pr{|X|\leq\tfrac{1}{\sqrt{m}}}
	\E{|X|^\alpha \Big| |X|\leq\tfrac{1}{\sqrt{m}}}  \\
	&\quad\leq \max\left\{
	\E{|X|^\alpha \Big| |X|>\tfrac{1}{\sqrt{m}}},
	\E{|X|^\alpha \Big| |X|\leq\tfrac{1}{\sqrt{m}}}
	\right\}.
	\end{align*}
\else
	\begin{align*}
	\E{|X|^\alpha} 
	=& \Pr{|X|>\tfrac{1}{\sqrt{m}}}
	\E{|X|^\alpha \Big| |X|>\tfrac{1}{\sqrt{m}}} \EQnum\label{eqn:EX<maxEX>c,EX<c}
	+ \Pr{|X|\leq\tfrac{1}{\sqrt{m}}}
	\E{|X|^\alpha \Big| |X|\leq\tfrac{1}{\sqrt{m}}}  \\
	\leq& \max\left\{
	\E{|X|^\alpha \Big| |X|>\tfrac{1}{\sqrt{m}}},
	\E{|X|^\alpha \Big| |X|\leq\tfrac{1}{\sqrt{m}}}
	\right\}.
	\end{align*}
\fi
Since
\begin{equation*}
	\E{|X|^\alpha \Big| |X|\leq\tfrac{1}{\sqrt{m}}}
	\leq \left(\tfrac{1}{\sqrt{m}}\right)^\alpha
	<\E{|X|^\alpha \Big| |X|>\tfrac{1}{\sqrt{m}}},
\end{equation*}
we obtain that
\if@twocolumn
	\begin{align*}
	&\max\left\{
	\E{|X|^\alpha \Big| |X|>\tfrac{1}{\sqrt{m}}},
	\E{|X|^\alpha \Big| |X|\leq\tfrac{1}{\sqrt{m}}}
	\right\}\\
	&\qquad=\E{|X|^\alpha \Big| |X|>\tfrac{1}{\sqrt{m}}}.
	\end{align*}
\else
\begin{align*}
	&\max\left\{
	\E{|X|^\alpha \Big| |X|>\tfrac{1}{\sqrt{m}}},
	\E{|X|^\alpha \Big| |X|\leq\tfrac{1}{\sqrt{m}}}
	\right\}
	=\E{|X|^\alpha \Big| |X|>\tfrac{1}{\sqrt{m}}}.
\end{align*}
\fi
Therefore, \eqref{eqn:EX>c>EX} is obtained from \eqref{eqn:EX<maxEX>c,EX<c}.
Hence, \eqref{eqn:QuantNormLBnd} is proved.

\emph{Proof of \eqref{eq.QuantNorm2.1}}:
Note that random variable $\widetilde{X}_m$ has the same distribution as the following random variable:
\begin{equation*}
	\widetilde{X}_m = \quant{X}{m} + U_m,
\end{equation*}
where $U_m$ and $X$ are independent, and
\begin{equation*}
	p_{U_m}(x) = \left\lbrace
	\begin{array}{ll}
		m & |x|\leq\frac{1}{2m} \\
		0 & |x|>\frac{1}{2m}
	\end{array}\right..
\end{equation*}
Thus,
\if@twocolumn
	\begin{align*}
	\big|\widetilde{X}_m - X\big|
	\leq& \big| \widetilde{X}_m-\quant{X}{m} \big|
	+ \left| \quant{X}{m}-X \right| \\
	\leq& \left| U_m \right|+\tfrac{1}{2m}\leq\tfrac{1}{m},
	\end{align*}
\else
	\begin{align*}
	\big|\widetilde{X}_m - X\big|
	\leq& \big| \widetilde{X}_m-\quant{X}{m} \big|
	+ \left| \quant{X}{m}-X \right| 
	\leq \left| U_m \right|+\tfrac{1}{2m}\leq\tfrac{1}{m},
	\end{align*}
\fi
where the last inequality is true because $\left| \quant{X}{m}-X \right|$ and $|U_m|$ are always less than ${1}/{2m}$.
Therefore, the lemma is proved.
\qed

%%%%%%%%%%%%%
% Proof of Lemma: H([aX])
\subsection{Proof of Lemma \ref{lmm.H([aX])}} \label{subsec:ProofH([aX]}
From the definition of entropy, we know that the entropy of a random variable does not depend on the value of the random variable, rather it only depends on the distribution of the random variable.
Therefore, if one finds a correspondence between the values of $\quant{aX}{m/a}$ and the values of $\quant{X}{m}$, while they have the same probability, then the entropy of them will be the same.
Thus, we define the following correspondence between the values of $\quant{aX}{m/a}$, which are from the set $\left\lbrace k\frac{a}{m}\left| k\in\mathbb{Z}\right.\right\rbrace$, and the values of $\quant{X}{m}$, which are from the set $\left\lbrace k\frac{1}{m}\left| k\in\mathbb{Z}\right.\right\rbrace$.
\begin{equation*}
	k\tfrac{a}{m}\in\left\lbrace k\tfrac{a}{m}\left| k\in\mathbb{Z}\right.\right\rbrace
	\leftrightarrow k\tfrac{1}{m}\in\left\lbrace k\tfrac{1}{m}\left| k\in\mathbb{Z}\right.\right\rbrace
\end{equation*}
Now, we show that the correspondent values have the same probability.
\if@twocolumn
	\begin{align*}
	\Pr{\quant{aX}{\frac{m}{a}}=k\tfrac{a}{m}} 
	=&\Pr{aX\in\left[\left(k-\tfrac{1}{2}\right)\tfrac{a}{m},\left(k+\tfrac{1}{2}\right)\tfrac{a}{m} \right)}\\
	=&\Pr{X\in\left[\left(k-\tfrac{1}{2}\right)\tfrac{1}{m},\left(k+\tfrac{1}{2}\right)\tfrac{1}{m}\right)}\\
	=&\Pr{\quant{X}{m}=k\tfrac{1}{m}}.
	\end{align*}
\else
	\begin{align*}
	\Pr{\quant{aX}{\frac{m}{a}}=k\tfrac{a}{m}} 
	&=\Pr{aX\in\left[\left(k-\tfrac{1}{2}\right)\tfrac{a}{m},\left(k+\tfrac{1}{2}\right)\tfrac{a}{m} \right)}\\
	&=\Pr{X\in\left[\left(k-\tfrac{1}{2}\right)\tfrac{1}{m},\left(k+\tfrac{1}{2}\right)\tfrac{1}{m}\right)}
	=\Pr{\quant{X}{m}=k\tfrac{1}{m}}.
	\end{align*}
\fi
Therefore, the lemma is proved.
\qed

%%%%%%%%%%%%%%%%%%%%%%%%%%%%%%
% Proof of  Lemma: Entropy dimension for shifted sontinuous RVs
\subsection{Proof of Lemma \ref{lmm.DifEntShft}} \label{subsec:ProofDifEntShft}
% First Assertion
Observe that for any given $m,c_m\in\mathbb{R}$, there exist unique $k_m\in\mathbb{Z}$, and $d_m\in\mathbb{R}$ such that
\begin{equation*}
	c_m=\tfrac{k_m}{m}+d_m,\qquad d_m\in \left[0,\tfrac{1}{m}\right).
\end{equation*}
Therefore, because of the definition of amplitude quantization in Definition \ref{def.AmpQuant}, we can write
\begin{equation*}
	\quant{X+c_m}{m}=\quant{X+d_m}{m}+\tfrac{k_m}{m}.
\end{equation*}
Since entropy is invariant with respect to constant shift, we have
\begin{equation*}
	\H{\quant{X+c_m}{m}}=\H{\quant{X+d_m}{m}}.
\end{equation*}

As a result, we only need to prove the theorem for $\{d_m\}$, instead of $\{c_m\}$.
Take the pdf $q_m(x)$ as follows:
\if@twocolumn
\begin{align}
	q_m(x) =& m\Pr{X+d_m\in\left[\tfrac{i-\frac{1}{2}}{m},\tfrac{i+\frac{1}{2}}{m}\right)} \nonumber\\
	=& m\int_{\frac{i-\frac{1}{2}}{m}-d_m}^{\frac{i+\frac{1}{2}}{m}-d_m}{p(x) \ud x}, \label{eq:qm-Pm}
\end{align}
\else
\begin{align}
	q_m(x) = m\Pr{X+d_m\in\left[\tfrac{i-\frac{1}{2}}{m},\tfrac{i+\frac{1}{2}}{m}\right)} 
	= m\int_{\frac{i-\frac{1}{2}}{m}-d_m}^{\frac{i+\frac{1}{2}}{m}-d_m}{p(x) \ud x}, \label{eq:qm-Pm}
\end{align}
\fi
Therefore, from Lemma \ref{lmm:H([X])-ln m=h(q)}, we obtain that
\if@twocolumn
	\begin{align*}
	\int_\bR{q_m(x)\log\frac{1}{q_m(x)} \ud x}
	=&\sum_{i\in\bZ}{\tfrac{ q_m(i/m) }{m} \log\tfrac{1}{\frac{1}{m} q_m\left(\frac{i}{m}\right)}} - \log m\\
	=&\H{\quant{X+d_m}{m}}-\log m,
	\end{align*}
\else
	\begin{align*}
	\int_\bR{q_m(x)\log\frac{1}{q_m(x)} \ud x}
	=\H{\quant{X+d_m}{m}}-\log m,
	\end{align*}
\fi
where the last equation is true because of \eqref{eq:qm-Pm}.
So, if we take $\widetilde{X}_m$ a continuous random variable with pdf $q_m$, we only need to prove
\begin{equation*}
	\lim_{m\to\infty}\h{\widetilde{A}_m} = \h{X}.
\end{equation*}
In order to do this, we write $\left|\h{\widetilde{A}_m}-\h{X}\right|$ as follows
\if@twocolumn
\begin{align*}
	\left|\h{\widetilde{A}_m} \hspace{-0.8mm}-\hspace{-0.8mm} \h{X}\right|
	\leq&\left|\int_{-l}^{l}{q_m(x)\log\tfrac{1}{q_m(x)} \hspace{-0.7mm}-\hspace{-0.7mm} p(x)\log\tfrac{1}{p(x)} \ud x}\right|\\
	&+ \left|\int_{|x|>l}{q_m(x)\log\tfrac{1}{q_m(x)}\ud x}\right|\\
	&+ \left|\int_{|x|>l}{p(x)\log\tfrac{1}{p(x)} \ud x}\right|,
\end{align*}
\else
\begin{align*}
	\left|\h{\widetilde{A}_m} \hspace{-0.8mm}-\hspace{-0.8mm} \h{X}\right|
	\leq & \left|\int_{-l}^{l}{q_m(x)\log\tfrac{1}{q_m(x)} \hspace{-0.7mm}-\hspace{-0.7mm} p(x)\log\tfrac{1}{p(x)} \ud x}\right| + \left|\int_{|x|>l}{q_m(x)\log\tfrac{1}{q_m(x)}\ud x}\right|\\
	&+ \left|\int_{|x|>l}{p(x)\log\tfrac{1}{p(x)} \ud x}\right|,
\end{align*}
\fi
where $l>0$ is arbitrary.
Thus, it suffices to prove that
\begin{align}
	&\lim_{m\to\infty}\int_{-l}^{l}{q_m(x)\log\tfrac{1}{q_m(x)} \ud x}
	=\int_{-l}^{l}{p(x)\log\tfrac{1}{p(x)} \ud x},\label{eq:Bnd1}
\end{align}
\begin{equation}
	\lim_{l\to\infty}{\left|\int_{|x|>l}{p(x)\log\tfrac{1}{p(x)} \ud x}\right|}=0,\label{eq:Bnd2}
\end{equation}
\begin{align}
	&\lim_{l\to\infty}{\left|\int_{|x|>l}{q_m(x)\log\tfrac{1}{q_m(x)}\ud x}\right|}=0,
	\label{eq:Bnd3}
\end{align}
hold  for all $l>0$ and uniformly on $m$.

\emph{Proof of \eqref{eq:Bnd1}}:
From \eqref{eq:qm-Pm}, we obtain that since $p(x)$ is a piecewise continuous, the mean value theorem yields that there exists $x^*_m\in\left[(i-\tfrac{1}{2})/m-d_m,(i-\tfrac{1}{2})/m-d_m\right)$, such that $q_m(x) = p(x^*_m)$.
Since $p(x)$ is piecewise continuous, and $[-l,l]$ is a compact set, $p(x)$ is uniformly continuous over $[-l,l]$.
Therefore, for any $\epsilon'$, there exists $M\in\bR$ such that for all $x\in[-l,l]$, we have that
\begin{equation*}
	m>M
	~\Rightarrow ~ |q_m(x)-p(x)|<\epsilon'.
\end{equation*}
Furthermore, since the function $x\mapsto x\log x$ is continuous and $p(x)$ is uniformly continuous, we have that for all $x\in[-l,l]$
\begin{equation*}
	\lim_{m\to\infty}{q_m(x)\log \tfrac{1}{q_m(x)}}=p(x)\log \tfrac{1}{p(x)},
\end{equation*}
uniformly on $x$.
Thus \eqref{eq:Bnd1} is proved.

\emph{Proof of \eqref{eq:Bnd2}}:
We can write
\begin{align}
	&\left|\int_{|x|>l}{p(x)\log\tfrac{1}{p(x)} \ud x}\right|
	\leq \int_{|x|>l}{p(x)\left|\log\tfrac{1}{p(x)}\right| \ud x}. \label{eq:Bnd2.1}
\end{align}
By assuming the lemma, we know
\begin{align}
	\int_\bR \hspace{-0.5mm} p(x)\left|\log\tfrac{1}{p(x)}\right| \hspace{-0.8mm} \ud x \hspace{-0.8mm}<\hspace{-0.8mm} \infty
	\Rightarrow\lim_{l\to\infty}\int_{|x|>l} \hspace{-1mm} p(x)\left|\log\tfrac{1}{p(x)}\right| \hspace{-0.8mm} \ud x=0.\label{eq:LhpX=0}
\end{align}
Hence, \eqref{eq:Bnd2.1} implies \eqref{eq:Bnd2}.

\emph{Proof of \eqref{eq:Bnd3}}:
It suffices to show that
\if@twocolumn
\begin{align}
	&\int_{|x+d_m|>l} \hspace{-1mm} p(x)\log\tfrac{1}{p(x)} \ud x 
	\leq\int_{|x|>l} \hspace{-1mm} q_m(x)\log\tfrac{1}{q_m(x)}\ud x, \label{eq:Bnd3.1}\\
	&\int_{|x|>l}{q_m(x)\log\tfrac{1}{q_m(x)}\ud x} \nonumber\\
	&\leq \sum_{|i|>l-1} \hspace{-2mm} P[i]\log\tfrac{1}{P[i]} + \Pr{|X|>l \hspace{-0.5mm}-\hspace{-0.5mm} 1}\log 2+o_l(1), \label{eq:Bnd3.2}
\end{align}
\else
\begin{align}
	&\int_{|x+d_m|>l}  p(x)\log\tfrac{1}{p(x)} \ud x 
	\leq\int_{|x|>l}  q_m(x)\log\tfrac{1}{q_m(x)}\ud x, \label{eq:Bnd3.1}\\
	&\int_{|x|>l}{q_m(x)\log\tfrac{1}{q_m(x)}\ud x} 
	\leq \sum_{|i|>l-1}{P[i]\log\tfrac{1}{P[i]}}
	+\Pr{|X|>l-1}\log 2+o_l(1), \label{eq:Bnd3.2}
\end{align}
\fi
where $o_l(1)$ means that $\lim_{l\to\infty}{o_l(1)}=0$, and
\begin{equation*}
	P[i]=\int_{i-\frac{1}{2}}^{i+\frac{1}{2}}{p(x)\ud x}.
\end{equation*}
By changing the variable $y=Lx$, we can write
\if@twocolumn
	\begin{align*}
	\int_{|x+d_m|>l}{p(x)\log\tfrac{1}{p(x)} \ud x} 
	=&\int_{|x+d_m|>l}{p(x)\log\tfrac{1}{\frac{1}{L}p(x)} \ud x} \\
	&-\Pr{|X+d_m|>l} \log L,\nonumber\\
	=\int_{|y+Ld_m|>lL} & {\tfrac{1}{L}p\left(\tfrac{y}{L}\right)\log\tfrac{1}{\frac{1}{L}p\left(\tfrac{y}{L}\right)} \ud y} \nonumber\\
	&-\Pr{|X+d_m|>l} \log L,
	\end{align*}
\else
	\begin{align*}
	\int_{|x+d_m|>l}{p(x)\log\tfrac{1}{p(x)} \ud x}
	=&\int_{|x+d_m|>l}{p(x)\log\tfrac{1}{\frac{1}{L}p(x)} \ud x}
	-\Pr{|X+d_m|>l} \log L,\nonumber\\
	=&\int_{|y+Ld_m|>lL}{\tfrac{1}{L}p\left(\tfrac{y}{L}\right)\log\tfrac{1}{\tfrac{1}{L}p\left(\tfrac{y}{L}\right)} \ud y}
	-\Pr{|X+d_m|>l} \log L,
	\end{align*}
\fi
Because $p(x)<L$ almost everywhere, we have that:
\begin{equation*}
	\frac{1}{L}p\left(\tfrac{y}{L}\right)<1
	~\Rightarrow ~ \frac{1}{L}p\left(\tfrac{y}{L}\right)\log\tfrac{1}{\frac{1}{L}p\left(\frac{y}{L}\right)}>0.
\end{equation*}
So we obtain that
\if@twocolumn
	\begin{align*}
	&\int_{|y+Ld_m|>lL}{\tfrac{1}{L}p\left(\tfrac{y}{L}\right)\log\tfrac{1}{\frac{1}{L}p\left(\frac{y}{L}\right)} \ud y} \\
	&\qquad\geq\int_{|y|>lL+L}{\tfrac{1}{L}p\left(\tfrac{y}{L}\right)\log\tfrac{1}{\frac{1}{L}p\left(\tfrac{y}{L}\right)} \ud y}\\
	&\qquad=\int_{|x|>l+1}{p(x)\log\tfrac{1}{p(x)} \ud x}\\
	&\qquad\quad+\Pr{|X|>l+1} \log L,
	\end{align*}
\else
	\begin{align*}
	\int_{|y+Ld_m|>lL}{\frac{1}{L}p\left(\tfrac{y}{L}\right)\log\tfrac{1}{\frac{1}{L}p\left(\frac{y}{L}\right)} \ud y}
	\geq&\int_{|y|>lL+L}{\frac{1}{L}p\left(\tfrac{y}{L}\right)\log\tfrac{1}{\frac{1}{L}p\left(\frac{y}{L}\right)} \ud y}\\
	=&\int_{|x|>l+1}{p(x)\log\tfrac{1}{p(x)} \ud x}
	+\Pr{|X|>l+1} \log L,
	\end{align*}
\fi
where $x=y/L$.
Therefore, we can write that
\if@twocolumn
	\begin{align*}
	\int_{|x+d_m|>l}{p(x)\log\tfrac{1}{p(x)} \ud x}
	\geq&\int_{|x|>l+1}{p(x)\log\tfrac{1}{p(x)} \ud x}\\
	&-\Pr{|X|\in[l,l+1]} \log L.
	\end{align*}
\else
	\begin{align*}
	\int_{|x+d_m|>l}{p(x)\log\tfrac{1}{p(x)} \ud x}
	\geq&\int_{|x|>l+1}{p(x)\log\tfrac{1}{p(x)} \ud x}
	-\Pr{|X|\in[l,l+1]} \log L.
	\end{align*}
\fi
Thus, from \eqref{eq:LhpX=0} we conclude that the lower bound vanishes as $l$ tends to $\infty$ uniformly on $m$.
Now, we are going to show that \eqref{eq:Bnd3.2} leads that the upper bound vanishes as $l$ tends to $\infty$ uniformly on $m$.
In order to prove this, note that $\Pr{|X|>l-1}$ vanishes as $l$ tends to infinity uniformly on $m$.
Furthermore,
\if@twocolumn
	\begin{align*}
	\H{\quant{X}{1}}<\infty
	~\Rightarrow~ & \sum_{i\in\mathbb{Z}}{P[i]\log\tfrac{1}{P[i]}}<\infty \\
	~\Rightarrow~ & \lim_{l\to\infty}\sum_{|i|>l-1}{P[i]\log\tfrac{1}{P[i]}}=0.
	\end{align*}
\else
	\begin{align*}
	\H{\quant{X}{1}}<\infty
	~\Rightarrow~  \sum_{i\in\mathbb{Z}}{P[i]\log\tfrac{1}{P[i]}}<\infty 
	~\Rightarrow~  \lim_{l\to\infty}\sum_{|i|>l-1}{P[i]\log\tfrac{1}{P[i]}}=0.
	\end{align*}
\fi
Thus, in order to prove \eqref{eq:Bnd3}, it only remains to prove \eqref{eq:Bnd3.1} and \eqref{eq:Bnd3.2}.
The proof of \eqref{eq:Bnd3.1} exists in \cite{Renyi59} in the proof of Theorem \ref{rem.CEntDim}.
Thus, we only need to prove \eqref{eq:Bnd3.2}.
Similar to the proof of Theorem \ref{rem.CEntDim} in \cite{Renyi59}, it can be shown that
\begin{equation*}
	\int_{|x|>l}{q_m(x)\log\tfrac{1}{q_m(x)}\ud x}
	\leq\sum_{|i|>l}{P_m[i]\log\tfrac{1}{P_m[i]}},
\end{equation*}
where $P_m[i]:=\Pr{\quant{X+d_m}{1}=i}$.
This implies
\if@twocolumn
	\begin{align*}
	&\H{\quant{X+d_m}{1}\big|\left|\quant{X+d_m}{1}\right|>l}\\
	&\qquad=\sum_{|i|>l}{\tfrac{P_m[i]}{\Pr{|\quant{X+d_m}{1}|>l}}} 
	\log\tfrac{\Pr{|\quant{X+d_m}{1}|>l}}{P_m[i]},
	\end{align*}
\else
	\begin{align*}
	\H{\quant{X+d_m}{1}\big|\left|\quant{X+d_m}{1}\right|>l}
	=&\sum_{|i|>l}{\tfrac{P_m[i]}{\Pr{|\quant{X+d_m}{1}|>l}}
	\log\tfrac{\Pr{|\quant{X+d_m}{1}|>l}}{P_m[i]}},
	\end{align*}
\fi
where $\H{X|Y=y}$ is defined in \cite[p. 29]{Cover06}.
Hence, we can write
\if@twocolumn
	\begin{align}
	&\sum_{|i|>l}{P_m[i]\log\tfrac{1}{P_m[i]}} \nonumber\\
	&~~=\phantom{+}\Pr{|\quant{X+d_m}{1}|>l} 
	\H{\quant{X+d_m}{1}\big||\quant{X+d_m}{1}|>l} \nonumber\\
	&~~\phantom{=}+\Pr{|\quant{X+d_m}{1}|>l} 
	\log\tfrac{1}{\Pr{|\quant{X+d_m}{1}|>l}}. \label{eq:Bnd3.2.4}
	\end{align}
\else
	\begin{align}
	\sum_{|i|>l}{P_m[i]\log\tfrac{1}{P_m[i]}} 
	=&\Pr{|\quant{X+d_m}{1}|>l}
	\H{\quant{X+d_m}{1}\big||\quant{X+d_m}{1}|>l} \nonumber\\
	&+\Pr{|\quant{X+d_m}{1}|>l}
	\log\tfrac{1}{\Pr{|\quant{X+d_m}{1}|>l}}. \label{eq:Bnd3.2.4}
	\end{align}
\fi
Because of the definition of the quantization, we can write
\begin{equation*}
	\quant{X+d_m}{1}=\quant{X}{1}+E_m,
\end{equation*}
where $E_m$ is a random variable taking values from $\lbrace0,1\rbrace$.
Thus, $\quant{X+d_m}{1}$ is a function of $\quant{X}{1}$ and $E_m$; as a result
\if@twocolumn
	\begin{align*}
	&\H{\quant{X+d_m}{1}\big||\quant{X+d_m}{1}|>l} \\
	&\qquad\leq\H{\quant{X}{1},E_m\big||\quant{X+d_m}{1}|>l}\\
	&\qquad\leq\H{\quant{X}{1}\big||\quant{X+d_m}{1}|>l}+\H{E_m}\\
	&\qquad\leq\H{\quant{X}{1}\big||\quant{X+d_m}{1}|>l}+\log 2.
	\end{align*}
\else
	\begin{align*}
	\H{\quant{X+d_m}{1}\big||\quant{X+d_m}{1}|>l}
	\leq&\H{\quant{X}{1},E_m\big||\quant{X+d_m}{1}|>l}\\
	\leq&\H{\quant{X}{1}\big||\quant{X+d_m}{1}|>l}+\H{E_m}\\
	\leq&\H{\quant{X}{1}\big||\quant{X+d_m}{1}|>l}+\log 2.
	\end{align*}
\fi
From the definition of $\H{\quant{X}{1}\big||\quant{X+d_m}{1}|>l}$, we can write that
\begin{align*}
	&\Pr{|\quant{X+d_m}{1}|>l}\H{\quant{X}{1}\big||\quant{X+d_m}{1}|>l} \\
	&\qquad=\sum_{|i|>l}{P[i]\log\tfrac{1}{P[i]}}
	-r_m\log r_m-s_m\log s_m+t_m\log t_m \\
	&\qquad\quad -\Pr{|\quant{X+d_m}{1}|>l}\log\tfrac{1}{\Pr{|\quant{X+d_m}{1}|>l}},
\end{align*}
where
\if@twocolumn
	\begin{align*}
	r_m =& \Pr{X\in\left(l-\tfrac{1}{2}-d_m,l-\tfrac{1}{2}\right)}, \\
	s_m =& \Pr{X\in\left[-l-\tfrac{1}{2},-l+\tfrac{1}{2}-d_m\right)}, \\
	t_m =& \Pr{X\in\left[-l-\tfrac{1}{2},-l+\tfrac{1}{2}\right)}.
	\end{align*}
\else
	\begin{align*}
	r_m =& \Pr{X\in\left(l-\tfrac{1}{2}-d_m,l-\tfrac{1}{2}\right)},
	\qquad s_m = \Pr{X\in\left[-l-\tfrac{1}{2},-l+\tfrac{1}{2}-d_m\right)},\\
	t_m =& \Pr{X\in\left[-l-\tfrac{1}{2},-l+\tfrac{1}{2}\right)}.
	\end{align*}
\fi
Therefore, \eqref{eq:Bnd3.2.4} can be simplified as follows:
\if@twocolumn
	\begin{align}
	&\sum_{|i|>l}{P_m[i]\log\tfrac{1}{P_m[i]}} \nonumber\\
	&\qquad\leq\sum_{|i|>l}{P[i]\log\tfrac{1}{P[i]}} \label{eqn:UpBndSFT1}\\
	&\qquad\quad-r_m\log r_m-s_m\log s_m+t_m\log t_m \label{eqn:UpBndSFT2}\\
	&\qquad\quad+\Pr{|\quant{X+d_m}{1}|>l} \log 2 \label{eqn:UpBndSFT3}.
	\end{align}
\else
	\begin{align}
	\sum_{|i|>l}{P_m[i]\log\tfrac{1}{P_m[i]}}
	\leq&\sum_{|i|>l}{P[i]\log\frac{1}{P[i]}} \label{eqn:UpBndSFT1}\\
	&-r_m\log r_m-s_m\log s_m+t_m\log t_m \label{eqn:UpBndSFT2}\\
	&+\Pr{|\quant{X+d_m}{1}|>l} \log 2 \label{eqn:UpBndSFT3}.
	\end{align}
\fi
The first term, \eqref{eqn:UpBndSFT1}, vanishes as $l$ tends to $\infty$ because
\begin{align*}
	\H{\quant{X}{1}} \hspace{-0.5mm}=\hspace{-0.8mm} \sum_{i\in\bZ} \hspace{-0.5mm} P[i]\log\tfrac{1}{P[i]} \hspace{-0.5mm}<\hspace{-0.5mm} \infty
	\Rightarrow& \lim_{l\to\infty}\sum_{|i|>l} \hspace{-0.5mm}P[i]\log\tfrac{1}{P[i]} \hspace{-0.5mm} = 0.
\end{align*}
It can be achieved that \eqref{eqn:UpBndSFT2} and \eqref{eqn:UpBndSFT3} vanish as $l$ tends to $\infty$, uniformly on $m$, but we do not write the details here.
Thus, \eqref{eq:Bnd3.2} is proved and the proof of lemma is complete.
\qed

%%%%%%%%%%%%%%%%%%
% Proof of Lemma:
\subsection{Proof of Lemma \ref{lmm:h(X+Y)c}} \label{subsec:prf:lmm:h(X+Y)c}
	According to Lemma \ref{lmm:DC+DC}, we can define random variable $U$ with support $\{\mathrm{dc},\mathrm{cd},\mathrm{cc}\}$ and pmf $p_{U}(u)$ such that
	\begin{equation*}
		p_{U}(u) =
		\frac{1}{\Delta}
		\begin{cases}
			\Pr{X \text{ is discrete}} \Pr{Y \text{ is continuous}} & u=\mathrm{dc} \\
			\Pr{X \text{ is continuous}} \Pr{Y \text{ is discrete}} & u=\mathrm{cd} \\
			\Pr{X \text{ is continuous}} \Pr{Y \text{ is continuous}} & u=\mathrm{cc}
		\end{cases},
	\end{equation*}
	and
	\begin{equation*}
		p_{Z_c | U}(x|u) =
		\begin{cases}
			p_{X_D+Y_c} & u=\mathrm{dc} \\
			p_{X_c+Y_D} & u=\mathrm{cd} \\
			p_{X_c+Y_c} & u=\mathrm{cc}
		\end{cases},
	\end{equation*}
	where
	\begin{align*}
		\Delta := &
		\Pr{X \text{ is discrete}} \Pr{Y \text{ is continuous}} \\
		&+\Pr{X \text{ is continuous}} \Pr{Y \text{ is discrete}} \\
		&+\Pr{X \text{ is continuous}} \Pr{Y \text{ is continuous}}.
	\end{align*}
	Thus, we have that
	\begin{align*}
		\h{Z_c}
		\geq & \h{Z_c \Big| Q} \\
		= & p_U(\mathrm{cd}) \h{X_c+Y_D}
		+ p_U(\mathrm{dc}) \h{X_D+Y_c}
		+ p_U(\mathrm{cc}) \h{X_c+Y_c} \\
		\geq & p_U(\mathrm{cd}) \h{X_c+Y_D | Y_D}
		+p_U(\mathrm{dc}) \h{X_D+Y_c | X_D}
		+p_U(\mathrm{cc}) \h{X_c+Y_c | Y_c} \\
		\geq & \min(\h{X_c}, \h{Y_c}).
	\end{align*}
	Hence, the lemma is proved.
\qed

%%%%%%%%%%%%%%%%%%%%%%
% Proof of Lemma: Addition of countinuous RV
\subsection{Proof of Lemma \ref{lmm:AC+RV}} \label{subsec:prf:lmm:AC+RV}
From the Fubini's theorem we obtain that
\begin{equation*}
	\Pr{Z \leq z}
	=\Pr{X \leq z-Y}
	=\E{F_X(z-Y)},
\end{equation*}
where $F_X(x) := \Pr{X\leq x}$ is the cdf of $X$.
In order to prove the first part of the lemma, we only need to prove that
\begin{equation} \label{eqn:EFX=intEpX}
	\E{F_X(z-Y)}
	=\int_{-\infty}^z {\E{p_X(z-Y)} \ud z}.
\end{equation}
The above integral can be written as the following limit:
\begin{equation*}
	\int_{-\infty}^z {\E{p_X(z-Y)} \ud z}
	=\lim_{\ell\to\infty}{\int_{-\ell}^z {\E{p_X(z-Y)} \ud z}}.
\end{equation*}
Note that, from Fubini's theorem, for every finite $\ell$, we have that
\begin{align*}
	\int_{-\ell}^z {\E{p_X(z-Y)} \ud z}
	=& \E{\int_{-\ell}^z{p_X(z-Y) \ud z}} \\
	=& \E{F_X(z-Y) - F_X(-\ell-Y)} \\
	=& \E{F_X(z-Y)} - \E{F_X(-\ell-Y)}.
\end{align*}
Hence, to prove \eqref{eqn:EFX=intEpX}, it is sufficient to show that
\begin{equation*}
	\lim_{\ell\to\infty}{\E{F_X(-\ell-Y)}} = 0.
\end{equation*}
In order to do so, note that there exists $\ell_X$, $\ell_Y$ such that
\begin{equation*}
	F_X(\ell_X) \leq \epsilon/2,
	\qquad F_Y(\ell_Y) \leq \epsilon/2.
\end{equation*}
Thus, we can write
\begin{align*}
	\E{F_X(-\ell-Y)}
	=& \Pr{Y\leq\ell_Y} \E{F_X(-\ell-Y)|Y\leq\ell_Y} \\
	&+\Pr{Y>\ell_Y} \E{F_X(-\ell-Y)|Y>\ell_Y} \\
	\leq& \Pr{Y\leq\ell_Y} \times 1
	+ 1 \times \E{F_X{-\ell-\ell_Y}}
	\leq \epsilon.
\end{align*}
Therefore, the the lemma is proved.
\qed

%%%%%%%%%%%%%%%%%%%%%%
% Proof of Lemma: Maximum zeta in continuous
\subsection{Proof of Lemma \ref{lmm:zeta<EntP}} \label{subsec:prf:lmm:zeta<EntP}
	Note that by the entropy power inequality (EPI) \cite[Theorem 17.7.3]{Cover06}, we have that for $n$ iid scalar continuous random variables $X_1,\cdots,X_n$
	\begin{align*}
		\ue^{2\h{X_1+\cdots+X_n}}
		\geq n \ue^{2\h{X_1}}
		\Longrightarrow &
		\h{X_1}
		\leq \h{X_1+\cdots+X_n}+\log\frac{1}{\sqrt{n}}.
	\end{align*}
	Therefore, from Lemma \ref{lmm:Seperatedindependent} we have that
	\begin{equation*}
		\h{X_1^{(n)}}
		\leq \h{X_1^{(n)}+\cdots+X_n^{(n)}}
		+\log\frac{1}{\sqrt{n}}
		=\h{X_0}+\log{\frac{1}{\sqrt{n}}},
	\end{equation*}
	where $X_1^{(n)}$ was defined in Definition \ref{def.TimeQuant}.
	According to Theorem \ref{thm:GenerelEntropyDim}, $\zeta(n) = \h{X_1^{(n)}}$ satisfies \eqref{eqn:H/k-lnm-z=0}.
	This will complete the proof since based on Theorem \ref{thm:Uniqueness}, for any $\zeta'(n)$ satisfying \eqref{eqn:H/k-lnm-z=0}, we have that
	\begin{equation*}
		\lim_{n\to\infty}\left|\zeta(n)-\zeta'(n)\right|
		=0.
	\end{equation*}
	Hence,
	\begin{equation*}
		\zeta'(n) - \log\frac{1}{\sqrt{n}}
		=\zeta'(n) - \zeta(n)
		+\zeta(n) - \log\frac{1}{\sqrt{n}}
		\leq \left|\zeta'(n) - \zeta(n)\right| + \h{X_0}.
	\end{equation*}
	Thus, the lemma is proved.
\qed

%%%%%%%
% Conclusion
%%%%%%%
\section{Conclusion}
In this paper, a definition of quantization entropy for random processes based on quantization in the time and amplitude domains was given. The criterion was applied to a wide class of white noise processes, including stable and impulsive Poisson innovation processes.  It was shown that the stable has a higher growth rate of entropy compared to the impulsive Poisson process.
In our study, we assumed that the amplitude quantization steps $1/m$ is shrinking sufficiently fast with respect to  time quantization steps $1/n$, \emph{i.e.,}  $m$ is larger than $m(n)$ for some function $m:\bN\to\bN$. As a future work, it would be interesting to look at cases where $m$ is restricted grow slowly with $n$. Characterization of the entropy for other stochastic processes is also left as a future work.

%%%%%%%
% References
%%%%%%%
\bibliographystyle{IEEEtran}
\bibliography{IEEEabrv,mybib}

\begin{thebibliography}{10}
\providecommand{\url}[1]{#1}
\csname url@samestyle\endcsname
\providecommand{\newblock}{\relax}
\providecommand{\bibinfo}[2]{#2}
\providecommand{\BIBentrySTDinterwordspacing}{\spaceskip=0pt\relax}
\providecommand{\BIBentryALTinterwordstretchfactor}{4}
\providecommand{\BIBentryALTinterwordspacing}{\spaceskip=\fontdimen2\font plus
\BIBentryALTinterwordstretchfactor\fontdimen3\font minus
  \fontdimen4\font\relax}
\providecommand{\BIBforeignlanguage}[2]{{%
\expandafter\ifx\csname l@#1\endcsname\relax
\typeout{** WARNING: IEEEtran.bst: No hyphenation pattern has been}%
\typeout{** loaded for the language `#1'. Using the pattern for}%
\typeout{** the default language instead.}%
\else
\language=\csname l@#1\endcsname
\fi
#2}}
\providecommand{\BIBdecl}{\relax}
\BIBdecl

\bibitem{Renyi59}
A.~R{\'e}nyi, ``On the dimension and entropy of probability distributions,''
  \emph{Acta Mathematica Academiae Scientiarum Hungarica}, vol.~10, no. 1-2,
  pp. 193--215, Mar. 1959.

\bibitem{Verdu10}
Y.~Wu and S.~Verd{\'u}, ``R{\'e}nyi information dimension: Fundamental limits
  of almost lossless analog compression,'' \emph{{IEEE} Trans. Inf. Theory},
  vol.~56, no.~8, pp. 3721--3748, Aug. 2010.

\bibitem{Jalali2017}
S.~Jalali and H.~V. Poor, ``Universal compressed sensing for almost lossless
  recovery,'' \emph{IEEE Trans. Inform. Theo.}, vol.~63, no.~5, pp. 2933--2953,
  May 2017.

\bibitem{Donoho2006}
D.~L. Donoho, ``Compressed sensing,'' \emph{{IEEE} Trans. Inf. Theory},
  vol.~52, no.~4, pp. 1289--1306, 2006.

\bibitem{Candes2006}
E.~Cand{\`e}s, ``Compressive sampling,'' in \emph{Proc. International Congress
  of Mathematicians}, vol.~3, 2006, pp. 1433--1452.

\bibitem{Baraniuk2010}
R.~Baraniuk, V.~Cevher, and M.~Wakin, ``Low-dimensional models for
  dimensionality reduction and signal recovery: A geometric perspective,''
  \emph{Proc. of the IEEE}, vol.~98, no.~6, pp. 959--971, 2010.

\bibitem{DeVore1998}
R.~A. DeVore, ``Nonlinear approximation,'' \emph{Acta numerica}, vol.~7, pp.
  51--150, 1998.

\bibitem{Unser2014_P1}
M.~Unser, P.~Tafti, and Q.~Sun, ``A unified formulation of {G}aussian versus
  sparse stochastic processes---{P}art {I}: {C}ontinuous-domain theory,''
  \emph{{IEEE} Trans. Inf. Theory}, vol.~60, no.~3, pp. 1945--1962, 2014.

\bibitem{Unser2014_P2}
M.~Unser, P.~Tafti, A.~Amini, and H.~Kirshner, ``A unified formulation of
  {G}aussian versus sparse stochastic processes---{P}art {II}:{D}iscrete-domain
  theory,'' \emph{{IEEE} Trans. Inf. Theory}, vol.~60, no.~5, pp. 3036--3051,
  2014.

\bibitem{Cohen2009}
A.~Cohen, W.~Dahmen, and R.~DeVore, ``Compressed sensing and best $k$-term
  approximation,'' \emph{J. Amer. Math. Soc.}, vol.~22, no.~1, pp. 211--231,
  2009.

\bibitem{Cevher2009}
V.~Cevher, ``Learning with compressible priors,'' in \emph{Proc. Adv. in Neural
  Inform. Process. Sys. (NIPS)}, 2009, pp. 261--269.

\bibitem{Amini11}
A.~Amini, M.~Unser, and F.~Marvasti, ``Compressibility of deterministic and
  random infinite sequences,'' \emph{{IEEE} Trans. Signal Process.}, vol.~59,
  no.~11, pp. 5193--5201, Nov. 2011.

\bibitem{Silva2012}
J.~F. Silva and E.~Pavez, ``Compressibility of infinite sequences and its
  interplay with compressed sensing recovery,'' in \emph{Proc. Sig. \& Inform.
  Proc. Assoc. Annual Summit and Conf. (APSIPA ASC), 2012 Asia-Pacific}, Dec.
  2012.

\bibitem{Gribonval2012}
R.~Gribonval, V.~Cevher, and M.~E. Davies, ``Compressible distributions for
  high-dimensional statistics,'' \emph{{IEEE} Trans. Inf. Theory}, vol.~58,
  no.~8, pp. 5016--5034, 2012.

\bibitem{Silva15}
J.~F. Silva and M.~S. Derpich, ``On the characterization of
  $\ell_p$-compressible ergodic sequences,'' \emph{{IEEE} Trans. Signal
  Process.}, vol.~63, no.~11, pp. 2915--2928, Jun. 2015.

\bibitem{Unser14}
M.~Unser and P.~D. Tafti, \emph{An introduction to sparse stochastic
  processes}.\hskip 1em plus 0.5em minus 0.4em\relax Cambridge University
  Press, 2014.

\bibitem{Amini2012}
A.~Amini, U.~S. Kamilov, and M.~Unser, ``The analog formulation of sparsity
  implies infinite divisibility and rules out {B}ernoulli-{G}aussian priors,''
  in \emph{Inform. Theo. Workshop (ITW), 2012 IEEE}.\hskip 1em plus 0.5em minus
  0.4em\relax IEEE, 2012, pp. 682--686.

\bibitem{Amini14}
A.~Amini and M.~Unser, ``Sparsity and infinite divisibility,'' \emph{{IEEE}
  Trans. Inf. Theory}, vol.~60, no.~4, pp. 2346--2358, Apr. 2014.

\bibitem{Fageot2017}
J.~Fageot, M.~Unser, and J.~P. Ward, ``The n-term approximation of periodic
  generalized {L}\'evy processes,'' \emph{arXiv preprint arXiv:1702.03335},
  2017.

\bibitem{Lorentz66}
G.~Lorentz, ``Metric entropy and approximation,'' \emph{Bulletin of the
  American Mathematical Society}, vol.~72, no.~6, pp. 903--937, 1966.

\bibitem{Ponser73}
E.~C. Posner and E.~R. Rodemich, ``Epsilon entropy of stochastic processes with
  continuous paths,'' \emph{The Annals of Probability}, pp. 674--689, Aug.
  1973.

\bibitem{Ihara93}
S.~Ihara, \emph{Information theory for continuous systems}.\hskip 1em plus
  0.5em minus 0.4em\relax Singapore: World Scientific, 1993.

\bibitem{Nair06}
C.~Nair, B.~Prabhakar, and D.~Shah, ``On entropy for mixtures of discrete and
  continuous variables,'' \emph{arXiv preprint cs/0607075}, 2006.

\bibitem{Cover06}
T.~M. Cover and J.~A. Thomas, \emph{Elements of information theory},
  2nd~ed.\hskip 1em plus 0.5em minus 0.4em\relax New York: John Wiley {\&}
  Sons, 2006.

\bibitem{OurFirstPaper}
H.~Ghourchian, A.~Gohari, and A.~Amini, ``Existence and continuity of
  differential entropy for a class of distributions,'' \emph{IEEE Comm.
  Letters}, vol.~21, no.~7, pp. 1469--1472, 2017.

\bibitem{Samor94}
G.~Samoradnitsky and M.~S. Taqqu, \emph{Stable non-Gaussian random processes:
  stochastic models with infinite variance}.\hskip 1em plus 0.5em minus
  0.4em\relax New York: Chapman {\&} Hall, 1994.

\bibitem{Sato99}
\BIBentryALTinterwordspacing
K.~Sato, \emph{L{\'e}vy Processes and Infinitely Divisible Distributions}, ser.
  Cambridge Studies in Advanced Mathematics.\hskip 1em plus 0.5em minus
  0.4em\relax Cambridge University Press, 1999. [Online]. Available:
  \url{https://books.google.com/books?id=tbZPLquJjSoC}
\BIBentrySTDinterwordspacing

\bibitem{Tucker62}
H.~Tucker, ``Absolute continuity of infinitely divisible distributions,''
  \emph{Pacific Journal of Mathematics}, vol.~12, no.~3, pp. 1125--1129, 1962.

\bibitem{Blum59}
J.~Blum and M.~Rosenblatt, ``On the structure of infinitely divisible
  distributions,'' \emph{Pacific Journal of Mathematics}, vol.~9, no.~1, pp.
  1--7, 1959.

\bibitem{Rajput1989}
\BIBentryALTinterwordspacing
B.~S. Rajput and J.~Rosinski, ``Spectral representations of infinitely
  divisible processes,'' \emph{Probability Theory and Related Fields}, vol.~82,
  no.~3, pp. 451--487, Aug. 1989. [Online]. Available:
  \url{https://doi.org/10.1007/BF00339998}
\BIBentrySTDinterwordspacing

\end{thebibliography}

%%%%%%%
% Appendix
%%%%%%%
%\appendices
\makeatother
\end{document}